\documentclass[reqno,a4paper,11pt]{article}
\pdfoutput=1

\usepackage{xcolor}
\usepackage{graphicx}
\usepackage[textwidth = 430 pt, textheight = 630 pt]{geometry}

\definecolor{MyDarkBlue}{rgb}{0.15,0.25,0.45}
\usepackage{epsfig,rotating}
\usepackage{mathtools}
\usepackage{amsmath,amssymb}
\usepackage{amsfonts}
\usepackage{mathrsfs}
\usepackage{bbm}

\usepackage[normalem]{ulem}

\usepackage{latexsym}
\usepackage{mathabx} 
\usepackage{enumitem}

\usepackage{amsthm}
\usepackage[all,knot]{xy}
\xyoption{arc}

\usepackage[utf8x]{inputenc}

\usepackage{hyperref}
\hypersetup{
    hypertexnames=false,
    colorlinks=true,
    citecolor=MyDarkBlue,
    linkcolor=MyDarkBlue,
    urlcolor=MyDarkBlue,
    pdfauthor={Leron Borsten, Hyungrok Kim, Christian S\"amann},
    pdftitle={E2L-infinity--algebras, Generalized Geometry, and Tensor Hierarchies},    
    pdfsubject={hep-th math-ph},
    breaklinks=true
}
\usepackage[nameinlink,noabbrev]{cleveref}

\usepackage{tikz}
\usetikzlibrary{matrix,cd,arrows}

\let\fn\footnote
\renewcommand{\footnote}[1]{\linespread{1.1}\fn{#1}\linespread{1.29}}

\makeatletter\renewcommand{\section}{\@startsection
    {section}{1}{\z@}{-3.5ex plus -1ex minus
        -.2ex}{2.3ex plus .2ex}{\bf }}
\makeatletter\renewcommand{\subsection}{\@startsection{subsection}{2}{\z@}{-3.25ex
        plus -1ex minus
        -.2ex}{1.5ex plus .2ex}{\bf }}
\makeatletter\renewcommand{\subsubsection}{\@startsection{subsubsection}{3}{-2.45ex}{-3.25ex
        plus -1ex minus -.2ex}{1.5ex plus .2ex}{\it }}
\renewcommand{\thesection}{\arabic{section}}
\renewcommand{\thesubsection}{\arabic{section}.\arabic{subsection}}
\renewcommand{\@seccntformat}[1]{\@nameuse{the#1}.~~}

\renewcommand{\theequation}{\thesection.\arabic{equation}}
\makeatletter \@addtoreset{equation}{section}
\def\Ddots{\mathinner{\mkern1mu\raise\p@
        \vbox{\kern7\p@\hbox{.}}\mkern2mu
        \raise4\p@\hbox{.}\mkern2mu\raise7\p@\hbox{.}\mkern1mu}}
\DeclareRobustCommand{\cev}[1]{%
    {\mathpalette\do@cev{#1}}%
}
\newcommand{\do@cev}[2]{%
    \vbox{\offinterlineskip
        \sbox\z@{$\m@th#1 x$}%
        \ialign{##\cr
            \hidewidth\reflectbox{$\m@th#1\vec{}\mkern4mu$}\hidewidth\cr
            \noalign{\kern-\ht\z@}
            $\m@th#1#2$\cr
        }%
    }%
}
\providecommand*{\shuffle}{%
    \mathbin{\mathpalette\shuffle@{}}%
}
\newcommand*{\shuffle@}[2]{%
    \sbox0{$#1\vcenter{}$}%
    \kern .15\ht0 
    \rlap{\vrule height .25\ht0 depth 0pt width 2.5\ht0}%
    \raise.1\ht0\hbox to 2.5\ht0{%
        \vrule height 1.75\ht0 depth -.1\ht0 width .17\ht0 %
        \hfill
        \vrule height 1.75\ht0 depth -.1\ht0 width .17\ht0 %
        \hfill
        \vrule height 1.75\ht0 depth -.1\ht0 width .17\ht0 %
    }%
    \kern .15\ht0 
}

\usepackage[toc,page]{appendix}

\newtheorem{thm}{Theorem}[section]
\renewcommand{\thethm}{\thesection.\arabic{thm}}
\newtheorem{lemma}[thm]{Lemma}
\newtheorem{definition}[thm]{Definition}
\newtheorem{theorem}[thm]{Theorem}
\newtheorem{proposition}[thm]{Proposition}
\newtheorem{corollary}[thm]{Corollary}

\renewcommand{\appendices}{
    \section*{Appendix}\label{appendices}\setcounter{subsection}{0}
    \addcontentsline{toc}{section}{Appendix}
    \setcounter{equation}{0}
    \crefalias{subsection}{appendix}
    \makeatletter
    \renewcommand{\theequation}{\Alph{subsection}.\arabic{equation}}
    \renewcommand{\thesubsection}{\Alph{subsection}}
    \renewcommand{\thethm}{\Alph{subsection}.\arabic{thm}}
    \@addtoreset{equation}{subsection}
    \@addtoreset{thm}{subsection}
    \makeatother
}

\makeatother

\let\oldref\ref
\AtBeginDocument{\renewcommand{\ref}[1]{\cref{#1}}}
\AtBeginDocument{\renewcommand{\eqref}[1]{{\rm (\oldref{#1})}}}

\newcommand{\makecommand}[3]{%
    \foreach \i in #3 {%
        \expandafter\xdef\csname #1\i\endcsname{\noexpand#2{{\unexpanded\expandafter{\i}}}}%
    }%
}

\newcommand{\latinalphabet}{A,a,B,b,C,c,d,D,E,e,F,f,G,g,H,h,I,i,J,j,K,k,L,l,M,m,N,n,O,o,P,p,Q,q,R,r,S,s,T,t,U,u,V,v,W,w,X,x,Y,y,Z,z}
\makecommand{I}{\mathbbm}{\latinalphabet}
\makecommand{bf}{\mathbf}{\latinalphabet}
\makecommand{bm}{\bm}{\latinalphabet}
\makecommand{ca}{\mathcal}{\latinalphabet}
\makecommand{fr}{\mathfrak}{\latinalphabet}
\makecommand{rm}{\mathrm}{\latinalphabet}
\makecommand{sf}{\mathsf}{\latinalphabet}
\makecommand{sc}{\mathscr}{\latinalphabet}

\makecommand{fr}{\mathfrak}{{der,gl,sl,so,osp,su,sp,spin,inn,string,Poisson,aut}}
\makecommand{sf}{\mathsf}{{id,Spin,String,Lie,Hom,GL,SL,SO,SU,Sp,USp,Aut,Str,Diff,HSymp,jac,alt,sym,Sym,CE,W,T,TE,End,ChEL,ZE,MC,pr}}

\makecommand{rm}{\mathrm}{{Ham,sk,lp,wk,Im}}

\newcommand{\rep}[1]{{\mathbf{#1}}}

\newcommand{\opCom}{\mathcal{\caC}om}
\newcommand{\opAss}{\mathcal{\caA}ss}

\newcommand{\opLie}{\mathcal{\caL}ie}
\newcommand{\ophLie}{h\mathcal{\caL}ie}

\newcommand{\opEilh}{\mathcal{\caE}ilh}

\newcommand{\opLeib}{\mathcal{\caL}eib}
\newcommand{\opZinb}{\mathcal{\caZ}inb}
\newcommand{\opELinfty}{\mathcal{\caE}L_\infty}

\newcommand{\eps}{\varepsilon}
\newcommand{\ewith}{~~~\mbox{with}~~~}
\newcommand{\eand}{~~~\mbox{and}~~~}

\newcommand{\sst}[1]{{\scriptscriptstyle #1}}
\def\0{{\sst{(0)}}}
\def\1{{\sst{(1)}}}
\def\2{{\sst{(2)}}}
\def\3{{\sst{(3)}}}
\def\4{{\sst{(4)}}}
\def\5{{\sst{(5)}}}
\def\6{{\sst{(6)}}}
\def\7{{\sst{(7)}}}

\DeclareMathOperator{\tr}{tr}

\def\slasha#1{\setbox0=\hbox{$#1$}#1\hskip-\wd0\hbox to\wd0{\hss\sl/\/\hss}}

\def\periodb#1{\setbox0=\hbox{$#1$}#1\hskip-\wd0\hbox to\wd0{-}}

\newcommand{\dpar}{\partial}
\newcommand{\rmder}[2][]{%
    \ifthenelse{\equal{#1}{}}{%
        \frac{\mathrm{d}}{\mathrm{d} #2}%
    }{%
        \frac{\mathrm{d} #1}{\mathrm{d} #2}%
    }%
}
\newcommand{\parder}[2][]{%
    \ifthenelse{\equal{#1}{}}{%
        \frac{\partial}{\partial #2}%
    }{%
        \frac{\partial #1}{\partial #2}%
    }%
}
\newcommand{\delder}[2][]{%
    \ifthenelse{\equal{#1}{}}{%
        \frac{\delta}{\delta #2}%
    }{%
        \frac{\delta #1}{\delta #2}%
    }%
}

\newcommand{\der}[1]{\frac{\partial}{\partial #1}}   		

\newcommand{\acton}{\vartriangleright}     			
\newcommand{\comment}[1]{}     				
\def\tyng(#1){\hbox{\tiny$\yng(#1)$}}			
\def\tyoung(#1){\hbox{\tiny$\young(#1)$}}			

\newcommand{\beq}{\begin{eqnarray}}
    \newcommand{\eeq}{\end{eqnarray}}

\let\oldbigotimes\bigotimes
\renewcommand\bigotimes{\oldbigotimes\nolimits}

\newcommand{\antishriek}{\text{\raisebox{\depth}{\textexclamdown}}}

\begin{document}
    
    \begin{titlepage}
        \begin{flushright}
            EMPG--21--07
        \end{flushright}
        \vskip2.0cm
        \begin{center}
            {\LARGE \bf\mathversion{bold}
                $E_2L_\infty$-algebras, Generalized Geometry,\\[0.4cm]   
                and Tensor Hierarchies
            }
            \vskip1.5cm
            {\Large Leron Borsten$^{1,2}$, Hyungrok Kim$^1$, and Christian Saemann$^3$}
            \setcounter{footnote}{0}
            \renewcommand{\thefootnote}{\arabic{thefootnote}}
            \vskip1cm
            {\em ${}^1$ Department of Physics, Astronomy and Mathematics\\
                University of Hertfordshire, Hatfield, Hertfordshire AL10 9AB, U.K.}\\[0.2cm]
            {\em ${}^2$ Blackett Laboratory, Imperial College London\\
                London SW7 2AZ, U.K.}\\[0.2cm]
            {\em ${}^3$ Maxwell Institute for Mathematical Sciences\\
                Department of Mathematics, Heriot--Watt University\\
                Colin Maclaurin Building, Riccarton, Edinburgh EH14 4AS, U.K.}\\[0.5cm]
            {Email: {\ttfamily l.borsten@herts.ac.uk~,~h.kim2@herts.ac.uk~,~c.saemann@hw.ac.uk}}
        \end{center}
        \vskip1.0cm
        \begin{center}
            {\bf Abstract}
        \end{center}
        \begin{quote}
            We define a generalized form of $L_\infty$-algebras called $E_2L_\infty$-algebras. As we show, these provide the natural algebraic framework for generalized geometry and the symmetries of double field theory as well as the gauge algebras arising in the tensor hierarchies of gauged supergravity. Our perspective shows that the kinematical data of the tensor hierarchy is an adjusted higher gauge theory, which is important for developing finite gauge transformations as well as non-local descriptions. Mathematically, $E_2L_\infty$-algebras shed some light on Loday's problem of integrating Leibniz algebras.
        \end{quote}

    \end{titlepage}
    
    \tableofcontents
    
    \section{Introduction and results}
    
    In this paper, we define a new class of homotopy algebras called $E_2L_\infty$-algebras, which can be regarded as a weaker yet quasi-isomorphic form of $L_\infty$-algebras. These generalize previous definitions of weaker forms of categorified Lie algebras~\cite{Roytenberg:0712.3461,Squires:2011aa}, from where we also borrowed the nomenclature.
    
    Recall that in~\cite{Baez:2003aa}, Baez--Crans introduced semistrict Lie 2-algebras: linear categories equipped with a strictly antisymmetric  bilinear functor,  the categorified Lie bracket, that is only required to satisfy the Jacobi identity up to a coherent trilinear natural transformation, the \emph{Jacobiator}.  In~\cite{Roytenberg:0712.3461}  semistrict Lie 2-algebras were fully categorified  to \emph{weak} Lie 2-algebras by also relaxing  antisymmetry  to hold only up to a coherent natural transformation, the \emph{alternator}. 
    
    By passing to its normalized chain complex, we transition from the categorical description containing many redundancies to a computationally more convenient description in terms of differential graded algebras. In particular, a semistrict Lie 2-algebra is seen to be equivalent to a 2-term $L_\infty$-algebra~\cite{Baez:2003aa}, i.e.~an $L_\infty$-algebra with underlying graded vector space concentrated in degrees $-1$ and $0$. Analogously, by passing to its normalized chain complex, any weak Lie 2-algebra is seen to be equivalent to a  2-term $E_2L_\infty$-algebra in the sense of~\cite{Roytenberg:0712.3461}, where the letter $E$ was added to indicate that `everything' is relaxed up to homotopy. 
    
    Here, we extend Roytenberg's construction to $\IZ$-graded differential complexes. We show that our generalization describes important data in many areas: from generalized geometry and double field theory to multisymplectic geometry and the tensor hierarchies of supergravity.

    \subsection{Four questions}
    
    Our generalized homotopy Lie algebras arise naturally in many contexts, and there are four questions that they can answer or, at least, suggest an answer to. 
    
    \subsection*{What is the algebraic structure underlying adjusted higher curvatures?}
    
    Higher gauge theories, i.e.~gauge theories of higher form potentials corresponding to connections on higher or categorified principal bundles, may underlie a future non-abelian M5-brane model~\cite{Saemann:2017zpd,Saemann:2019leg}, but they certainly arise in heterotic supergravity as well as the tensor hierarchies of gauged supergravities, see e.g.~\cite{Samtleben:2008pe,Trigiante:2016mnt} for reviews, and of double and exceptional field theories~\cite{Aldazabal:2013via}. A proper understanding of the relevant mathematical structures is evidently important: global constructions are only possible with the correct notion of corresponding higher principal bundles.
    
    However, already the appropriate definition of the notion of curvature in a higher gauge theory is not straightforward for non-flat theories. Using categorification, there is a straightforward definition of higher curvature forms, which was also found in the first mathematical papers on non-abelian gerbes, \cite{Breen:math0106083} and~\cite{Aschieri:2003mw}. Unfortunately, these curvatures are too restrictive for non-flat connections, cf.~\cite{Saemann:2019dsl}. A consistency condition known as fake flatness must be imposed for the full cocycle data of non-abelian gerbes to glue together consistently (and for gauge or BRST transformations to close). This condition requires all curvature components except for the one of highest form degree to vanish, which in turn implies that all of the components of the gauge potentials except for the one of highest form degree can locally be gauged away. This is readily seen in the strict case~\cite{Saemann:2019dsl}; see also~\cite{Gastel:2018joi} for a detailed analytical proof.
    
    Physicists have long known that coupling abelian 2-form potentials to non-abelian 1-form potentials is best done using a different expression for the $3$-form component of the curvature~\cite{Bergshoeff:1981um,Chapline:1982ww}. This modification involves adding terms proportional to curvature expressions to the flat curvature, which were not visible to the homotopy Maurer--Cartan theory. These necessary modifications were implemented for the string and five-brane structures in~\cite{Sati:2009ic,Sati:2008eg} by performing a coordinate change on the Weil algebra of the gauge $L_\infty$-algebra\footnote{At the purely algebraic level, this ensures that one can define invariant polynomials in a way that is compatible with quasi-isomorphisms of $L_\infty$-algebras, see also~\cite{Saemann:2019dsl}.}, which induces the necessary change in the definition of curvatures. Considering the curvatures arising in the tensor hierarchy, one encounters similar modifications. 
    
    The modifications required for the definitions of non-flat higher curvatures (and which lead to a closed BRST complex) were dubbed~{\em adjustments} in~\cite{Saemann:2019dsl}, and a consistent higher parallel transport for adjusted higher curvatures was defined in~\cite{Kim:2019owc}. An adjustment now requires additional higher products, which are not visible in the original higher gauge algebra. It turns out that they are components of higher products in $E_2L_\infty$-algebras that are quasi-isomorphic to the higher gauge algebra.
    
    In particular, we show in \ref{thm:firmly_adjust} that there exists a particular class of $L_\infty$-algebras that come with a natural adjustment encoded in an $E_2L_\infty$-algebra. This class is precisely the one arising in the tensor hierarchies of gauged supergravity. The latter are thus adjusted higher gauge algebras, employing $E_2L_\infty$-algebras in their construction.
    
    \subsection*{What is the full algebra underlying generalized geometry?}
    
    Generalized geometry in its simplest form is described by an exact Courant algebroid, which captures the infinitesimal gauge symmetries of Einstein--Hilbert gravity coupled to a Kalb--Ramond 2-form potential. Roytenberg has shown that the Courant algebroid is best regarded as a symplectic Lie 2-algebroid~\cite{Roytenberg:0203110}, and more general forms of the generalized tangent bundle can similarly be encoded in higher symplectic Lie algebroids. Dually, they are described by their Chevalley--Eilenberg algebras, which are certain differential graded Poisson algebras. The latter, in turn, give rise to associated $L_\infty$-algebras via a derived bracket construction~\cite{Fiorenza:0601312,Getzler:1010.5859}, describing generalizations of the above mentioned gauge symmetries. In particular, the binary product of the $L_\infty$-algebra for the Courant algebroid is simply the Courant bracket. This picture has extensions to double field theory, cf.~\cite{Deser:2016qkw}.
    
    It is somewhat unsatisfactory that the derived bracket construction only reproduces the Courant bracket, while the Dorfman bracket, which antisymmetrizes to the Courant bracket, has to be constructed ``by hand''. It turns out that the derived bracket construction has a refinement, in which the differential graded Poisson algebra first gives rise to an $E_2L_\infty$-algebra, which can then be antisymmetrized to the $L_\infty$-algebra obtained in~\cite{Getzler:1010.5859}. This construction of an $E_2L_\infty$-algebra from a differential graded Lie algebra also underlies the gauge algebras appearing in the tensor hierarchies of gauged supergravity. As mentioned above, the refinement here provides the additional structure constants required for an adjustment.
    
    \subsection*{What is the higher Poisson algebra arising in multisymplectic geometry?}
    
    Multisymplectic forms are higher, non-degenerate differential forms generalizing symplectic forms. Just as symplectic forms define a Poisson algebra structure on the algebra of functions, multisymplectic forms define higher analogues of Poisson algebras involving functions and differential forms~\cite{Baez:2008bu}. As known since the work of Rogers~\cite{Rogers:2010sc,Rogers:2010nw}, these $L_\infty$-algebras can be embedded into the above mentioned $L_\infty$-algebras arising in the case of higher symplectic Lie-algebroids. Therefore it is not surprising that we again have a refinement which constructs an $E_2L_\infty$-algebra from the multisymplectic form. For multisymplectic 3-forms, this had already been observed in~\cite{Baez:2008bu}. With a general definition of $E_2L_\infty$-algebras, we can now generalize the statement to arbitrary multisymplectic manifolds.
        
    \subsection*{What is the object integrating Leibniz algebras?}
    
    In~\cite{RCP25_1993__44__127_0}, Loday posed the ``coquecigrue problem'' of generalizing Lie's third theorem to Leibniz algebras. A conventional answer to this problem is given in terms of Lie racks, cf.~\cite{Kinyon:0403509}. However, there is an even simpler answer suggested by $E_2L_\infty$-algebras.
    
    It appears to be a common phenomenon that some forms of integration are only possible after extending to the right cohomology in higher structures. For example, central extensions of Lie algebras do not, in general, integrate to central extensions of Lie groups. This obstruction, however, may be overcome by integrating to a Lie 2-group~\cite{Zhu:1204.5583}. Similarly, the general integration of Lie algebroids only becomes possible after they are regarded as Lie $\infty$-algebroids and integrated to Lie $\infty$-groupoids, cf.~\cite{Severa:1506.04898}.
    
    We observe that any Leibniz algebra canonically gives rise to an $\ophLie_2$-algebra, a weak Lie 2-algebra, cf.~\cite[Example~2.22]{Roytenberg:0712.3461} as well as \ref{prop:Leib_is_hemistrict_ELinfty}. From the standpoint of higher Lie theory, there has to be a natural extension of the usual integration theory of $L_\infty$-algebras, cf.~\cite{Getzler:0404003,Henriques:2006aa}, that allows for an integration of this weak Lie 2-algebra.
    
    The higher version of the ``coquecigrue problem,'' i.e.~an integration of $\opLeib_\infty$-algebras, would then be similarly resolved: each $\opLeib_\infty$-algebra needs to be promoted by additional (higher) alternators to an $EL_\infty$-algebra, that may have to be more general than the $E_2L_\infty$-algebras constructed here. These then should be integrable, in principle, by general higher Lie theory.
    
    \subsection{Conclusions and outlook}
    
    Besides giving the general definition for $E_2L_\infty$-algebra in a fashion that can be readily used for explicit computations, we also develop the general structure theory:
    \begin{itemize}
        \item[$\Diamond$] The key to most of our discussion is the notion of $\ophLie_2$-algebras (\ref{def:hLie2}). These are differential graded Leibniz algebras in which the Leibniz bracket fails to be graded antisymmetric up to a homotopy given by an alternator.
        \item[$\Diamond$] Koszul dual to the operad $\ophLie_2$ is the operad $\opEilh_2$ (\ref{def:operad_Eilh2}), and we can use semifree $\opEilh_2$-algebras to define the homotopy algebras of $\ophLie_2$-algebras in \ref{def:E2Linfty}, which we call $E_2L_\infty$-algebras. 
        \item[$\Diamond$] $E_2L_\infty$-algebras come with a good notion of homotopy transfer, see \ref{thm:homotopy_transfer}, and, correspondingly, with a minimal model theorem, see \ref{thm:minimal_model}.
        \item[$\Diamond$] The category of $L_\infty$-algebras embeds into the category of $E_2L_\infty$-algebras (\ref{thm:el_infty_contain_l_infty} and \ref{thm:lift_L_infty_morphism}).
        \item[$\Diamond$] An $E_2L_\infty$-algebra is quasi-isomorphic to an $L_\infty$-algebra, regarded as an $E_2L_\infty$-algebra, by \ref{thm:quasi-iso_of_antisym}. Therefore, any $E_2L_\infty$-algebra is also quasi-isomorphic to a differential graded Lie algebra.
        \item[$\Diamond$] Any differential graded Lie algebra gives naturally rise to an $\ophLie_2$-algebra by \ref{thm:ophLie_from_dgLA}. In particular, there is a hemistrict $E_2L_\infty$-algebra on the shifted and truncated differential complex.
        \item[$\Diamond$] All of the above can be made explicit in terms of multilinear maps, at least order by order, and we give explicit formulas that should prove useful in future applications.
    \end{itemize}
    
    We can then show that given an $\ophLie_2$-algebra originating from a differential graded Lie algebra, adjusted notions of the curvatures of higher gauge theory are naturally found. These adjusted curvatures are precisely the ones of the tensor hierarchies of gauged supergravity for maximal supersymmetry. In the past, these gauge theories have been regarded as gauge theories of Leibniz algebras or various notions of enhanced Leibniz algebras, see our discussion in \ref{sec:comparison}. By the principles of categorification, one expects higher gauge algebras to be some higher form of Lie algebras, as these are the ones integrating to (higher) symmetry Lie groups. We show that the various forms of enhanced Leibniz algebras proposed in the literature are indeed axiomatically incomplete forms of $\ophLie_2$-algebras or weaker higher Lie algebras.
    
    There are three main questions that remain or arise from our work. First, it would be certainly very interesting to explore further the relationship of our constructions to ones existing in the literature. We feel that e.g.~$\opEilh_2$-algebras should have appeared in other algebraic contexts; for example, the deformed Leibniz rule arising in $\ophLie_2$- and $\opEilh_2$-algebras is very similar to the formula in~\cite[Theorem 5.1]{Steenrod:1947aa} for Steenrod's cup products.\footnote{We thank Jim Stasheff for pointing out this potential link.}
    
    Second, most of our applications of $E_2L_\infty$-algebras involved them only in their hemistrict form, namely as $\ophLie_2$-algebras. This is due to the fact that we were only able to refine the derived bracket construction to a construction of an $\ophLie_2$-algebra from a differential graded Lie algebra. As we explain in \ref{ssec:gen_tensor_hierarchies}, there is a clear indication that some tensor hierarchies originate from $E_2L_\infty$-algebras that are not $\ophLie_2$-algebras but that can be obtained from $L_\infty$-algebras. This suggests a much wider generalization of the derived bracket construction, which would be certainly very useful to have. In particular, it would allow us to characterize a very large class of $E_2L_\infty$-algebras for which the problem of defining the kinematical data of higher adjusted gauge theories, such as the data arising in the tensor hierarchies, is fully under control. 
    
    Finally, it would be interesting to generalize our $E_2L_\infty$-algebras to contain the weak Lie 3-algebras of~\cite{Dehling:1710.11104} and 2-crossed modules of Lie algebras. We note that the latter are not contained in the former weak Lie 3-algebras, and seem to require a much more comprehensive generalization.
    
    \section{The operads \texorpdfstring{$\ophLie_2$}{hLie2} and \texorpdfstring{$\opEilh_2$}{Eilh2}}
    
    We start with the definition of the two Koszul-dual operads $\ophLie_2$ and $\opEilh_2$ that underlie our definition of $E_2L_\infty$-algebras. The invocation of operads provides us with a solid mathematical foundation of our constructions; the algebras over the operad $\ophLie_2$, however, will prove to be very interesting in their own right. For background on operads and Koszul duality, see~\cite{Ginzburg:0709.1228,0821843621,Loday:2012aa} as well as~\cite{Vallette:1202.3245}. We stress, however, that only an intuitive understanding of both topics is required for our discussion.
    
    \subsection{The operad \texorpdfstring{$\ophLie_2$}{hLie2}}
    
    Differential graded Lie algebras are algebras over the differential graded operad $\opLie$. In this section, we define a generalization of this operad whose algebras generalize the hemistrict Lie 2-algebras defined in~\cite{Roytenberg:0712.3461}, see also~\cite{Baez:2008bu}. 
    \begin{definition}\label{def:hLie2}
        The \uline{operad $\ophLie_2$} is a differential graded operad endowed with binary operations $\eps_2^i$ of degree~$-i$ for each $i\in \{0,1\}$ that satisfy the following relations: 
        \begin{equation}\label{eq:hLie-relations}
            \begin{aligned}
                \eps_1(\eps_1(x_1))&=0~,
                \\
                \eps_1(\eps^i_2(x_1,x_2))&=(-1)^i\big(\eps^i_2(\eps_1(x_1),x_2)+(-1)^{|x_1|}\eps^i_2(x_1,\eps_1(x_2))\big)
                \\
                &\hspace{1cm}+\eps^{i-1}_2(x_1,x_2)-(-1)^{i+|x_1|\,|x_2|}\eps^{i-1}_2(x_2,x_1)~,
                \\
                \eps^i_2(\eps_2^i(x_1,x_2),x_3)&=(-1)^{i(|x_1|+1)}\eps^i_2(x_1,\eps^i_2(x_2,x_3))-(-1)^{(|x_1|+i)|x_2|}\eps^i_2(x_2,\eps_2^i(x_1,x_3))~,
                \\
                \eps^{0}_2(\eps^{1}_2(x_1,x_2),x_3)&=0~,
                \\
                \eps^{1}_2(\eps^{0}_2(x_1,x_2),x_3)&=
                (-1)^{|x_1|}\eps_2^{0}(x_1,\eps_2^{1}(x_2,x_3))-(-1)^{|x_1|\,|x_2|}\eps_2^1(x_2,\eps_2^0(x_1,x_3))
            \end{aligned}
        \end{equation}
        for all $i\in \{0,1\}$, where $\eps_2^{-1}=0$.
    \end{definition}
    
    For $\ophLie_2$ to be well-defined as an operad, we need the following result:
    \begin{proposition}\label{prop:associativity}
        Relations~\eqref{eq:hLie-relations} are associative. That is, applying these to nested expressions in arbitrary order leads to the same result.
    \end{proposition}
    \begin{proof}
        We verify this by considering cubic expressions of the form
        \begin{equation}
            \eps_1(\eps_2^i(\eps_2^j(x_1,x_2),x_3))\eand
            \eps_2^i(\eps_2^j(\eps_2^k(x_1,x_2),x_3),x_4)
        \end{equation} 
        and applying relations~\eqref{eq:hLie-relations} in arbitrary order and in all possible positions. This is a straightforward computation. After fully simplifying all expressions, these always agree.
    \end{proof}
    
    An algebra over the operad $\ophLie_2$, or an $\ophLie_2$-algebra for short, is then a graded vector space $\frE$ together with a differential and a collection of binary products,
    \begin{equation}
        \begin{aligned}
            \eps_1&: \frE\rightarrow \frE~,~~~|\eps_1|=1~,
            \\
            \eps^i_2&: \frE\otimes \frE\rightarrow \frE~,~~~|\eps^i_2|=-i~,
        \end{aligned}
    \end{equation}
    such that~\eqref{eq:hLie-relations} are satisfied for all $x_1,x_2,x_3\in \frE$. 
    
    Note that for $i=0$, the first three relations in~\eqref{eq:hLie-relations} amount to the relations for a differential graded Leibniz algebra with differential $\eps_1$ and Leibniz product $\eps_2^0$. If $\eps_2^1$ vanishes, then $\eps_2^0$ is graded antisymmetric, and the Leibniz bracket becomes Lie. If we restrict to the case $\eps_2^1=0$, we simply recover differential graded Lie algebras. If we restrict ourselves to 2-term $\ophLie_2$-algebras, i.e.~$\ophLie_2$-algebras concentrated in degrees~$-1$ and $0$, we obtain the hemistrict Lie 2-algebras of~\cite{Roytenberg:0712.3461} with a graded symmetric $\eps_2^1$. The latter map is a chain homotopy sometimes called the \uline{alternator}, capturing the failure of $\eps_2^0(x_1,x_2)$ to be antisymmetric:
    \begin{equation}
        \begin{aligned}
            &\eps^0_2(x_1,x_2)+(-1)^{|x_1|\,|x_2|}\eps^0_2(x_2,x_1)
            \\
            &\hspace{2cm}=\eps_1(\eps^1_2(x_1,x_2))
            +\eps^1_2(\eps_1(x_1),x_2)+(-1)^{|x_1|}\eps^1_2(x_1,\eps_1(x_2))~.
        \end{aligned}
    \end{equation}
    
    It is certainly tempting\footnote{This was attempted in a first version of this paper on the arXiv.} to look for generalizations of $\ophLie_2$-algebras that allow us to remove the subscript 2 and that have binary products $\eps_2^i$ for $i\in \IN$ or even $i\in \IZ$. One immediate goal, e.g., would be to recover the hemistrict weak Lie 3-algebras of~\cite{Dehling:1710.11104}. This, however, proves to be quite subtle: while the hemistrict Lie 2-algebras of~\cite{Roytenberg:0712.3461} are easily generalized to $\IZ$-graded vector spaces, this is not the case for the hemistrict weak Lie 3-algebras of~\cite{Dehling:1710.11104}. In particular, \ref{prop:associativity} is not satisfied. Even if we succeeded, it is already clear that 2-crossed modules of Lie algebras (the linearizations of the structures introduced in~\cite{Conduche:1984:155,Conduche:2003}), can not be captured in this way, as the Peiffer lifting has no corresponding structure, but requires rather a second product of degree $0$ besides $\eps_2^0$. Finally, it turns out that all our applications will merely require $\ophLie_2$-algebras, and we therefore content ourselves with these.
    
    We close this section with a useful observation: we can create a larger $\ophLie_2$-algebra by tensoring $\ophLie_2$-algebras with differential graded commutative algebras, just as in the case of Lie algebras, 
    \begin{proposition}\label{prop:tensor_product_dgca_hLie}
        The tensor product of a differential graded commutative algebra and an $\ophLie_2$-algebra carries a natural $\ophLie_2$-algebra structure.
    \end{proposition}
    \begin{proof}
        On the tensor product of a differential graded commutative algebra $(\frA,\rmd)$ and an $\ophLie_2$-algebra $\frE$, we define 
        \begin{equation}
            \begin{aligned}
                \hat \frE\coloneqq \frA\otimes \frE=\oplus_{k\in \IN}(\frA\otimes \frE)_k~,~~~(\frA\otimes \frE)_k=\sum_{i+j=k} \frA_i\otimes \frE_j~,
                \\
                \hat\eps_1(a_1\otimes x_1)=(\rmd a_1)\otimes x_1+(-1)^{|a_1|}a_1\otimes \eps_1(x_1)~,
                \\
                \hat\eps^i_2(a_1\otimes x_1,a_2\otimes x_2)=(-1)^{i(|a_1|+|a_2|)}(a_1a_2)\otimes \eps^i_2(x_1,x_2)~.
            \end{aligned}
        \end{equation}
        One then readily verifies the axioms~\eqref{eq:hLie-relations}.
    \end{proof}

    \subsection{The operad \texorpdfstring{$\opEilh_2$}{Eilh2}}\label{ssec:Eilh-algebras}
    
    It is one of our aims to construct homotopy $\ophLie_2$-algebras. For this, we need to have the Koszul-dual operad to $\ophLie_2$ at our disposal. The subtlety here is that the defining relations of the operad $\ophLie_2$ are not quadratic, as is the case e.g.~for the operads $\opLie$, $\opAss$, or $\opCom$, but quadratic-linear.

    \begin{definition}\label{def:operad_Eilh2}
        The operad $\opEilh_2$ is a differential graded operad endowed with binary operations $\oslash_i$ of degree\footnote{Our convention for identifying the correct Koszul sign is that we regard $\oslash_i$ as a binary function of degree~$i$. For example,
            \begin{equation}
                t_1 \oslash_i t_2\coloneqq \oslash_i(t_1,t_2)~,~~~(-\oslash_i (-\oslash_j-))(t_1,t_2,t_3)=(-1)^{j|t_1|}t_1\oslash_i (t_2\oslash_j t_3)~.
            \end{equation}} $i\in \{0,1\}$ that satisfy the following identities:
        \begin{equation}\label{eq:sym_oslash_1}
            a\oslash_1 b=(-1)^{|a|\,|b|}b\oslash_1 a~,
        \end{equation}
        \begin{equation}\label{eq:Eilh-relations}
            \begin{aligned}
                a\oslash_i(b\oslash_i c)&=-(-1)^{i(|a|+1)}((a\oslash_i b)\oslash_i c+(-1)^{|a|\,|b|}(b\oslash_i a)\oslash_i c)~,
                \\
                a\oslash_0(b\oslash_1 c)&=(-1)^{|a|}(a \oslash_0 b)\oslash_1 c~,
                \\
                a\oslash_1(b\oslash_0 c)&=(-1)^{|a|\,|b|}(b \oslash_0 a)\oslash_1 c~.
            \end{aligned}
        \end{equation}
        The differential fulfills a deformed Leibniz rule given by
        \begin{equation}\label{eq:def_Leibniz}
            \begin{aligned}
                Q(a\oslash_i b)&=(-1)^i\big((Qa)\oslash_i b+(-1)^{|a|}a\oslash_i Qb\big)\\
                &\hspace{1cm}+(-1)^i (a \oslash_{i+1}b)-(-1)^{|a|\,|b|} (b\oslash_{i+1} a)
            \end{aligned}
        \end{equation}
        with $\oslash_2\coloneqq0$.
    \end{definition}
    
    We then have the following relation between $\ophLie_2$ and $\opEilh_2$
    \begin{proposition}
        The operad $\opEilh_2$ is the Koszul dual of the operad $\ophLie_2$.
    \end{proposition}
    This fact can be established by computation, preferably using a computer algebra program, see also~\cite[Section 7.6.4]{Loday:2012aa} for details. Instead of this, let us construct the analogue of the Chevalley--Eilenberg algebra for an $\ophLie_2$-algebra by constructing a semifree $\opEilh_2$-algebra.
    
    We start from an $\ophLie_2$-algebra $\frE$ which we assume for clarity of the discussion to have an underlying degree-wise dualizable vector space with basis $(\tau_\alpha)$. We will make this assumption for all the graded vector spaces from here on, mostly for pedagogical reasons: it allows us to give very explicit formulas. More generally, one may want to work with (graded) pseudocompact vector spaces, cf.~the discussion in~\cite[\S 1.1]{Guan:1909.11399}.
    
    Consider the graded vector space\footnote{Our convention for degree shift is the common one, $V[k]\coloneqq \bigoplus_i V[k]_i$ with $V[k]_i\coloneqq V_{k+i}$, implying that $V[k]$ is the graded vector space $V$ shifted by $-k$.} $V\coloneqq \frE[1]^*$ and together with its free, non-associative tensor algebra $\bigoslash^\bullet_\bullet V$ with respect to both products $\oslash_0$ and $\oslash_1$ simultaneously, and taking into account the symmetry of $\oslash_1$. The quadratic identities~\eqref{eq:Eilh-relations} allow us to rearrange all tensor products such that they are nested from left to right. Thus, we define 
    \begin{equation}
        \begin{aligned}
            {\caE(V)}&\coloneqq \IR~\oplus~V~\oplus~\bigoplus_{i\in \{0,1\}}V\oslash_i V~\oplus~\bigoplus_{i,j\in \{0,1\}}(V\oslash_i V)\oslash_j V~\oplus~\dotsb
            \\
            &\eqqcolon\caE^{(0)}(V)~\oplus~\caE^{(1)}(V)~\oplus~\caE^{(2)}(V)~\oplus~\caE^{(3)}(V)~\oplus~\dotsb~.
        \end{aligned}
    \end{equation}
    We also define the reduced tensor algebra 
    \begin{equation}\label{eq:red_ten_algebra}
        \overline{\caE(V)}\coloneqq V~\oplus~\bigoplus_{i\in \{0,1\}}V\oslash_i V~\oplus~\bigoplus_{i,j\in \{0,1\}}(V\oslash_i V)\oslash_j V~\oplus~\dotsb~,
    \end{equation}    
    which is, in fact, sufficient for the description of $\ophLie_2$-algebras. We will call an $\opEilh_2$-algebra of the form $(\caE(V),Q)$ for some graded vector space $V$ without any restrictions on the products $\oslash_i$ beyond~\eqref{eq:Eilh-relations} \uline{semifree}, since the only further implied relations are due to the application of the differential $Q$.
    
    Semifree $\opEilh_2$-algebras yield the Chevalley--Eilenberg description of $\ophLie_2$-algebras:
    \begin{definition}
        The \uline{Chevalley--Eilenberg algebra} $\sfCE(\frE)$ of an $\ophLie_2$-algebra $\frE$ with basis $(\tau_\alpha)$ whose differential and binary products are given by 
        \begin{equation}
            \begin{aligned}
                \eps_1&: \frE\rightarrow \frE~,~~~&\tau_\alpha&\mapsto m^\beta_\alpha \tau_\beta~,~~~&&|\eps_1|=1~,
                \\
                \eps^i_2&: \frE\otimes \frE\rightarrow \frE~,~~~&\tau_\alpha\otimes \tau_\beta&\mapsto m^{i,\gamma}_{\alpha\beta}\tau_\gamma~,~~~&&|\eps^i_2|=-i
            \end{aligned}
        \end{equation}
        for some $m^\beta_\alpha$ and $m^{i,\gamma}_{\alpha\beta}$ taking values in the underlying ground field is the semifree $\opEilh_2$-algebra $\overline{\caE(V)}$ with $V=\frE[1]^*$ and the differential
        \begin{equation}\label{eq:Q_hLie}
            Q t^\alpha=-(-1)^{|\beta|}m^\alpha_\beta t^\beta-(-1)^{i(|\beta|+|\gamma|)+|\gamma|(|\beta|-1)}\,m^{i,\alpha}_{\beta\gamma}\,t^\beta \oslash_i t^\gamma~,
        \end{equation}
        where $t^\alpha$ denotes a basis of $V=\frE[1]^*$, which is grade-shifted dual to $\tau_\alpha$. Moreover, $|\beta|$ is shorthand for $|t^\beta|$, the degree of $t^\beta$ in $V$.
    \end{definition}
    \noindent Note that~\eqref{eq:Q_hLie} makes the reason for the graded antisymmetry~\eqref{eq:sym_oslash_1} of $\oslash_1$ obvious.    
    
    Helpful for further computations is now the following lemma.
    \begin{lemma}\label{thm:restriction}
        The equation $Q^2=0$ on a semifree $\opEilh_2$-algebra $\caE(V)$ is equivalent to $Q^2V=0$. 
    \end{lemma}
    \begin{proof}
        Using the deformed Leibniz rule~\eqref{eq:def_Leibniz}, we have
        \begin{equation}
            Q^2(a\oslash_i b)=(Q^2 a)\oslash_i b+a\oslash_i (Q^2 b)~,
        \end{equation}
        and iterating this relation, we obtain the desired result. 
    \end{proof}
    
    In the case of Lie algebras and $L_\infty$-algebras, the (homotopy) Jacobi relations are equivalent to a nilquadratic differential in the corresponding Chevalley--Eilenberg algebra, and this is also the case here:
    \begin{proposition}
        The equation $Q^2=0$ for the differential of the Chevalley--Eilenberg algebra of an $\ophLie_2$-algebra $\frE$ is equivalent to the relations~\eqref{eq:hLie-relations} for $\frE$.
    \end{proposition}
    \begin{proof}
        Using \ref{thm:restriction}, the proof is a straightforward computation starting from equation~\eqref{eq:Q_hLie} and applying $Q$ once more.
    \end{proof}
    
    \subsection{Cohomology of semifree \texorpdfstring{$\opEilh_2$}{Eilh2}-algebras}\label{ssec:cohomology_Eilh}
    
    An important tool in studying Lie algebras is Lie algebra cohomology, and we consider the generalization to $\ophLie_2$-algebras in the following. As we saw above, this amounts to the cohomology of semifree $\opEilh_2$-algebras. Due to the deformation of the usual Leibniz rule to~\eqref{eq:def_Leibniz}, a subtle and important point arises. For ordinary differential graded algebras, the cohomology again carries the structure of a differential graded algebra of the same type, with product induced by the product on the original algebra. In particular, the product of two cocycles is again a cocycle. The deformation of the Leibniz rules can now evidently break this connection.   
    
    To start, let us consider the cohomology of the semifree $\opEilh_2$-algebra $(\caE(V),Q_0)$ with the most trivial differential $Q_0$ satisfying
    \begin{equation}
        \begin{aligned}
            Q_0(v)&=0~,
            \\
            Q_0(a\oslash_i b)&=(-1)^i\big((Q_0a)\oslash_i b+(-1)^{|a|}a\oslash_i Q_0b\big)
            \\
            &\hspace{1cm}+(-1)^i (a \oslash_{i+1}b)-(-1)^{|a|\,|b|} (b\oslash_{i+1} a)
        \end{aligned}
    \end{equation}
    for all $v\in V$ and $a,b\in \caE(V)$. 
    
    It is useful to decompose the product $\oslash_0$ into even and odd parts:
    \begin{equation}\label{eq:decomposition}
        a\oslash_0 b=\tfrac12(a \oslash_0^+ b+ a\oslash_0^- b)
    \end{equation} 
    for $a,b\in \caE(V)$ with
    \begin{equation}
        a \oslash_0^\pm b=a\oslash_0 b\pm(-1)^{|a|\,|b|}b\oslash_0 a~.
    \end{equation}
    We note that
    \begin{equation}\label{eq:elementary_Q0}
        Q_0(v_1\oslash_0^+ v_2)=0
        \eand
        Q_0(v_1\oslash_0^- v_2)=v_1\oslash_1 v_2
    \end{equation} 
    for $v_{1,2}\in V$.
    
    Furthermore, we have the following lemma:
    \begin{lemma}\label{lem:embedding}
        Define an embedding map $\sfE_0:\bigodot^\bullet V\mapsto \caE(V)$ by 
        \begin{equation}\label{eq:embedding_map}
            \sfE_0(v_1\odot \dotsb \odot v_n)\coloneqq \sum_{\sigma\in S_n}(\dotsb(v_{\sigma(1)}\oslash_0 v_{\sigma(2)})\oslash_0\dotsb )\oslash_0 v_{\sigma(n)}~.
        \end{equation}
        For $n\geq 2$, we have 
        \begin{equation}
            \sfE_0(v_1\odot \dotsb \odot v_n)=(((v_1\oslash_0^+v_2)\oslash_0^+v_3)\ldots)\oslash_0^+ v_n
        \end{equation} 
    \end{lemma}
    \begin{proof}
        The proof proceeds by induction. The case $n=2$ is obviously true. Consider now 
        \begin{equation}
            \begin{aligned}
                &(((v_1\oslash_0^+v_2)\oslash_0^+v_3)\ldots)\oslash_0^+ v_n)\oslash_0^+ v_{n+1}
                \\
                &~~~=
                \left(\sum_{\sigma\in S_n}(\dotsb(v_{\sigma(1)}\oslash_0 v_{\sigma(2)})\oslash_0\dotsb )\oslash_0 v_{\sigma(n)}\right)\oslash_0^+ v_{n+1}~.
            \end{aligned}
        \end{equation}
        Substituting the definition of $\oslash_0^+$, we can iteratively apply the first relation in~\eqref{eq:Eilh-relations}. This distributes $v_{n+1}$ into all possible positions together with the correct Koszul sign.
    \end{proof}
    
    This lemma now allows us to prove an important result.
    \begin{proposition}\label{prop:Q0_cohomology}
        The $Q_0$-cohomology of $(\caE(V),Q_0)$ is the vector space\footnote{Here, $\odot$ denotes the symmetrized tensor product.} $\sfE_0(\bigodot^\bullet V)$, with $\sfE_0$ defined in~\eqref{eq:embedding_map}. This vector space carries the evident structure of a differential graded commutative algebra.
    \end{proposition}
    \begin{proof}
        First, consider an element in $\caE(V)$ of the form
        \begin{equation}
            (((v_1\oslash_0 v_2)\oslash_0 v_3)\ldots)\oslash_0 v_n~.
        \end{equation} 
        We decompose each product into even and odd parts $\oslash_0^\pm$. Because of~\eqref{eq:elementary_Q0}, it is clear that the part
        \begin{equation}
            \sfE_0(v_1\odot \dotsb \odot v_n)=(((v_1\oslash_0^+v_2)\oslash_0^+v_3)\ldots)\oslash_0^+ v_n
        \end{equation} 
        is in the kernel of $Q_0$, while all other parts will be mapped to elements involving one product $\oslash_1$.
        
        Beyond this, the kernel of $Q_0$ consists of sums of terms with at least one product $\oslash_1$ and all other products $\oslash_0^+$, e.g.
        \begin{equation}
            a=(((v_1\oslash_0^+ v_2)\oslash_1 v_3)\oslash_0^+ v_4)\oslash_1 v_5
        \end{equation} 
        with $v_{1,2,3,4,5}\in V$. Such terms, however, are easily seen to be in the image of $Q_0$: to obtain a pre-image, one replaces one product $\oslash_1$ by $\oslash_0^-$. For example, we have 
        \begin{equation}
            \begin{aligned}
                a&=Q\left(\tfrac14(((v_1\oslash_0^+ v_2)\oslash_0^- v_3)\oslash_0^+ v_4)\oslash_1 v_5\right)
                \\
                &=Q\left(-\tfrac14(((v_1\oslash_0^+ v_2)\oslash_1 v_3)\oslash_0^+ v_4)\oslash_0^- v_5\right)~.
            \end{aligned}
        \end{equation}
        Hence, there is no new contribution to cohomology.
        
        The kernel then necessarily consists also of the difference of the possible pre-images, which are again in the image of $Q_0$, and again the pre-images are obtained by replacing one product $\oslash_1$ by $\oslash_0^-$, again not leading to any new contribution to cohomology. Iterating this argument, we see that the cohomology can indeed be identified with the image of $\sfE_0$.
    \end{proof}
    
    \Cref{prop:Q0_cohomology} together with the usual arguments underlying a general Hodge--Ko\-daira decomposition\footnote{see e.g.~\cite{Weibel:1994aa}} then yield the following theorem:
    \begin{theorem}
        Consider the trivial semifree $\opEilh_2$-algebra from \ref{prop:Q0_cohomology}. Then we have the following contracting homotopy between $\caE(V)$ and the differential graded commutative algebra $(\bigodot^\bullet V,0)$:
        \begin{subequations}\label{eq:contracting_hom_Q0}
            \begin{equation}
                \begin{tikzcd}
                    \ar[loop,out=194,in= 166,distance=20,"\sfH_0"](\caE(V),Q_0)\arrow[r,shift left]{}{\sfP_0} & (\bigodot^\bullet V,0) \arrow[l,shift left]{}{\sfE_0}~,
                \end{tikzcd}
            \end{equation}
            with the projection defined by
            \begin{equation}\label{eq:projector_P0}
                \begin{aligned}
                    \sfP_0(v_1)&=v_1
                    \\
                    \sfP_0\big(((v_0\oslash_{i_1} v_2)\oslash_{i_2} v_3)\dotsb \oslash_{i_n} v_n\big)&=\begin{cases}
                        \tfrac{1}{(n+1)!}v_0\odot v_1\odot\dotsb\odot v_n& i_1=\dotsb=i_n=0~,
                        \\
                        0 & \mbox{else}
                    \end{cases}
                \end{aligned}
            \end{equation}
            for all $v_0,\ldots, v_n\in V$, such that
            \begin{equation}\label{eq:homotopy_transfer_relations}
                \begin{gathered}
                    \sfH_0\circ \sfH_0= \sfH_0\circ  \sfE_0=0~,~~~ \sfP_0\circ  \sfH_0=0~,
                    \\
                    \sfid_{\caE(V)}- \sfE_0\circ  \sfP_0= \sfH_0\circ  Q_0+ Q_0\circ  \sfH_0~.
                \end{gathered}
            \end{equation}
        \end{subequations}
    \end{theorem}
    \noindent An explicit form of the homotopy $\sfH_0$ is the following:
    \begin{equation}\label{eq:homotopy_H0}
        \begin{aligned}
            \sfH_0(a\oslash_0^-v)&\coloneqq 0~,
            \\
            \sfH_0(a\oslash_0^+v)&\coloneqq \sfH_0(a)\oslash_0^+v~,
            \\
            \sfH_0(a\oslash_1 v)&=\coloneqq \tfrac14 a\oslash_0^- v~,
            \\
            \sfH_0(v)&=0
        \end{aligned}
    \end{equation}
    for all $a\in\caE(V)$ and $v\in V$, which is easily checked by computation.
    
    We note that such a contracting homotopy for ordinary differential graded algebras often induces an algebra morphism $\Phi\coloneqq \sfE_0\circ \sfP_0$. This is not the case here, as clearly $\Phi(a)\oslash_i\Phi(b)\neq \Phi(a\oslash_i b)$ in general. We shall return to this point in~\ref{ssec:EL_infty_and_L_infty}.
    
    The contracting homotopy~\eqref{eq:contracting_hom_Q0} has a number of important generalizations and applications. Here, we note that it evidently extends to differentials $Q_{\rm lin}=Q_0+\rmd$, where $\rmd:V\rightarrow V$ and $\rmd$ satisfies the ordinary Leibniz rules on $\caE(V)$ and $\bigodot^\bullet V$:
    \begin{equation}\label{eq:homotopy_for_transfer_EL_to_L}
        \begin{tikzcd}
            \ar[loop,out=194,in= 166,distance=20,"\sfH_0"](\caE(V),Q_0+\rmd)\arrow[r,shift left]{}{\sfP_0} & (\bigodot\nolimits^\bullet V,\rmd) \arrow[l,shift left]{}{\sfE_0}~.
        \end{tikzcd}
    \end{equation}
    Moreover, if we have a differential $\rmd:\bigodot^\bullet V\rightarrow \bigodot^\bullet V$ non-linear in $\odot$, we still have a corresponding contracting homotopy
    \begin{equation}
        \begin{tikzcd}
            \ar[loop,out=194,in= 166,distance=20,"\sfH_0"](\caE(V),Q_0+Q_1)\arrow[r,shift left]{}{\sfP_0} & (\bigodot\nolimits^\bullet V,\rmd) \arrow[l,shift left]{}{\sfE_0}
        \end{tikzcd}
    \end{equation}
    with
    \begin{equation}
        Q_1=\sfE_0\circ \rmd\circ \sfP_0~.
    \end{equation}
    We therefore arrive at the following statement:
    \begin{theorem}\label{thm:lift_dgca_to_Eilh}
        Any semifree differential graded commutative algebra $(\bigodot^\bullet V,\rmd)$ gives rise to the semifree $\opEilh_2$-algebra $(\caE(V),Q_0+\sfE_0\circ \rmd\circ \sfP_0)$. 
    \end{theorem}
    
    We will return to this thread of our discussion later.
    
    \section{\texorpdfstring{$E_2L_\infty$}{E2L-infinity}-algebras, their morphisms, and structure theorems}\label{sec:EL_infty}
    
    We can now define the homotopy algebras of $\ophLie_2$-algebras, generalizing $L_\infty$-algebras.
    
    \subsection{\texorpdfstring{$E_2L_\infty$}{E2L-infinity}-algebras and examples}
    
    \begin{definition}\label{def:E2Linfty}
        The differential graded operad \(\opELinfty=\sfT(s^{-1}\overline{\ophLie_2^\antishriek})\) is the cobar construction applied to the Koszul-dual cooperad \(\ophLie_2^\text{¡}\). An algebra over it is called an \underline{$E_2L_\infty$-algebra}.
    \end{definition}
    
    This abstract definition needs to be unwrapped to become useful. To get an explicit handle on $E_2L_\infty$-algebras, we consider their Chevalley--Eilenberg algebras. For clarity, we restrict ourselves again to the cases where the graded vector space $\frE$ comes with a nice basis and can be dualized degree-wise, cf.~\ref{ssec:Eilh-algebras}. Then an $E_2L_\infty$-algebra structure on a graded vector space $\frE$ is encoded in a nilquadratic differential on the $\opEilh_2$-algebra $\sfCE(\frE)\coloneqq \caE(V)$ for $V=\frE[1]^*$. 
    
    Clearly, the differential $Q$ is fully specified by its action on $V$. With respect to a basis $(t^\alpha)$ on $V$, we encode this action in structure constants $m$ taking values in the ground field as follows:
    \begin{equation}\label{eq:structure_constants_ELinfty}
        Q t^\alpha=\pm m^\alpha\pm m^\alpha_\beta t^\beta~\pm~m^{i_1,\alpha}_{\beta_1\beta_2}~t^{\beta_1}\oslash_{i_1} t^{\beta_2}~\pm~m^{i_1i_2,\alpha}_{\beta_1\beta_2\beta_3}~(t^{\beta_1}\oslash_{i_1} t^\gamma)\oslash_{i_2} t^\delta+\ldots~.
    \end{equation}
    Here the choice of signs $\pm$ is a convention\footnote{For $L_\infty$-algebras, we follow the conventions of~\cite{Jurco:2018sby} for the structure constants and the differential in the Chevalley--Eilenberg algebra.}, and we shall be more explicit below, cf.\linebreak also~\eqref{eq:Q_hLie}. The structure constants $m$ define \uline{higher products},
    \begin{equation}
        \begin{gathered}
            \eps^I_n:\frE^{\otimes n}\rightarrow \frE~,
            \\
            \eps_0=m^\alpha \tau_\alpha~,~~~
            \eps_1(\tau_\alpha)=m^\beta_\alpha \tau_\beta~,~~~\eps_2^i(\tau_\alpha,\tau_\beta)=m^{i,\gamma}_{\alpha\beta}\tau_\gamma~,~~~\ldots
            \\
            \eps^I_n(\tau_{\alpha_1},\ldots,\tau_{\alpha_n})=m_{\alpha_1\ldots\alpha_n}^{I,\beta} \tau_\beta~,
        \end{gathered}
    \end{equation}
    where $I$ is a multi-index consisting of $n-1$ indices $i_1,i_2,\ldots,\in \IN$, and we define $|I|\coloneqq i_1+i_2+\ldots$. The products $\eps^I_n$ have degree~$-|I|$.
    
    It is useful to identify the following special cases, extending the nomenclature of~\cite{Roytenberg:0712.3461}:
    \begin{definition}\label{def:special_EL_infty}
        Let $(\frE,\eps_k^I)$ be an $E_2L_\infty$-algebra. If $\eps_0\neq 0$, we call $\frE$ \uline{curved}, otherwise $\frE$ is \uline{uncurved}. An uncurved $E_2L_\infty$-algebra is called \uline{hemistrict}, if $\eps_k^I=0$ for $k\geq 3$. It is called \uline{strict} if it is hemistrict and $\eps_2^i=0$ for $i>0$. Finally, $\frE$ is called \uline{semistrict} if $\epsilon_n^I=0$ for $I\neq (0,\ldots,0)$.
    \end{definition}
    Note that in the case of uncurved $E_2L_\infty$-algebras, we can restrict $\caE(V)$ to the reduced tensor product algebra $\overline{\caE(V)}$ defined in~\eqref{eq:red_ten_algebra}. In the following, all our $E_2L_\infty$-algebras will be uncurved. Clearly, hemistrict $E_2L_\infty$-algebras are simply $\ophLie_2$-algebras, and therefore $E_2L_\infty$-algebras subsume differential graded Lie algebras.
    
    As an example, let us consider a family of weak models of the string Lie 2-algebra which we will use to exemplify many of our constructions in the following. We consider the graded vector space $V=(\frg\oplus \IR[1])[1]^*$, where $\frg$ is a finite-dimensional quadratic (i.e.~metric) Lie algebra with structure constants $f^\alpha_{\beta\gamma}$ and Cartan--Killing form $\kappa_{\alpha\beta}$ with respect to a basis $(\tau_\alpha)$. On $V$, we introduce basis vectors $t^\alpha\in \frg[1]^*$ and $r\in \IR[2]^*$ of degrees $1$ and $2$, respectively. The differential is given by
    \begin{equation}
        \begin{aligned}
            Q t^\alpha&=-f^\alpha_{\beta\gamma}\,t^\beta \oslash_0 t^\gamma+(1-\vartheta)\kappa_{\alpha\delta} f^\delta_{\beta\gamma}\,(t^\beta \oslash_0 t^\gamma) \oslash_0 t^\delta-2\vartheta \kappa_{\beta\gamma}\,t^\beta \oslash_0 t^\gamma~,
            \\
            Q r &=0
        \end{aligned}
    \end{equation}
    with $\vartheta\in \IR$, and a direct calculation verifies $Q^2=0$. This defines the family of $E_2L_\infty$-algebras $\frstring^{\rmwk,\vartheta}_\rmsk(\frg)$ with the following underlying graded vector space and higher products:
    \begin{equation}\label{eq:skeltal_model_string_Lie_2_algebra}
        \begin{gathered}
            \frstring^{\rmwk,\vartheta}_\rmsk(\frg)~\coloneqq ~(\IR[1]\xrightarrow{~0~}\frg)~,
            \\
            \eps_2(x_1,x_2)=[x_1,x_2]~,~~~\eps_2^1(x_1,x_2)=2\vartheta(x_1,x_2)
            \\
            \eps^{00}_3(x_1,x_2,x_3)=(1-\vartheta)(x_1,[x_2,x_3])~,
        \end{gathered}
    \end{equation}
    for $x\in \frg$ and $y\in \IR$. All other higher products vanish. We notice that $\frstring^{\rmwk,\vartheta}_\rmsk(\frg)$ is an uncurved $E_2L_\infty$-algebra for each $\vartheta\in\IR$. It becomes semistrict for $\vartheta=0$ and hemistrict for $\vartheta=1$.
    
    A second, much more general example is worked out in \ref{app:involved_example}.
    
    It also evident that the 2-term $EL_\infty$-algebras of~\cite{Roytenberg:0712.3461} with symmetric alternators are a special case of our $E_2L_\infty$-algebras.
    
    Another very general and useful example is the $E_2L_\infty$-algebra of inner derivations $\frinn(\frE)$ of another $E_2L_\infty$-algebra $\frE$. This is obtained as a straightforward generalization of the definition of the (unadjusted\footnote{We clarify this nomenclature later.}) Weil algebra of an $L_\infty$-algebra.
    \begin{definition}\label{def:Weil-EL-infty}
        The \uline{(unadjusted) Weil algebra} of an $E_2L_\infty$-algebra $\frE$ is the $\opEilh_2$-algebra
        \begin{equation}
            \sfW(\frE)\coloneqq \Big(~\bigoslash\nolimits^\bullet~(\frE[1]^*\oplus \frE[2]^*)~,~Q_\sfW~\Big)~,
        \end{equation}
        where the Weil differential is defined by the relations
        \begin{equation}
            Q_\sfW=Q_\sfCE+\sigma~,~~~Q_\sfW\sigma=-\sigma Q_\sfW
        \end{equation}
        with $\sigma:\frE[1]^*\rightarrow \frE[2]^*$ the shift isomorphism, trivially extended to $\sfW(\frE)$ by the (undeformed) Leibniz rule, i.e.
        \begin{equation}
            \sigma(a\oslash_i b)=(-1)^i\big(\sigma a\oslash_i b+(-1)^{|a|}a \oslash_i \sigma b\big)~.
        \end{equation}
        The $E_2L_\infty$-algebra dual to $\sfW(\frE)$ is the \uline{inner derivation $E_2L_\infty$-algebra} $\frinn(\frE)$ of $\frE$.
    \end{definition}    
    
    \subsection{Morphisms and equivalences of \texorpdfstring{$E_2L_\infty$}{E2L-infinity}-algebras}
    
    The notion of a morphism of $E_2L_\infty$-algebras becomes evident in the dual, Chevalley--Eilenberg algebra description.
    \begin{definition}\label{def:morphisms}
        A \uline{morphism of $E_2L_\infty$-algebras} $\phi:\frE\rightarrow \tilde \frE$ is a morphism dual to the corresponding morphism $\Phi:\sfCE(\frE)\rightarrow \sfCE(\tilde \frE)$ of $\opEilh_2$-algebras. For $\sfCE(\frE)=({\caE(V)}, Q)$ and $\sfCE(\tilde \frE)=({\caE(\tilde V)},\tilde Q)$, such a morphism is compatible with the differential,
        \begin{equation}
            Q\circ \Phi=\Phi\circ \tilde Q~,
        \end{equation}
        and the product structure,
        \begin{equation}
            \Phi(x\oslash_i y)=\Phi(x)\oslash_i\Phi(y)
        \end{equation}
        for all $x,y\in \sfCE(\frE)$ and $i\in \IN$.
    \end{definition}
    
    Recall that the appropriate notion of equivalence for homotopy algebras is that of a quasi-isomorphism, which we make explicit in the following definition.
    \begin{definition}
        Two $E_2L_\infty$-algebras $\frE$ and $\tilde \frE$ are called \underline{quasi-isomorphic} if there is a morphism of $E_2L_\infty$-algebras $\phi:\frE\rightarrow \tilde \frE$ such that the contained chain map $\phi_1$, the dual to the linear component of the dual morphism of $\opEilh_2$-algebras $\Phi$, descends to an isomorphism between the cohomologies of $\frE$ and $\tilde \frE$.
    \end{definition}
    We note that the cohomology $H^\bullet_{\eps_1}(\frE)$ of an $E_2L_\infty$-algebra $\frE$ is dual to the cohomology with respect to the linear part of the Chevalley--Eilenberg differential $Q$. In the case of the Weil algebra, the included shift isomorphism $\sigma$ renders the cohomology $H^\bullet_{\eps_1}(\frinn(\frE))$ trivial. We therefore obtain a quasi-isomorphism between $\frinn(\frE)$ and the trivial $E_2L_\infty$-algebra, extending the situation for $L_\infty$-algebras.
    
    In the future, it may be useful to have an inner product structure on an $E_2L_\infty$-algebra. The appropriate notion here, which could more formally be derived from lifting our above discussion to cyclic operads, is a simple generalization of cyclic structures on $L_\infty$-algebras.
    \begin{definition}
        An $E_2L_\infty$-algebra $\frE$ is called \uline{cyclic} if it is equipped with a non-degenerate graded-symmetric bilinear form $\langle-,-\rangle:\frE\times \frE\rightarrow K$, where $K$ is the ground field, such that
        \begin{equation}\label{eq:cyclicity}
            \langle e_1,\eps_i(e_2,\ldots,e_{i+1})\rangle=(-1)^{i+i(|e_1|+|e_{i+1}|)+|e_{i+1}|\sum_{j=1}^{i}|e_j|}\langle e_{i+1},\eps_i(e_1,\ldots,e_{i})\rangle
        \end{equation}
        for all $e_i\in \frE$.
    \end{definition}
    
    \subsection{Homotopy transfer and minimal model theorem}\label{ssec:homotopy_transfer}
    
    Generically, homotopy algebras provide a notion of homotopy transfer, cf.~e.g.~\cite{Loday:2012aa}. We will require this technology later, and we therefore develop a form of the homological perturbation lemma below. 
    
    We start from a deformation retract between two differential graded complexes $(\frE,\eps_1)$ and $(\tilde \frE,\tilde \eps_1)$. That is, we have chain maps $\sfp$ and $\sfe$, together with a map $\sfh$ of degree~$-1$, which fit into the diagram
    \begin{subequations}\label{eq:HE_data_diff_complex}
        \begin{equation}
            \begin{tikzcd}
                \ar[loop,out=194,in= 166,distance=20,"\sfh"](\frE,\eps_1)\arrow[r,shift left]{}{\sfp} & (\tilde \frE,\tilde \eps_1) \arrow[l,shift left]{}{\sfe}~,
            \end{tikzcd}
        \end{equation}
        and satisfy the relations
        \begin{equation}\label{eq:homotopy_transfer_relations_dual}
            \begin{gathered}
                \sfp\circ \sfe=\sfid_{\tilde \frE}~,~~~\sfid_{\frE}-\sfe\circ \sfp=\sfh\circ \eps_1+\eps_1\circ \sfh~,
                \\
                \sfh\circ\sfh=0~,~~~\sfh\circ \sfe=0~,~~~\sfp\circ \sfh=0~.
            \end{gathered}
        \end{equation}
    \end{subequations}
    We now want to consider the homological perturbation lemma for the semifree $\opEilh_2$-algebras on ${\caE(V)}$ and ${\caE(\tilde V)}$ with differentials $Q$ and $\tilde Q$, respectively, defined by
    \begin{equation}
        V=\frE[1]^*~,~~~Q=\eps_1^*\eand \tilde V=\tilde \frE[1]^*~,~~~\tilde Q=\tilde \eps_1^*~.
    \end{equation}
    
    \begin{theorem}\label{thm:homotopy_transfer}
        The deformation retract~\eqref{eq:HE_data_diff_complex} lifts to the deformation retract
        \begin{subequations}\label{eq:HE_data_Eilh}
            \begin{equation}
                \begin{tikzcd}
                    \ar[loop,out=194,in= 166,distance=20,"\sfH_0"]({\caE(V)},Q)\arrow[r,shift left]{}{\sfP_0} & ({\caE(\tilde V)},\tilde Q) \arrow[l,shift left]{}{\sfE_0}~,
                \end{tikzcd}
            \end{equation}
            with\footnote{The notation is chosen to match more closely the formulas of \ref{ssec:cohomology_Eilh}.}
            \begin{equation}
                \begin{aligned}
                    \sfE_0(v)&=\sfp^*(v)~,~~~&\sfE_0(x\oslash_i y)&=\sfE_0(x)\oslash_i\sfE_0(y)~,
                    \\
                    \sfP_0(v)&=\sfe^*(v)~,~~~&\sfP_0(x\oslash_i y)&=\sfP_0(x)\oslash_i\sfP_0(y)~,
                \end{aligned}
            \end{equation}
            for all $v\in V$ and $x,y\in {\caE(V)}$. The dual to the contracting homotopy is continued by a modification of the tensor trick, 
            \begin{equation}
                \begin{aligned}
                    \sfH_0(v)&=\sfh^*(v)~,
                    \\
                    \sfH_0(x\oslash_i y)&=(-1)^i\sfH_0(x)\oslash_i\sfE_0(\sfP_0(y))+(-1)^{i+|x|}x\oslash_i\sfH_0(y)
                    \\
                    &\hspace{4cm}+(-1)^{|y|+|x|\,|y|}\sfH_0(x)\oslash_{i+1}\sfH_0(y)~.
                \end{aligned}
            \end{equation}
            These maps satisfy the relations
            \begin{equation}\label{eq:homotopy_transfer_relations_0}
                \begin{gathered}
                    \sfP_0\circ\sfE_0=\sfid_{\caE(\tilde V)}~,~~~\sfid_{\caE(V)}-\sfE_0\circ \sfP_0=\sfH_0\circ Q+Q\circ \sfH_0~,
                    \\
                    \sfH_0\circ\sfH_0=0~,~~~\sfH_0\circ \sfE_0=0~,~~~\sfP_0\circ \sfH_0=0~.
                \end{gathered}
            \end{equation}
        \end{subequations}
    \end{theorem}
    \begin{proof}
        The proof of~\eqref{eq:homotopy_transfer_relations_0} is a straightforward computation for elements in $\oslash^2_\bullet V$. The general case follows then by iterating the algebra relations and applying~\eqref{eq:homotopy_transfer_relations_dual}.
    \end{proof}
    
    The higher products $\eps_i$ with $i>1$ on $\frE$ can now be regarded as a perturbation of the differential. Dually, we have a perturbation of $Q$,
    \begin{equation}
        Q\rightarrow \hat Q\coloneqq Q+Q_\delta~.
    \end{equation}
    We can then use the homological perturbation lemma~\cite{Gugenheim1989:aa,gugenheim1991perturbation,Crainic:0403266} to transfer these structures over to higher products $\tilde \eps_i$ on $\tilde \frE$, or, dually, to a perturbed differential $\tilde Q$ on $\caE(\tilde V)$. The formulas for this are the usual ones, cf.~\cite{Crainic:0403266}.
    
    \begin{lemma}\label{prop:HPL}
        Starting from the deformation retract~\eqref{eq:HE_data_Eilh}, we can construct another deformation retract
        \begin{subequations}\label{eq:HE_data_Eilh2}
            \begin{equation}
                \begin{tikzcd}
                    \ar[loop,out=194,in= 166,distance=20,"\sfH"]({\caE(V)},\hat Q)\arrow[r,shift left]{}{\sfP} & ({\caE(\tilde V)},\hat{\tilde Q}) \arrow[l,shift left]{}{\sfE}~,
                \end{tikzcd}
            \end{equation}
            with
            \begin{equation}
                \begin{gathered}
                    \hat Q\coloneqq Q+Q_\delta~,~~~\hat{\tilde Q}=\tilde Q+ \sfP\circ Q_\delta \circ \sfE_0~,
                    \\
                    \sfP\coloneqq \sfP_0\circ(\sfid+Q_\delta\circ\sfH_0)^{-1}~,~~~\sfE\coloneqq \sfE_0-\sfH\circ Q_\delta\circ \sfE_0~,~~~\sfH=\sfH_0\circ(\sfid+Q_\delta\circ \sfH_0)^{-1}~,
                \end{gathered}
            \end{equation}
            where inverse operators are defined via the evident geometric series, such that 
            \begin{equation}\label{eq:homotopy_transfer_relations2}
                \begin{gathered}
                    \sfP\circ \sfE=\sfid_{\caE(\tilde V)}~,~~~
                    \sfid_{\caE(V)}-\sfE\circ \sfP=\sfH\circ \hat Q+\hat Q\circ \sfH~,
                    \\
                    \sfH_0\circ\sfH_0=\sfH_0\circ \sfE=0~,~~~\sfP\circ \sfH=0~,
                \end{gathered}
            \end{equation}
        \end{subequations}
        In particular, $\sfP$ and $\sfE$ are morphisms of $\opEilh_2$-algebras.
    \end{lemma}
    \begin{proof}
        The lemma follows from the usual perturbation lemma, cf.~\cite{Crainic:0403266}, with the specialization that $\sfE$ and $\sfP$ are here morphisms of $\opEilh_2$-algebras. To see this, we note that $Q_\delta$ acts as a derivation,
        \begin{equation}
            Q_\delta(x\oslash_i y)=(-1)^i\big((Q_\delta x)\oslash_i y+(-1)^{|x|}x\oslash_i Q_\delta y\big)~.
        \end{equation}
        Moreover, in the non-vanishing terms of
        \begin{equation}
            \begin{aligned}
                \sfP(x\oslash_i y)&=(\sfP_0\circ(\sfid-Q_\delta \circ \sfH_0+Q_\delta\circ \sfH_0\circ Q_\delta\circ \sfH_0-\ldots)\big)(x\oslash_i y)
                \\
                &=\sfP_0(x)\oslash_i\sfP_0(y)-\sfP_0(Q_\delta(\sfH_0(x))\oslash_i\sfP_0(y)-\sfP_0(x)\oslash_i\sfP_0(Q_\delta(\sfH_0(x)))+\dotsb~,
            \end{aligned}
        \end{equation}
        all the $\sfH_0$ are pre-composed by a $Q_\delta$, as otherwise the map $\sfP_0$, which is applied to all summands, would annihilate the term due to $\sfP_0\circ \sfH_0=0$. The relation
        \begin{equation}
            \sfP(x\oslash_i y)=\sfP(x)\oslash_i\sfP(y)
        \end{equation}
        follows then by a direct computation. The same holds for $\sfE$.
    \end{proof}
    We note that for small perturbations $Q_\delta$, the homological perturbation \ref{prop:HPL} implies that
    \begin{equation}\label{eq:HPL_geom_series}
        \hat{\tilde Q}=\tilde Q+\sfP_0\circ Q_\delta\circ \sfE_0-\sfP_0\circ Q_\delta\circ \sfH_0\circ Q_\delta \circ \sfE_0+\sfP_0\circ Q_\delta\circ \sfH_0\circ Q_\delta\circ \sfH_0\circ Q_\delta \circ \sfE_0-\dotsb~.
    \end{equation}
    
    A direct consequence of homotopy transfer is the existence of minimal models for homotopy algebras. Consider the deformation retract~\eqref{eq:HE_data_diff_complex} with $(\tilde \frE,\tilde \eps_1=0)=H^\bullet_{\eps_1}(\frE)$ the cohomology of $(\frE,\eps_1)$ as well as $\sfp$ and $\sfe$ the projection onto the cohomology and a choice of embedding, respectively. Then the homotopy transfer yields the structure of a quasi-isomorphic $E_2L_\infty$-algebra on the cohomology of $(\frE,\eps_1)$. This implies the minimal model theorem.
    \begin{theorem}\label{thm:minimal_model}
        Any $E_2L_\infty$-algebra $\frE$ comes with a quasi-isomorphic $E_2L_\infty$-algebra structure on its cohomology $H^\bullet_{\eps_1}(\frE)$. We call this a \uline{minimal model} of $\frE$.
    \end{theorem}
    
    \subsection{The relationship between \texorpdfstring{$E_2L_\infty$}{EL-infinity}-algebras and  \texorpdfstring{$L_\infty$}{L-infinity}-algebras}\label{ssec:EL_infty_and_L_infty}
    
    Let us explain the relationship between $E_2L_\infty$-algebras and $L_\infty$-algebras in detail; we will come to examples in \ref{ssec:string_models}.
    
    First, we develop the expected result that the category of $L_\infty$-algebras is a subcategory of the category of $E_2L_\infty$-algebras.
    \begin{proposition}\label{thm:el_infty_contain_l_infty}
        A semistrict $E_2L_\infty$-algebra $\frE$ is an $L_\infty$-algebra. Conversely, any $L_\infty$-algebra is a (semistrict) $E_2L_\infty$-algebra. Dually, the data contained in a differential $Q$ in a semifree $\opEilh_2$-algebra $({\caE(V)},Q)$ with $\rmIm(Q|_V)\subset \bigoslash_0^\bullet V$ is equivalent to the data contained in a differential $\tilde Q$ on the semifree differential graded commutative algebra $(\bigodot^\bullet V,\tilde Q)$.
    \end{proposition}
    \begin{proof}
        It suffices to show the dual statement, which is a direct consequence of \ref{thm:lift_dgca_to_Eilh}.
    \end{proof}
    
    Concretely, given an $L_\infty$-algebra $\frL$ with higher products $\mu_k$ this yields a semistrict $E_2L_\infty$-algebra with higher products
    \begin{equation}
        \eps^I_k=\begin{cases}
            \mu_k & \mbox{for~$|I|=0$}~,
            \\
            0 & \mbox{else}~.
        \end{cases}
    \end{equation}
    Dually, we can embed the Chevalley--Eilenberg algebra $\sfCE(\frL)=(\odot^\bullet \frL[1]^*)$ into the semifree $\opEilh_2$-algebra $(\caE(\frL[1]^*),Q_{\rm EL})$ using the map $\sfE_0$ of~\eqref{eq:embedding_map} in \ref{lem:embedding}. The Chevalley--Eilenberg differential is trivially lifted, and the structure constants agree up to combinatorial factors:
    \begin{equation}
        \begin{aligned}
            Q_{\rm L} t^\alpha&= q^\alpha_\beta t^\beta+\tfrac12 q^\alpha_{\beta\gamma}t^\beta t^\gamma+\tfrac1{3!}q^\alpha_{\beta\gamma\delta}t^\beta t^\gamma t^\delta+\ldots~,
            \\
            Q_{\rm EL} t^\alpha&= q^\alpha_\beta t^\beta+q^\alpha_{\beta\gamma}t^\beta \oslash_0t^\gamma+q^\alpha_{\beta\gamma\delta}(t^\beta\oslash_0 t^\gamma)\oslash_0 t^\delta+\ldots~.
        \end{aligned}
    \end{equation}
    
    Conversely, any semistrict $E_2L_\infty$-algebra $\frE$ is an $L_\infty$-algebra with higher products $\mu_k=\eps^0_k$. As an example, consider the $\vartheta=0$ case of the family of weak string Lie 2-algebra models~\eqref{eq:skeltal_model_string_Lie_2_algebra}. This is a semistrict $E_2L_\infty$-algebra and therefore an $L_\infty$-algebra.
    
    \begin{proposition}\label{thm:lift_L_infty_morphism}
        Any $L_\infty$-algebra morphism $\phi:\frL\rightarrow \tilde \frL$ lifts to an $E_2L_\infty$-algebra morphism $\hat \phi:\hat \frL\rightarrow \hat{\tilde \frL}$, where $\hat \frL$ and $\hat{\tilde \frL}$ are the $L_\infty$-algebras $\frL$ and $\tilde \frL$, regarded as $E_2L_\infty$-algebras.
    \end{proposition}
    \begin{proof}
        We prove this statement from the dual perspective. Let $(\bigodot^\bullet V,Q)$ and $(\bigodot^\bullet \tilde V,\tilde Q)$ be the Chevalley--Eilenberg algebras of $\frL$ and $\tilde \frL$, respectively.
        
        We have the embeddings $\sfE_0$ and $\tilde \sfE_0$ into the semifree $\opEilh_2$-algebras $(\caE(\frL[1]^*),\hat Q)$ and $(\caE(\tilde \frL[1]^*),\hat{\tilde Q})$, as defined in~\eqref{eq:embedding_map}. Furthermore, we also have the projectors $\sfP_0$ and $\tilde \sfP_0$ as defined in~\eqref{eq:projector_P0}. The Chevalley--Eilenberg algebras of $\hat \frL$ and $\hat{\tilde \frL}$ are then
        \begin{equation}
            (\caE(V),\hat Q\coloneqq Q_0+\sfE_0\circ Q\circ \sfP_0)
            \eand 
            (\caE(\tilde V),\hat{\tilde Q}\coloneqq \tilde Q_0+\tilde\sfE_0\circ \tilde Q\circ \tilde\sfP_0)~,
        \end{equation}
        cf.~\ref{thm:lift_dgca_to_Eilh}. The dual of the morphism $\phi$,
        \begin{equation}
            \Phi:\sfCE(\tilde \frL)\rightarrow \sfCE(\frL)~,
        \end{equation}
        trivially lifts to the following dual of an $E_2L_\infty$-algebra morphism 
        \begin{equation}
            \hat \Phi:\sfCE(\hat{\tilde \frL})\rightarrow \sfCE(\hat\frL)\ewith\hat \Phi(v)\coloneqq (\sfE_0\circ\Phi\circ\tilde \sfP_0)(v)~,
        \end{equation}
        and we note that $\hat \Phi \circ \tilde \sfE_0=\sfE_0\circ \Phi$. It then follows that 
        \begin{equation}
            \begin{aligned}
                (\hat \Phi\circ \hat{\tilde Q})(v)&=(\sfE_0\circ \Phi\circ \tilde \sfP_0\circ \tilde \sfE_0\circ \tilde Q\circ \tilde \sfP_0)(v)
                \\
                &=(\sfE_0\circ Q\circ \sfP_0\circ \sfE_0\circ \Phi\circ \tilde \sfP_0)(v)
                \\
                &=(\hat Q\circ hat \Phi)(v)
            \end{aligned}
        \end{equation}
        and $\hat\Phi$ is the dual to the desired morphism of $E_2L_\infty$-algebras $\hat \phi$.
    \end{proof}
    
    We now come to the inverse rectification theorem, which generalizes\footnote{for 2-term $EL_\infty$-algebras with symmetric alternator} the result of~\cite{Roytenberg:0712.3461} for 2-term $EL_\infty$-algebras.
    \begin{theorem}\label{thm:antisym}
        Any $E_2L_\infty$-algebra $(\frE,\eps_i)$ induces an $L_\infty$-algebra structure on the graded vector space $\frE$. This $L_\infty$-algebra structure is induced by homotopy transfer using the homotopy~\eqref{eq:homotopy_for_transfer_EL_to_L}.
    \end{theorem}
    \begin{proof}
        The proof is readily obtained by applying the homological perturbation lemma to the contracting homotopy 
        \begin{equation}
            \begin{tikzcd}
                \ar[loop,out=204,in= 156,distance=70,"\sfH_0"](\caE(V),Q_0+Q_1+Q_\delta)\arrow[r,shift left]{}{\sfP_0} & (\odot^\bullet V,Q_{\rm L}) \arrow[l,shift left]{}{\sfE_0}~,
            \end{tikzcd}
        \end{equation}
        cf.~\eqref{eq:homotopy_for_transfer_EL_to_L}. Consider the Chevalley--Eilenberg algebra $\sfCE(\frE)$ of $\frE$, and split the differential $Q=Q_0+Q_1+Q_\delta$ into $Q_0$, a linear part $Q_1$, and a perturbation $Q_\delta$. Then homotopy transfer yields a differential
        \begin{equation}\label{eq:HPL_formula_Q_EL_to_L}
            Q_{\rm L}=Q_1+\sfP_0\circ Q_\delta\circ\sfE_0-\sfP_0\circ Q_\delta\circ \sfH_0\circ Q_\delta\circ \sfE_0+\sfP_0\circ Q_\delta\circ \sfH_0\circ Q_\delta\circ \sfH_0\circ Q_\delta\circ \sfE_0+\ldots
        \end{equation}
        on $\odot^\bullet(\frE[1]^*)$. By construction, $Q_{\rm L}^2=0$.
        Moreover, $Q_{\rm L}$ satisfies the Leibniz rule on $\odot^\bullet(\frE[1]^*)$: the deformation terms in the Leibniz rule~\eqref{eq:def_Leibniz} are graded antisymmetric, and this graded antisymmetry is preserved by subsequent applications of $\sfH$ and $Q_\delta$. The final projector $\sfP_0$ then eliminates these terms.
    \end{proof}
    As an example, we can compute the antisymmetrization of a 2-term $E_2L_\infty$-algebra and reproduce\footnote{Due to different conventions, there is a relative minus sign between our $\mu_3$ and that of~\cite{Roytenberg:0712.3461}.}~\cite[Proposition 3.1]{Roytenberg:0712.3461}.
    \begin{corollary}
        Given a 2-term $E_2L_\infty$-algebra $\frE$, there is an $L_\infty$-algebra structure on the graded vector space $\frE$ with higher products
        \begin{equation}\label{eq:antisymmetrization_2-term}
            \begin{aligned}
                \mu_1(y)&\coloneqq \eps_1(y)~,\\
                \mu_2(x_1,x_2)&\coloneqq \tfrac12(\eps^0_2(x_1,x_2)-\eps^0_2(x_2,x_1))~,\\
                \mu_2(x_1,y)=-\mu_2(y,x_1)&\coloneqq \tfrac12(\eps^0_2(x_1,y)-\eps^0_2(y,x_1))~,\\
                \mu_3(x_1,x_2,x_3)&\coloneqq \tfrac{1}{3!}\sum_{\sigma\in S_3}\chi(\sigma;x_1,x_2,x_3)\Big(\eps^{00}_3(x_{\sigma(1)},x_{\sigma(2)},x_{\sigma(3)})\\
                &\hspace{4.5cm}+\tfrac12 \eps^1_2(\eps^0_2(x_{\sigma(1)},x_{\sigma(2)}),x_{\sigma(3)})\Big)
            \end{aligned}
        \end{equation}
        for all $x\in \frE_0$ and $y\in \frE_{-1}$.
    \end{corollary}
    \begin{proof}
        We start from the Chevalley--Eilenberg algebra $\sfCE(\frE)$. As always, we assume for convenience that there is a basis, explicitly given by elements $r^\alpha$, $s^a$ of degrees $1$ and $2$ of $\frE[1]^*$. The differential then reads as
        \begin{equation}\label{eq:CE_differential_2-term}
            \begin{aligned}
                Q r^a&=-m^a_a s^a-m^{0,a}_{bc}\,r^b \oslash_0 r^c~,
                \\
                Q s^i&=-m^{0,i}_{a j}\,r^a\oslash_0 s^j+m^{0,a}_{ja}\,s^j\oslash_0 r^a\\
                &~~~+m^{00,i}_{abc}\,(r^a\oslash_0 r^b) \oslash_0 r^c-m^{1,i}_{ab}\,r^a\oslash_1 r^b~,
            \end{aligned}
        \end{equation}
        cf.~\eqref{eq:structure_constants_ELinfty}. We then evaluate formula~\eqref{eq:HPL_formula_Q_EL_to_L} for the homotopy $\sfH_0$:
        \begin{equation}
            \begin{aligned}
                Q_{\rm L} r^a &=-m^a_i s^i-\tfrac12 m^{0,a}_{bc}\,r^b \odot r^c~,
                \\
                Q_{\rm L} s^i&=-(m^{0,i}_{a j}+m^{0,i}_{ja})r^a \odot s^j+\tfrac1{3!}(m^{00,i}_{abc}+\tfrac12 m^{1,i}_{dc}m^{0,d}_{ab})r^a\odot r^b\odot r^c~.
            \end{aligned}
        \end{equation}
        This is the differential for the Chevalley--Eilenberg algebra of $(\frE,\mu_i)$ with the higher products~\eqref{eq:antisymmetrization_2-term}.
    \end{proof}
    
    As observed already in~\cite{Roytenberg:0712.3461}, the antisymmetrization map is functorial for 2-term $E_2L_\infty$-algebras, and any morphism of $E_2L_\infty$-algebra induces a morphism between the corresponding antisymmetrized $L_\infty$-algebras.
    
    A new result is that this antisymmetrization map lifts indeed to a quasi-isomorphism of $E_2L_\infty$-algebras. 
    \begin{theorem}\label{thm:quasi-iso_of_antisym}
        Let $\frE$ be an $E_2L_\infty$-algebra, and let $\frE'$ be the $L_\infty$-algebra induced by \ref{thm:antisym}, regarded as an $E_2L_\infty$-algebra. Then there is a quasi-isomorphism $\phi:\frE\rightarrow \frE'$.
    \end{theorem}
    \begin{proof}
        We prove this statement again using the Chevalley--Eilenberg algebras $\sfCE(\frE)$ and $\sfCE(\frE')$. Note that as graded vector spaces, $\frE[1]^*=\frE'[1]^*$. If $Q=Q_0+Q_1+Q_\delta$ is the differential on $\sfCE(\frE)$, then the differential on $\sfCE(\frE')$ reads as
        \begin{equation}\label{eq:lifted_differential}
            Q'v=Q_1v+\sfE_0\circ \sfP_0\circ (\sfid+Q_\delta\circ \sfH_0)^{-1}\circ Q_\delta v
        \end{equation}
        for all $v\in \frE[1]^*$. We need to construct an invertible $\opEilh_2$-algebra morphism $\Phi:\caE(\frE[1]^*)\rightarrow \caE(\frE[1]^*)$ satisfying $Q\Phi=\Phi Q'$. The desired morphism on $\frE[1]^*$ is\footnote{This morphism implements a coordinate transformation such that the image of $Q$ on $\tilde v=\Phi(v)$ has no component in the subspace $Q_0\sfH_0\caE(\frE[1]^*)$. This then implies that it has no component in $\sfH_0 Q_0\caE(\frE[1]^*)$ either. The only remaining component is in $\sfE_0\sfP_0\caE(\frE[1]^*)$, which implies that $Q$ is the Chevalley--Eilenberg differential of an $L_\infty$-algebra, trivially regarded as an $E_2L_\infty$-algebra.}
        \begin{equation}
            \Phi(v)=(\sfid-\sfH_0\circ Q_\delta+\sfH_0\circ Q_\delta\circ \sfH_0\circ Q_\delta-\ldots)(v)=(\sfid+\sfH_0\circ Q_\delta)^{-1}(v)~,
        \end{equation}
        and using 
        \begin{equation}
            \begin{aligned}
                Q_0\circ \sfH_0&=\sfid-\sfE_0\circ \sfP_0-\sfH_0\circ Q_0~,
                \\
                Q_0Q_\delta&=-Q_\delta^2-Q_\delta Q_0~,
            \end{aligned}
        \end{equation}
        one readily verifies that $Q\Phi v=\Phi Q'v$ for all $v\in \frE[1]^*$, which is sufficient. Since the morphism is clearly invertible, this is a quasi-isomorphism.
    \end{proof}
    
    We can now combine \ref{thm:quasi-iso_of_antisym}, \ref{thm:lift_L_infty_morphism} as well as the strictification theorem for $L_\infty$-algebras to obtain the following.
    \begin{corollary}\label{thm:strictification}
        Any $E_2L_\infty$-algebra is quasi-isomorphic to a differential graded Lie algebra, trivially regarded as an $E_2L_\infty$-algebra.
    \end{corollary}
    More directly, this follows from the strictification theorem for generic homotopy algebras, see e.g.~\cite[Proposition 11.4.9]{Loday:2012aa}.
    
    We note that this expected and natural theorem does not hold for the 2-term $EL_\infty$-algebras of~\cite{Roytenberg:0712.3461}, since the classification of these 2-term $EL_\infty$-algebras is generally larger than that of $L_\infty$-algebras~\cite[Theorem 4.5]{Roytenberg:0712.3461}.
    
    A consequence of the strictification theorem and homotopy transfer is the following. Just as for $\ophLie_2$-algebras, we can also tensor an $E_2L_\infty$-algebra by a differential graded commutative algebra\footnote{One may be tempted to replace the differential graded commutative algebra with an $\opEilh_2$-algebra, but already the product between an $\opEilh_2$-algebra and an $\ophLie_2$-algebra does {\em not} carry a natural $\ophLie_2$-algebra structure.}:
    \begin{theorem}
        The tensor product of an $E_2L_\infty$-algebra and a differential graded commutative algebra carries a natural $E_2L_\infty$-algebra structure.
    \end{theorem}
    \begin{proof}
        We can invoke the argument presented in~\cite{Borsten:2021hua} for the existence of general tensor products between certain homotopy algebras. That is, by \ref{thm:strictification}, $\frE$ is quasi-isomorphic to a hemistrict $E_2L_\infty$-algebra $\frE^{\rm hst}$, and the chain complexes $\frA\otimes \frE^{\rm hst}$ and $\frA\otimes \frE$ are quasi-isomorphic. By \ref{prop:tensor_product_dgca_hLie}, $\frA\otimes \frE^{\rm hst}$ carries an $\ophLie_2$-algebra structure, and the homological perturbation lemma allows us to perform a homotopy transfer from $\frA\otimes \frE^{\rm hst}$ to $\frA\otimes \frE$, leading to the desired $E_2L_\infty$-algebra structure.
    \end{proof}
    Instead of using the above elegant but abstract argument, we can also perform a direct computation in the dual Chevalley--Eilenberg picture. This leads to the following explicit formulas for the tensor product $\hat \frE=\frA\otimes \frE$ of a differential graded commutative algebra $\frA$ and an $E_2L_\infty$-algebra $\frE$:
    \begin{equation}
        \begin{gathered}
            \hat \frE\coloneqq \frA\otimes \frE=\oplus_{k\in \IZ}(\frA\otimes \frE)_k~,~~~(\frA\otimes \frE)_k=\sum_{i+j=k} \frA_i\otimes \frE_j~,
            \\
            \hat\eps_1(a_1\otimes x_1)=(\rmd a_1)\otimes x_1+(-1)^{|a_1|}a_1\otimes \eps_1(x_1)~,
            \\
            \hat\eps^I_k(a_1\otimes x_1,\ldots,a_k\otimes x_k)=(-1)^{|I|(|a_1|+\ldots+|a_k|)}(a_1\ldots a_k)\otimes \eps^I_k(x_1,\ldots,x_k)~.
        \end{gathered}
    \end{equation}
    
    \subsection{Examples: String Lie algebra models}\label{ssec:string_models}
    
    Let us illustrate the above structure theorems using the important and archetypal examples of 2-term $E_2L_\infty$-algebra models for the string Lie algebra. We have already encountered the $E_2L_\infty$-algebras $\frstring^{\rmwk,\vartheta}_\rmsk(\frg)$ in~\eqref{eq:skeltal_model_string_Lie_2_algebra}. A short computation using formulas~\eqref{eq:antisymmetrization_2-term} shows that these all antisymmetrize to the following 2-term $L_\infty$-algebra:
    \begin{equation}\label{eq:skeletal_string}
        \begin{gathered}
            \frstring_\rmsk(\frg)~\coloneqq ~(\IR\xrightarrow{~0~}\frg)~,
            \\
            \mu_2(x_1,x_2)=[x_1,x_2]~,~~~\mu_3(x_1,x_2,x_3)=(x_1,[x_2,x_3])~.
        \end{gathered}
    \end{equation}
    It turns out that this $L_\infty$-algebra (which is a minimal model for its quasi-isomorphism class) is quasi-isomorphic to a strict one,
    \begin{equation}
        \begin{gathered}
            \frstring_\rmlp(\frg)~\coloneqq ~(L_0\frg\oplus\IR\frg\xrightarrow{~\mu_1~}P_0\frg)~,
            \\
            \mu_1(\lambda,r)=\lambda~,
            \\
            \mu_2(\gamma_1,\gamma_2)=[\gamma_1,\gamma_2]~,~~~\mu_2(\gamma_1,(\lambda,r))=([\gamma_1,\lambda],2\int_0^1\rmd \tau~(\dot \gamma_1,\lambda)~,
            \\
            \mu_3(\gamma_1,\gamma_2,\gamma_3)=0~,
        \end{gathered}
    \end{equation}
    where $L_0\frg$ and $P_0\frg$ are the based path and based loop spaces of $\frg$, cf.~\cite{Baez:2005sn}. There are two quasi-isomorphisms,
    \begin{equation}
        \begin{tikzcd}
            \frstring_\rmsk(\frg) \arrow[r,bend left=10]{}{\psi} & \frstring_\rmlp(\frg) \arrow[l,bend left=10]{}{\phi}
        \end{tikzcd}~,
    \end{equation}
    and their explicit forms are found e.g.~in~\cite{Saemann:2019dsl}. This implies that there is a quasi-isomorphic family of $E_2L_\infty$-algebras that antisymmetrize to $\frstring_\rmlp(\frg)$, which is readily found:
    \begin{equation}
        \begin{gathered}
            \frstring^{\rmwk,\vartheta}_\rmlp(\frg)~\coloneqq ~(L_0\frg\oplus \IR\xrightarrow{~\eps_1~}P_0\frg)~,
            \\
            \eps_1(\lambda,r)=\lambda~,
            \\
            \eps^0_2(\gamma_1,\gamma_2)=[\gamma_1,\gamma_2]~,~~~\eps^0_2(\gamma_1,(\lambda,r))=\left([\gamma_1,\lambda],2\int_0^1\rmd \tau~(\dot \gamma_1,\lambda)\right)~,
            \\
            \eps_2^1(\gamma_1,\gamma_2)=\big(0,2\vartheta(\dpar\gamma_1,\dpar\gamma_2)\big)
            \\
            \eps^{00}_3(\gamma_1,\gamma_2,\gamma_3)=\vartheta(\dpar \gamma_1,[\dpar \gamma_2,\dpar \gamma_3])~.
        \end{gathered}
    \end{equation}
    
    Altogether, we can summarize the situation in the following commutative diagram:
    \begin{equation}
        \begin{tikzcd}
            \frstring_\rmsk^{\rmwk,\vartheta}(\frg) \arrow[r,bend left=10]{}{\hat\psi} \arrow[d,"\rm asym"] & \frstring_\rmlp^{\rmwk,\vartheta}(\frg) \arrow[d,"\rm asym"] \arrow[l,bend left=10]{}{\hat\phi}
            \\
            \frstring_\rmsk(\frg) \arrow[r,bend left=10]{}{\psi} & \frstring_\rmlp(\frg) \arrow[l,bend left=10]{}{\phi}
        \end{tikzcd}
    \end{equation}
    The morphisms ${\rm asym}$ are special cases of the antisymmetrization map~\eqref{eq:antisymmetrization_2-term}, and the morphisms $\hat\phi$ and $\hat\psi$ are formed by lifts of the morphisms $\phi$ and $\psi$ as given by~\ref{thm:lift_L_infty_morphism}.
    
    Generically, on top of every $L_\infty$-algebra, there is a family of $E_2L_\infty$-algebras that antisymmetrize to it. The additional structure constants contained in the alternators of the $E_2L_\infty$-algebra will turn out to provide an important information for the construction of local connection forms on higher principal bundles.
    
    \section{Relations to other algebras}
    
    In the following, we explain the relation between $E_2L_\infty$-algebras and homotopy Leibniz algebras and, in particular, to differential graded Lie algebras. The latter prepares our interpretation of generalized geometry and the tensor hierarchies.    
    
    \subsection{Relation to homotopy Leibniz algebras}
    
    Just as Lie algebras are Leibniz algebras that happen to have an antisymmetric Leibniz bracket, $E_2L_\infty$-algebras are $\opLeib_{\infty}$-algebras whose higher Leibniz brackets are antisymmetric up to a homotopy. Homotopy Leibniz algebras were defined in~\cite{Ammar:0809.4328,Khudaverdyan:2013cta}, and they are the homotopy algebras over the Zinbiel operad $\opZinb$~\cite{JSTOR:24491899,Ginzburg:0709.1228} which, as suggested by the name\footnote{This nomenclature is a successful joke suggested by J.~M.~Lemaire. Zinbiel algebras are also known as {\em (commutative) shuffle algebras}, and the free Zinbiel algebra over a vector space is the shuffle algebra on its tensor algebra.}, is the Koszul-dual to the Leibniz operad $\opLeib$. 
    
    Explicitly, consider the semifree non-associative tensor algebra $\oslash_0^\bullet V$ for a graded vector space $V$ with only the first relation of~\eqref{eq:Eilh-relations} imposed. A (nilquadratic) differential $Q$ on this algebra which satisfies the ordinary Leibniz rule then defines a homotopy Leibniz algebra. All the additional structure in $\opEilh_2$ (as well as the resulting additional structure in $E_2L_\infty$-algebras) capture the appropriate notion of symmetry up to homotopy of the higher Leibniz brackets.
    
    Ordinary Leibniz algebras form an interesting source of 2-term $\ophLie_2$-algebras, which had been observed before:
    \begin{proposition}[\cite{Roytenberg:0712.3461}]\label{prop:Leib_is_hemistrict_ELinfty}
        Any Leibniz algebra induces canonically a hemistrict 2-term $E_2L_\infty$-algebra concentrated in degrees $-1$ and $0$.
    \end{proposition}   
    \noindent Explicitly, let $\frg$ be a Leibniz algebra, and write $\frg^{\rm ann}=[\frg,\frg]$. Then 
    \begin{equation}
        \frE(\frg)=\big(\frE(\frg)_{-1}\xrightarrow{~\eps_1~}\frE(\frg)_{0}\big)\coloneqq \big(\frg^{\rm ann} \xhookrightarrow{~~~~} \frg\big)
    \end{equation}
    is a differential graded Leibniz algebra, and we promote it to a 2-term $E_2L_\infty$-algebra by 
    \begin{equation}
        \sfalt(e_1,e_2)\coloneqq [e_1,e_2]+[e_2,e_1]\in \frg^{\rm ann}
    \end{equation}
    for all $e_1,e_2\in \frg$.\footnote{We note that this result, together with \ref{thm:antisym}, immediately implies that any Leibniz algebra gives rise to a 2-term $L_\infty$-algebra as shown separately in~\cite{Sheng:2015:1-5}.}

    \subsection{\texorpdfstring{$\ophLie_2$}{hLie2}-algebras from differential graded Lie algebras and derived brackets}\label{ssec:antisym_hLie}
    
    Given a differential graded Lie algebra $\frg=\bigoplus_{k\in \IZ}\frg_k$, one can construct an associated $L_\infty$-algebra on the grade-shifted partial complex $\frL=\bigoplus_{k\leq 0}\frg[1]$. As explained in~\cite{Getzler:1010.5859}, this is a corollary to the result of~\cite{Fiorenza:0601312} that the mapping cone of a morphism between two differential graded Lie algebras carries a natural $L_\infty$-algebra structure. In this section, we present a refinement of this associated $L_\infty$-algebra to an $\ophLie_2$-algebra. The existence of the $L_\infty$-algebra is then a corollary to the antisymmetrization \ref{thm:antisym}. Our construction extends the construction of $\opLeib_\infty$-algebras from $\opLeib$-algebras in~\cite{Uchino:0902.0044} as well as the construction of 2-term $E_2L_\infty$-algebras from 3-term differential graded Lie algebras in~\cite{Roytenberg:0712.3461}. 
    
    Given a differential graded Lie algebra, we readily construct a grade-shifted $\ophLie_2$-algebra.
    \begin{theorem}\label{thm:ophLie_from_dgLA}
        Given a differential graded Lie algebra $(\frg,\rmd,\{-,-\})$ with $\frg=\bigoplus_{k\in \IZ}\frg_k$, we have an associated $\ophLie_2$-algebra 
        \begin{equation}
            \frE=\bigoplus_{k\leq 0} \frE_k~,~~~\frE_k=\frg_{k-1}
        \end{equation}
        with higher products
        \begin{equation}
            \begin{aligned}
                \eps_1(x_1)&\coloneqq \begin{cases}
                    \rmd_\frg x~&\mbox{for}~|x|_\frE<0~,
                    \\
                    0 &\mbox{else}~,
                \end{cases}
                \\
                \eps_2^i(x_1,x_2)&\coloneqq \begin{cases}
                    \{\delta x_1,x_2\} & \mbox{for}~i=0~,
                    \\
                    (-1)^{|x_1|_\frE}\{x_1,x_2\} & \mbox{for}~i=1~,
                    \\
                    0& \mbox{else}
                \end{cases}
            \end{aligned}
        \end{equation}
        for all $x_1,x_2\in \frE$. Here, $\delta\coloneqq \rmd_{\frg}|_{\frg_{-1}}$ and $|x_1|_\frE$ denotes the degree of $x_1$ in $\frE$.
    \end{theorem}
    \begin{proof}
        The proof is a straightforward verification of the axioms of an $\ophLie_2$-algebra~\eqref{eq:hLie-relations}.
    \end{proof}
    
    Let us discuss the explicit form of the antisymmetrization in some more detail. We assume, as usual, that $\frE$ admits a nice basis $(\tau_\alpha)$, so that $\frE[1]^*$ has a dual basis $(t^\alpha)$. The Chevalley--Eilenberg differential then reads as
    \begin{equation}\label{eq:Q_hLie2}
        Q t^\alpha=-(-1)^{|\beta|}m^\alpha_\beta t^\beta-(-1)^{i(|\beta|+|\gamma|)+|\gamma|(|\beta|-1)}\,m^{i,\alpha}_{\beta\gamma}\,t^\beta \oslash_i t^\gamma~,
    \end{equation}
    and we have the following theorem.
    \begin{theorem}\label{thm:antisym_hLie}
        For each $\ophLie_2$-algebra $(\frE,\eps^i_j)$ (with the above mentioned restrictions), there is an $L_\infty$-algebra $(\frE,\mu_i)$ with first four higher products reading as
        \begin{equation}\label{eq:antisymmetrization_hLie}
            \begin{aligned}
                \mu_1(x_1)&\coloneqq \eps_1(x_1)~,\\
                \mu_2(x_1,x_2)&\coloneqq \tfrac12(\eps^0_2(x_1,x_2)-\eps^0_2(x_2,x_1))~,\\
                \mu_3(x_1,x_2,x_3)&\coloneqq \tfrac{1}{3!}\sum_{\sigma\in S_3}\chi(\sigma;x_1,x_2,x_3)\Big(\eps^0_3(x_{\sigma(1)},x_{\sigma(2)},x_{\sigma(3)})\\
                &\hspace{2.0cm}+\tfrac14\big( \eps^1_2(\eps^0_2(x_{\sigma(1)},x_{\sigma(2)}),x_{\sigma(3)})
                +\eps^1_2(x_{\sigma(1)},\eps^0_2(x_{\sigma(2)},x_{\sigma(3)}))\big)\Big)~,
                \\
                \mu_4(x_1,x_2,x_3,x_4)&\coloneqq 0~,
            \end{aligned}
        \end{equation}
        for all $x_i\in \frE$.
    \end{theorem}
    \begin{proof}   
        We use again \ref{thm:antisym} and determine the Chevalley--Eilenberg differential~\eqref{eq:HPL_formula_Q_EL_to_L} of $(\frE,\mu_i)$  using the homotopy~$\sfH_0$ from~\eqref{eq:homotopy_H0}, which allows us to compute $Q_{\rm L}$ up to quartic order. This produces the higher products~\eqref{eq:antisymmetrization_hLie} for $\alpha_1=\alpha_2=0$.
    \end{proof}
    
    We note that our choice $\alpha_1=\alpha_2=0$ is, in fact, not the most natural one. One gets a nicer pattern in the expressions for $\sfH_0$ if one puts $\alpha_1=\alpha_2=-\tfrac14$, and this results in an expression for $\mu_4$ which does not vanish but involves nestings of two maps $\eps_2^0$ and one map $\eps_2^2$. In the
    case of $\ophLie_2$-algebras obtained from differential graded Lie algebras, we have $\eps^2_2=0$, and therefore the distinction is irrelevant.
    
    We can now compose the map from differential graded Lie algebras to $\ophLie_2$-algebras with the antisymmetrization~\ref{thm:antisym}. This reproduces the following proposition of~\cite{Getzler:1010.5859}, which in turn is a specialization of~\cite{Fiorenza:0601312}:
    \begin{proposition}\label{prop:dgLA_to_L_infty}
        Given a differential graded Lie algebra $(\frg,\rmd,[-,-])$, we have an $L_\infty$-algebra structure on the truncated complex
        \begin{equation}
            \frg_{\leq 0}=\big(~~
            \ldots~\xrightarrow{~\rmd~}~\frg_{-2}~\xrightarrow{~\rmd~}~\frg_{-1}~\xrightarrow{~\rmd~}~\frg_{0}~\xrightarrow{~0~}~*~\xrightarrow{~0~}~\ldots~~)
        \end{equation}
        with
        \begin{equation}
            \begin{aligned}
                \mu_1(x_1)&=\begin{cases}
                    \rmd x_1 & \mbox{for $|x_1|<0$}~,
                    \\
                    0 & \mbox{for $|x_1|=0$}
                \end{cases}
                \\
                \mu_k(x_1,\ldots,x_k)&=\frac{(-1)^k}{(k-1)!}B_{k-1}\sum_{\sigma\in S_{k}}\chi(\sigma;x_1,\ldots,x_k)[[\ldots[[\delta x_{\sigma(1)}],x_{\sigma(2)}],\ldots],x_{\sigma(k)}]~,
            \end{aligned}
        \end{equation}
        where
        \begin{equation}
            \delta(x_1)=\begin{cases}
                \rmd x_1 & \mbox{for $|x_1|=0$}~,
                \\
                0 & \mbox{else}
            \end{cases}
        \end{equation}
        for all $x_i\in \frg_{\leq 0}$. Here, $B_k$ are the (first) Bernoulli numbers\footnote{i.e.\ $B_0,B_1,\ldots=1,-\tfrac12,\tfrac{1}{6},0,-\tfrac{1}{30},0,\tfrac{1}{42},\ldots$}.
    \end{proposition}
    
    Altogether, our above constructions suggest the following picture:
    \begin{equation}\label{eq:diagram_algebras}
        \begin{tikzcd}
            \mbox{dg Lie algebra} \arrow[r,"\text{\Cref{thm:ophLie_from_dgLA}}", bend left=30] \arrow[rr,"\text{\Cref{prop:dgLA_to_L_infty}}",swap, bend right=20]& \mbox{$\ophLie_2$-algebra} \arrow[r,"\text{\Cref{thm:antisym_hLie}}", bend left=30] & \mbox{$L_\infty$-algebra}
        \end{tikzcd}
    \end{equation}
    Our formulas~\eqref{eq:HPL_formula_Q_EL_to_L} show that this picture is true at least for differential graded Lie algebras concentrated in degrees $d\geq -3$.
    
    From \ref{prop:dgLA_to_L_infty} it is also clear that $\mu_4$ in~\eqref{eq:antisymmetrization_hLie} vanishes because $B_3=0$. Similarly, all even higher brackets $\mu_{2i}$ with $i\geq 1$ vanish, as the odd Bernoulli numbers $B_k$ for $k\geq 3$ vanish.
    
    As a simple example, consider a quadratic Lie algebra $\frg$, and construct the differential graded Lie algebra
    \begin{equation}
        \frG=(~\ldots\xrightarrow{~0~}*\xrightarrow{~0~}\underbrace{\IR}_{\frG_{-2}}\xrightarrow{~0~}\underbrace{\frg}_{\frG_{-1}}\xrightarrow{~\sfid~}\underbrace{\frg}_{\frG_{0}}\xrightarrow{~0~}*\xrightarrow{~0~}\ldots~)~,
    \end{equation}
    concentrated in degrees $-2$, $-1$, $0$ with differential and Lie brackets
    \begin{equation}
        [x_1,x_2]_\frG=2(x_1,x_2)~,~~~[y_1,x_1]_\frG=-[x_1,y_1]_\frG=y_1(x_1)~,~~~[y_1,y_2]_\frG=(y_1,y_2)
    \end{equation}
    for all $x_1,x_2\in \frG_0\cong\frg$ and $y_1,y_2\in \frG_{-1}\cong \frg$, where $[-,-]$ and $(-,-)$ are the Lie bracket and the Cartan--Killing form on $\frg$. Then the associated $\ophLie_2$-algebra is 
    \begin{equation}
        \begin{gathered}
            \frE=(~\IR\xrightarrow{~0~}\frg~)~,
            \\
            \eps_1(r)\coloneqq 0~,
            \\
            \eps^0_2(x_1,x_2)=[x_1,x_2]~,~~~\eps^1_2(x_1,x_2)=2(x_1,x_2)~.
        \end{gathered}
    \end{equation}
    We thus recover the hemistrict $E_2L_\infty$-algebra $\frstring^{\rm wk,1}_{\rm sk}(\frg)$ introduced in \ref{ssec:string_models}. The antisymmetrization of this $\ophLie_2$-algebra then yields the skeletal string Lie 2-algebra model $\frstring_{\rm sk}(\frg)$. Interestingly, a quick consideration of the case leads to the conclusion that there is no differential graded Lie algebra that reproduces the strict string Lie 2-algebra model $\frstring_\rmlp(\frg)=\frstring^{\rmwk,0}_\rmlp(\frg)$. This points towards a possible extension of \ref{thm:ophLie_from_dgLA} producing $E_2L_\infty$-algebras from certain $L_\infty$-algebras.
    
    \section{Generalized and multisymplectic geometry}
    
    We now come to our first two applications of $E_2L_\infty$-algebras, or rather $\ophLie_2$-algebras: the symmetry algebras of symplectic $L_\infty$-algebroids and, in a closely related way, a categorified version of higher Poisson algebras.

    \subsection{Generalized geometry from symplectic \texorpdfstring{$L_\infty$}{L-infinity}-algebroids}\label{ssec:symplectic_L_infty_algebroids}
    
    The string Lie 2-algebra $\frstring_{\rm sk}(\frspin(n))$ is a finite-dimensional $L_\infty$-subalgebra of the 2-term $L_\infty$-algebra of symmetries associated to the Courant algebroid~\cite{Baez:2009:aa} over $\sfSpin(n)$. It is therefore not surprising that the symmetries of symplectic $L_\infty$-algebroids are important sources for examples of $E_2L_\infty$-algebras. This link was noticed before in~\cite{Roytenberg:0712.3461} and~\cite{Dehling:1710.11104}, where 2- and 3-term $E_2L_\infty$-algebras were constructed. Here, we can present the general picture. We shall follow the conventions of~\cite{Deser:2016qkw}.  
    
    \begin{theorem}
        The symmetry algebra of a symplectic $L_\infty$-algebroid\footnote{cf.~\cite{Deser:2016qkw} for a definition} is naturally an $\ophLie_2$-algebra.
    \end{theorem}
    \begin{proof}
        The Chevalley--Eilenberg algebra of a symplectic $L_\infty$-algebroids $\sfL$ is a differential graded Lie algebra. The differential is the Chevalley--Eilenberg differential, encoding the anchor and the higher maps on sections of $\sfL$, and it is given by a vector field $Q$ on $\sfL$. The Lie bracket is the Poisson bracket induced by the symplectic form $\omega$. Compatibility of the differential with the Lie bracket amounts to the condition $\caL_Q\omega=0$, which is part of the definition of a symplectic $L_\infty$-algebroid. The $\ophLie_2$-algebra of this differential graded algebra is a refined version of the symmetry algebra of the $L_\infty$-algebroid, which is the $L_\infty$-algebra obtained from the original differential graded Lie algebra via \ref{prop:dgLA_to_L_infty}.
    \end{proof}
    
    This theorem explains the interest in extension of Leibniz algebras in the context of generalized geometry and double field theory. The generalized tangent bundles used there are indeed symplectic $L_\infty$-algebroids (or symplectic pre-N$Q$-manifolds, as explained in~\cite{Deser:2016qkw}). Therefore, the relevant symmetry algebras are $\ophLie_2$-algebras, and the most prominently visible feature of them in all construction is their Leibniz brackets $\eps_2^0$.
    
    As a short example, let us work out the case of Vinogradov Lie $n$-algebroids, which generalize the Courant algebroid. The latter case, i.e.~the case $n=2$, was sketched in~\cite[Example 5.4]{Roytenberg:0712.3461}. The Vinogradov Lie $n$-algebroids are given as the graded vector bundles
    \begin{equation}
        \caV_n(M)\coloneqq T^*[n]T[1]M
    \end{equation}
    over some manifold $M$. We introduce local coordinates $x^\mu$ on the base $M$ and extend these to Darboux coordinates $(x^\mu,\xi^\mu,\zeta_\mu,p_\mu)$ of degrees $0,1,n-1,n$, leading to the canonical symplectic form
    \begin{equation}
        \omega=\rmd x^\mu\wedge\rmd p_\mu +\rmd \xi^\mu\wedge \rmd \zeta_\mu\eand \caQ=\xi^\mu p_\mu~.
    \end{equation}
    This symplectic form induces the Poisson bracket
    \begin{equation}\label{eq:Poisson_Vinogradov}
        \begin{aligned}
            \{f,g\}\coloneqq \left(\der{p_\mu}f\right)\left(\der{x^\mu} g\right)&-\left(\der{x^\mu}f\right)\left(\der{p_\mu} g\right)
            \\
            &-(-1)^{|f|}\left(\der{\zeta_\mu} f\right)\left(\der{\xi^\mu}g\right)-(-1)^{|f|}\left(\der{\xi^\mu}f\right)\left(\der{\zeta_\mu} g\right)~,
        \end{aligned}
    \end{equation}
    and we have a Hamiltonian vector field $Q$ given by
    \begin{equation}
        Q=\{\caQ,-\}=\xi^\mu\der{x^\mu}+p_\mu\der{\zeta_\mu}~~~\mbox{for}~~\caQ=\xi^\mu p_\mu~.
    \end{equation}
    The algebra of functions $C^\infty(\caV_n(M))$ is identified with the smooth functions in $x^\mu$ and the analytical functions in the remaining coordinates, and it receives a grading from the grading of the coordinates. The vector field $Q$ is a natural differential on $C^\infty(\caV_n(M))$, and $Q^2=0$ is equivalent to $\{\caQ,\caQ\}=0$. 
    
    We note that the Poisson bracket~\eqref{eq:Poisson_Vinogradov} is a Poisson bracket of degree~$-n$. We can now shift the grading in the algebra of functions by $+n$ to obtain the differential graded Lie algebra 
    \begin{equation}
        \sfL(M)\coloneqq C^\infty(\caV_n(M))[-n]
    \end{equation}
    with differential $Q$ and Lie bracket $\{-,-\}$. The $\ophLie_2$-algebra associated to $\sfL(M)$ (and thus to $\caV_n(M)$) by \ref{thm:ophLie_from_dgLA} is then
    \begin{equation}
        \begin{gathered}
            \frE=\big( \underbrace{C^\infty(M)}_{\frE_{-n+1}}~\xrightarrow{~Q~}~\underbrace{C^\infty_1(M)}_{\frE_{-n+2}}~\xrightarrow{~Q~}~\dotsb~\xrightarrow{~Q~}~\underbrace{C^\infty_{n-1}(M)}_{\frE_0}\big)~,
            \\
            \eps_1(f_1)=\begin{cases}
                Q f_1 & \mbox{for}~|f_1|_\frE<0~,
                \\
                0& \mbox{else}~,
            \end{cases}
            \\
            \eps_2^i(f_1,f_2)\coloneqq \begin{cases}
                \{Q f_1,f_2\} & \mbox{for}~i=0~\mbox{and}~|f_1|_{\frE}=0~,
                \\
                (-1)^{|f_1|_\frE}\{f_1,f_2\} & \mbox{for}~i=1~,
                \\
                0& \mbox{else}
            \end{cases}
        \end{gathered}
    \end{equation}
    for all $f_1,f_2\in \frE$. We can identify the elements of $\frE_k$ with $\Omega^{k+n-1}(M)$ for $k<0$ and $\frE_0\cong \frX(M)\oplus \Omega^{n-1}(M)$, where $\frX(M)$ and $\Omega^k(M)$ are the vector fields and differential $k$-forms on $M$, respectively. The latter are the generalized vector fields on $\caV_n(M)$. Restricted to these, $\eps_2^0$ is (a generalization of) the Dorfman bracket, whose antisymmetrization yields the Courant bracket, and $\eps_2^1$ is a natural contraction $(\frX(M)\oplus \Omega^{n-1}(M))\times (\frX(M)\oplus \Omega^{n-1}(M))\rightarrow \Omega^{n-2}$. 
    
    As an explicit example, let us briefly present the case $n=2$ for some manifold $M$. Here, we have the $2$-term $\ophLie_2$-algebra $\frE$ with underlying differential complex
    \begin{equation}
        \frE=(~\frE_{-1}~\xrightarrow{~\eps_1~}~\frE_{0}~)=(~C^\infty(M)~\xrightarrow{~\rmd~} \frX(M)\oplus \Omega^1(M)~)~.
    \end{equation}
    The binary brackets are the Dorfman bracket, the evident action of $\frE_0$ on $\frE_{-1}$, and the evident dual pairing on $\frE_0$:
    \begin{equation}\label{eq:Dorfman}
        \begin{aligned}
            \eps_2^0(X+\alpha,Y+\beta)&=[X,Y]+\caL_X\beta-\iota_Y\rmd \alpha~,
            \\
            \eps_2^0(X+\alpha,f)&=\caL_X f=\iota_X \rmd f~,
            \\
            \eps_2^1(X+\alpha,Y+\beta)&=\iota_X \beta+\iota_Y \alpha
        \end{aligned}
    \end{equation}
    for all $f\in C^\infty(M)$, $X,Y\in \frX(M)$, and $\alpha,\beta\in \Omega^1(M)$.
    The corresponding $L_\infty$-algebra obtained from \ref{thm:antisym_hLie} yields the well known $L_\infty$-algebra of the Courant algebroid, cf.~e.g.~\cite{Deser:2016qkw}. This $L_\infty$-algebra has the same differential complex as $\frE$, but with higher brackets
    \begin{equation}\label{eq:ass_Courant_algebra}
        \begin{aligned}
            \mu_1(f)&=\rmd f~,\\
            \mu_2(X+\alpha,Y+\beta)&=[X,Y]+\caL_X\beta-\caL_Y\alpha-\tfrac12\rmd\big(\iota_X\beta-\iota_Y\alpha)~,\\
            \mu_2(X+\alpha,f)&=\tfrac12\caL_X f~,\\
            \mu_3(X+\alpha,Y+\beta,Z+\gamma)&=\tfrac{1}{3!}\big(\iota_X\iota_Y\rmd \gamma+\tfrac32\iota_X\rmd\iota_Y\gamma\pm\mbox{perm.}\big)~.
        \end{aligned}
    \end{equation}
    Following~\cite{Deser:2016qkw}, one readily extends this discussion to the pre N$Q$-manifolds underlying double field theory to reproduce the $D$- and $C$-brackets there.
    
    Another class of symplectic $L_\infty$-algebroids is given by the differential graded algebra given by the Batalin--Vilkovisky (BV) complex of a classical field theory, cf.~\cite{Jurco:2018sby} for definitions and conventions. Here, we have a Poisson bracket of degree~$-1$ and a BV complex
    \begin{equation}
        C^\infty(\frF)\coloneqq \left( \ldots~\xrightarrow{~Q~}~C^\infty_{-2}(\frF)~\xrightarrow{~Q~}~C^\infty_{-1}(\frF)~\xrightarrow{~Q~}~C^\infty_{0}(\frF)~\xrightarrow{~Q~}~C^\infty_{1}(\frF)~\xrightarrow{~Q~}~\ldots\right)~,
    \end{equation}
    where $\frF$ is the full BV field space and the $C^\infty_i$ contains (the coordinate functions for) ghosts or gauge parameters for $i=1$, fields for $i=0$, antifields for $i=-1$, and antifields of ghosts for $i=-2$. If we shift this complex by $-1$, we obtain a differential graded Lie algebra, which then gives rise to an $\ophLie_2$-algebra. At the moment, we do not have a concrete interpretation of this $E_2L_\infty$-algebra.
    
    \subsection{Multisymplectic geometry}
    
    There is a close relation between the associated $L_\infty$-algebras of $L_\infty$-algebroids and multisymplectic geometry, as explained in~\cite{Rogers:2010sc} and~\cite{Ritter:2015ffa}.
    
    A multisymplectic manifold $(M,\varpi)$ of degree~$p$, or a $p$-plectic manifold, is a manifold $M$ with a closed differential form $\varpi\in \Omega^{p+1}(M)$ which is non-degenerate in the sense that $\iota_X \varpi=0$ implies $X=0$ for all $X\in \frX(M)$.
    
    Any multisymplectic manifold $(M,\varpi)$ comes with a differential complex
    \begin{equation}
        \sfL(M,\varpi)=\left(~\underbrace{\Omega^0(M)}_{\sfL(M,\varpi)_{-n}}~\xrightarrow{~\rmd~}~\underbrace{\Omega^1(M)}_{\sfL(M,\varpi)_{1-n}}~\xrightarrow{~\rmd~}~\dotsb~\xrightarrow{~\rmd~}~\underbrace{\Omega^{n-1}_{\rm Ham}(M)}_{\sfL(M,\varpi)_{-1}}~\xrightarrow{~\delta~} \underbrace{\frX(M)}_{\sfL(M,\varpi)_0}~\right)~,
    \end{equation}
    where $\Omega^{n-1}_{\rm Ham}(M)$ are the Hamiltonian $n-1$-forms, i.e.~differential forms $\alpha$ for which there are vector fields $\delta(\alpha)$ such that
    \begin{equation}
        \iota_{\delta(\alpha)} \varpi=\rmd \alpha~.
    \end{equation}
    In previous work~\cite{Rogers:2010sc,Rogers:2010nw}, it was realized that the shifted complex $\sfL(M,\varpi)[-1]$ restricted to non-positive degrees carries an $L_\infty$-algebra as well as a differential graded Leibniz algebra. The situation is, in fact, a bit richer.
    \begin{theorem}
        The complex $\sfL(M,\varpi)$ carries a natural differential graded Lie algebra structure with the Lie bracket $\{-,-\}$ given by 
        \begin{equation}
            \begin{aligned}
                \{X_1,X_2\}&\coloneqq [X_1,X_2]~,
                \\
                \{X_1,\alpha_1\}&\coloneqq \caL_{X_1}\alpha_1~,
                \\
                \{\alpha_1,\alpha_2\}&\coloneqq \iota_{\delta(\alpha_1)}\alpha_2-(-1)^{|\alpha_1|\,|\alpha_2|}\iota_{\delta(\alpha_2)}\alpha_1
            \end{aligned}
        \end{equation}
        for all $X_1,X_2\in \frX(M)$ and $\alpha_{1,2}\in \sfL(M,\varpi)$ with $|\alpha_{1,2}|_{\sfL(M,\varpi)}<0$.
    \end{theorem}
    \begin{proof}
        The proof is a straightforward verification of the axioms of a differential graded Lie algebra.
    \end{proof}
    Via \ref{thm:ophLie_from_dgLA}, the above theorem has the following corollary.
    \begin{corollary}
        Any multisymplectic manifold $(M,\varpi)$ comes with an $\ophLie_2$-algebra
        \begin{equation}
            \frE(M,\varpi)=\left(~\underbrace{\Omega^0(M)}_{\frE(M,\varpi)_{-n+1}}~\xrightarrow{~\rmd~}~\underbrace{\Omega^1(M)}_{{\frE(M,\varpi)_{-n+2}}}~\xrightarrow{~\rmd~}~\ldots~\xrightarrow{~\rmd~}~\underbrace{\Omega^{n-1}_{\rm Ham}(M)}_{\frE(M,\varpi)_{0}}~\right)
        \end{equation}
        with nonvanishing binary products
        \begin{equation}
            \begin{aligned}
                \eps_2^0(\alpha,\beta_1)&=\{\delta(\alpha),\beta_1\}=\caL_{\delta(\alpha)}\beta_1
                \\
                \eps_2^1(\beta_1,\beta_2)&=(-1)^{|\beta_1|_{\frE}}\{\beta_1,\beta_2\}=\iota_{\delta(\beta_1)}\beta_2-(-1)^{|\beta_1|\,|\beta_2|}\iota_{\delta(\beta_2)}\beta_1
            \end{aligned}
        \end{equation}
        for all $\alpha\in \frE(M,\varpi)_0$ and $\beta_1,\beta_2\in \frE(M,\varpi)$.
    \end{corollary}
    The antisymmetrization of this $\ophLie_2$-algebra is the $L_\infty$-algebra described in~\cite{Rogers:2010sc,Rogers:2010nw}. Note that the special case $M=S^3$ and $\varpi={\rm vol}_{S^3}$, upon restricting to left-invariant objects, yields another derivation of the hemistrict string $E_2L_\infty$-algebra model $\frstring^{\rm wk,1}_{\rm sk}(\frg)$.
    
    \section{Higher gauge theory with \texorpdfstring{$E_2L_\infty$}{EL-infinity}-algebras}
    
    In this section, we develop and explore the generalities of higher gauge theory using $E_2L_\infty$-algebras as higher gauge algebras.    
    
    \subsection{Homotopy Maurer--Cartan theory for \texorpdfstring{$E_2L_\infty$}{EL-infinity}-algebras}
    
    Recall that given an $L_\infty$-algebra $\frL$ with higher products $\mu_i$, there is a functor $\sfMC(\frL,-)$ taking a differential graded commutative algebra $\fra$ to Maurer--Cartan elements with values in $\fra$, cf.~e.g.~\cite{Chuang:0912.1215}. This functor is represented by the Chevalley--Eilenberg algebra $\sfCE(\frL)$ of the $L_\infty$-algebra. 
    
    What we usually call Maurer--Cartan elements in $\frL$ are Maurer--Cartan elements with values in $\IR$, where the latter is regarded as a trivial differential graded algebra $\IR_\fra$ with underlying vector space $\IR$, spanned by a generator $w$ subject to the relation $w^2=w$, and trivial differential. 
    
    For concreteness sake, let us assume that $\frL$ is degree-wise finite, and let $(t^A)$ be the generators of $\frL[1]^*$ dual to some basis $(\tau_A)$ of $\frL$. A Maurer--Cartan element is encoded in a morphism of differential graded commutative algebras $a:\sfCE(\frL)\rightarrow \IR_\fra$, which is fully determined by the image of the generators $(t^\alpha)$ of degree~$0$,
    \begin{equation}\label{eq:dga_morphism}
        a: \sfCE(\frL)\rightarrow \IR~,~~~t^\alpha \mapsto a^\alpha w
    \end{equation}
    for $a^\alpha\in \IR$. Dually, we have an element $a\coloneqq a^\alpha\tau_\alpha\in \frL_1$, the {\em gauge potential}. Compatibility with the differential requires the {\em curvature}
    \begin{equation}
        f\coloneqq \mu_1(a)+\tfrac12 \mu_2(a,a)+\tfrac1{3!}\mu_3(a,a,a)+\dotsb~~\in \frL_2
    \end{equation}
    to vanish, and the equation $f=0$ is called the homotopy Maurer--Cartan equation. This curvature satisfies the {\em Bianchi identity}
    \begin{equation}
        \sum_{k\geq0}\frac{1}{k!}\mu_{k+1}(a,\ldots,a,f)=0~.
    \end{equation}
    Infinitesimal gauge transformations are obtained from infinitesimal homotopies between morphisms from $\sfCE(\frL)$ to $\IR$. They are 
    parameterized by elements $c\in \frg_0$ and act according to
    \begin{equation}\label{eq:GaugeTrafo}
        \delta_{c} a=\sum_{i\geq 0} \frac{1}{k!}\mu_{k+1}(a,\ldots,a, c)~.
    \end{equation}
    Higher homotopies yield higher gauge transformations.
    
    Similarly, one defines Maurer--Cartan elements of an $A_\infty$-algebra with values in a differential graded algebra.
    
    In the case of $E_2L_\infty$-algebras, we can still consider tensor products of a base $E_2L_\infty$-algebra $\frE$ and a differential graded commutative algebra $\frA$. However, the Chevalley--Eilenberg algebra $\sfCE(\frE)$ is an $\opEilh_2$-algebra and not a differential graded commutative algebra. Therefore the homotopy Maurer--Cartan functor cannot be represented by it directly.
    
    There are two loopholes to this obstruction. First, we can lift the differential graded commutative algebra $\frA$, if it is semifree, to an $\opEilh_2$-algebra $\hat \frA$ as explained in \ref{thm:lift_dgca_to_Eilh}. We can then consider $\opEilh_2$-algebra morphisms 
    \begin{equation}
        a:\sfCE(\frE)\rightarrow \hat \frA~.
    \end{equation}
    Second, we can project $\sfCE(\frE)$ to the Chevalley--Eilenberg algebra of the $L_\infty$-algebra $\frL$ induced by $\frE$ and consider the usual morphisms 
    \begin{equation}
        a:\sfCE(\frL)\rightarrow \frA~.
    \end{equation}
    A third approach is simply to consider general morphisms of $\opEilh_2$-algebras. In particular, one may want to replace differential forms with more general objects, cf.~also~\cite{Ritter:2015zur}.
    
    We note that, in general, the three different types of morphism will give rise to different sets of Maurer--Cartan elements with the first one encompassing the second one. In all the applications we are aware of, however, the second approach is the appropriate one. While the difference between an $E_2L_\infty$-algebra and the corresponding $L_\infty$-algebra obtained by antisymmetrization is then invisible at the level of homotopy Maurer--Cartan theory, the additional algebraic structure in an $E_2L_\infty$-algebra is important in adjusting non-flat higher gauge theories.
    
    \subsection{Adjustment of higher gauge theory}
    
    In the construction of a higher gauge theory from an $E_2L_\infty$-algebra $\frE$, we will always employ the corresponding $L_\infty$-algebra $\frL$ obtained from \ref{thm:antisym}. We then consider its {\em Weil algebra}, which is the Chevalley--Eilenberg algebra of the inner derivations of $\frL$,
    \begin{equation}
        \sfW(\frL)=\Big(\odot^\bullet(\frL[1]^*\oplus \frL[2]^*), Q_\sfW\Big)
        ~,~~~
        Q_\sfW=Q_\sfCE+\sigma~,
    \end{equation}
    where $Q_\sfCE$ is the Chevalley--Eilenberg differential of $\frL$ and $\sigma$ is the shift isomorphism $\sigma:\frL[1]^*\rightarrow \frL[2]^*$, extended to a morphism of differential graded commutative algebras. 
    
    The local kinematical data of an unadjusted higher gauge theory over a patch $U$ of some manifold $M$ is given by a differential graded algebra morphism
    \begin{equation}\label{eq:naive_dga_morphism}
        \caA:\sfW(\frL)\longrightarrow \Omega^\bullet(M)~.
    \end{equation}
    This yields the definition of gauge potentials (the images of $\frL[1]^*$), curvatures (the images of $\frL[2]^*$ together with compatibility of $\caA$ with the differentials on $\frL[1]^*$) and Bianchi identities (compatibility of $\caA$ with the differentials on $\frL[2]^*$). Infinitesimal gauge transformations are given as partially flat homotopies between two such morphisms, and they are therefore determined by the form of the curvatures. For details, see the original discussion in~\cite{Sati:2008eg}; the worked examples in~\cite{Saemann:2019dsl} may also be helpful.
    
    One severe issue with this direct definition of higher gauge theory is that consistency of the gauge algebroid (read: closure of the BRST differential) requires the so-called fake curvature condition, which is highly restrictive~\cite{Saemann:2019dsl}, as mentioned in the introduction. Within supergravity, this problem had been solved in a special case corresponding to the string Lie 2-algebra~\eqref{eq:skeletal_string} by working with different curvatures~\cite{Bergshoeff:1981um,Chapline:1982ww}. As shown in~\cite{Sati:2009ic}, this kinematical data can be obtained from a morphism~\eqref{eq:naive_dga_morphism} after a modification of the Weil algebra, which also results in nicer mathematical properties. Such a modification can be performed for a large class of higher gauge theories, and an appropriately modified Weil algebra was termed {\em adjusted Weil algebra} in~\cite{Saemann:2019dsl}, where also a number of examples were worked out that are relevant to the (1,0) tensor hierarchies of gauged supergravity. In fact, all the kinematical data arising within the tensor hierarchies seem to be adjusted higher gauge theories, and we shall return to them in \ref{sec:tensor_hierachies}. Moreover, the additional structure constants arising in the adjustment seem to originate from the higher products contained in $E_2L_\infty$-algebras that antisymmetrize to the gauge $L_\infty$-algebra. While we do not have a complete picture of the situation yet, we develop a partial one in the next section, which is sufficient for the treatment of tensor hierarchies in maximally supersymmetric gauged supergravities.
    
    \subsection{Firmly adjusted Weil algebras from \texorpdfstring{$\ophLie_2$}{hLie2}-algebras}\label{ssec:firmly_adjusted}
    
    Special cases of Weil algebras that are adjusted and whose corresponding morphisms~\eqref{eq:naive_dga_morphism} into differential forms yield adjusted higher gauge theories with closed BRST complex are the following ones:
    \begin{definition}
        A \uline{firmly adjusted Weil algebra} of an $L_\infty$-algebra $\frL$ is a differential graded commutative algebra obtained from the Weil algebra $\sfW(\frL)$ by a coordinate change
        \begin{equation}\label{eq:firm_rotation}
            \hat t^A\mapsto \hat t'^A\coloneqq \hat t^A+p^A_{B_1B_2\dotso B_mC_1C_2\dotso C_n}\hat t^{B_1}\dotsm\hat t^{B_m}t^{C_1}\dotsm t^{C_n}~,
        \end{equation}
        where $t^A\in \frL[1]^*$, $\hat t^A\in \frL[2]^*$, $m\ge1$, and $n\ge0$, such that the image of the resulting differential $Q_{\rm fadj}$ on generators in $\frL[2]^*$ contains no generator in $\frL[1]^*$ except for at most one of degree~$1$.
    \end{definition}
    We note that putting the generators $(\sigma t^A)$ to zero still recovers the Chevalley--Eilenberg algebra $\sfCE(\frL)$ of $\frL$. In this sense, the coordinate change has not changed the underlying $L_\infty$-algebra. Moreover, note that any Weil algebra is fully contractible in the sense that the cohomology of its linearized differential is trivial. Dually, it is the Chevalley--Eilenberg algebra of an $L_\infty$-algebra which is quasi-isomorphic to the trivial $L_\infty$-algebra. The non-trivial information contained in the Weil algebra is the relation between the generators $(t^A)$ and $(\sigma t^A)$, which translates under the morphism~\eqref{eq:naive_dga_morphism} into the relation between gauge potentials and their curvatures. Our coordinate change thus changes the definition of the curvatures and, as partially flat homotopies describe gauge transformations, also the gauge transformations. Firmly adjusted Weil algebras ensure that the corresponding BRST complex closes: the restricted terms govern the Bianchi identities, which fix the gauge transformations of the curvatures. Closure of the latter is what induces the fake curvature conditions, cf.~the discussion in~\cite[section 4.4]{Saemann:2019dsl}. Thus, firmly adjusted Weil algebras are adjusted Weil algebras in the sense of~\cite{Saemann:2019dsl}.    
    
    As an example, consider the following firmly adjusted Weil algebra of the string Lie 2-algebra~\eqref{eq:skeletal_string}:
    \begin{equation}\label{eq:ext_twt_string_sk_differential}
        \begin{aligned}
            Q_{\rm fadj}~&:~&t^\alpha &\mapsto -\tfrac12 f^\alpha_{\beta\gamma} t^\beta  t^\gamma + \hat t^\alpha~,~&r &\mapsto \tfrac{1}{3!} f_{\alpha\beta\gamma} t^\alpha  t^\beta  t^\gamma -\kappa_{\alpha\beta}t^\alpha\hat t^\beta+ \hat r'~,
            \\
            &&\hat t^\alpha &\mapsto -f^\alpha_{\beta\gamma} t^\beta  \hat t^\gamma~,~
            &\hat r'&\mapsto \kappa_{\alpha\beta}\hat t^\alpha\hat t^\beta~,
        \end{aligned}
    \end{equation}
    which is obtained from the coordinate transformation $\hat r\mapsto \hat r'=\hat r+\kappa_{\alpha\beta} \hat t^\alpha t^\beta$. Here, $t\in \frg[1]^*$, $r\in \IR[2]^*$ and $\hat t=\sigma t$, $\hat r=\sigma r$. Under the morphism~\eqref{eq:naive_dga_morphism}, this firmly adjusted Weil algebra gives rise to the usual string connections
    \begin{equation}
        \begin{aligned}
            a&=A+B\in \Omega^1(M,\frg)\oplus \Omega^2(M,\IR)~,
            \\
            f&=F+H\in \Omega^2(M,\frg)\oplus \Omega^3(M,\IR)~,
            \\
            F&=\rmd A+\tfrac12[A,A]~,
            \\
            H&=\rmd B-\tfrac1{3!}(A,[A,A])+(A,F)=\rmd B+{\rm cs}(A)~.
        \end{aligned}
    \end{equation}
    
    More generally, consider an $L_\infty$-algebra obtained from an $\ophLie_2$-algebra by antisymmetrization. For simplicity, we also assume that the $L_\infty$-algebra has maximally ternary brackets. Its Weil algebra then reads as
    \begin{equation}
        \begin{aligned}
            Q_\sfW t^\alpha&=-(-1)^{|\beta|}m^\alpha_\beta t^\beta-(-1)^{|\gamma|(|\beta|-1)}\tfrac12\,m^{\alpha}_{\beta\gamma}\,t^\beta t^\gamma
            \\
            &~~~~~-(-1)^{|\beta|(|\gamma|+1)+|\delta|(|\beta|+|\gamma|+1)}\tfrac{1}{3!}\,m^{\alpha}_{\beta\gamma\delta}\,t^\beta t^\gamma t^\delta +\hat t^\alpha~,
            \\
            Q_\sfW \hat t^\alpha&=(-1)^{|\beta|}m^\alpha_\beta \hat t^\beta+(-1)^{|\gamma|(|\beta|-1)}\,m^{\alpha}_{\beta\gamma}\,\hat t^\beta t^\gamma
            \\
            &~~~~~+(-1)^{|\beta|(|\gamma|+1)+|\delta|(|\beta|+|\gamma|+1)}\tfrac{1}{2}\,m^{\alpha}_{\beta\gamma\delta}\,\hat t^\beta t^\gamma t^\delta~.
        \end{aligned}
    \end{equation}
    In general, this Weil algebra is clearly not firmly adjusted because of the explicit form of $Q_\sfW\hat t^\alpha$. Let us therefore perform the coordinate change
    \begin{equation}
        \hat t^\alpha\mapsto \hat t'^\alpha\coloneqq \hat t^\alpha+s^\alpha_{\beta\gamma} \hat t^\beta t^\gamma~.
    \end{equation}
    The new Weil differential then reads as follows.
    \begin{equation}\label{eq:fadjust_Qp}
        \begin{aligned}
            Q'_\sfW \hat t'^\alpha&=(-1)^{|\beta|}m^\alpha_\beta \hat t'^\beta
            +(-1)^{1+|\beta|}s^\alpha_{\beta\gamma}\hat t'^\beta \hat t'^\gamma
            +(-1)^{|\gamma|(|\beta|-1)}\,m^{\alpha}_{\beta\gamma}\,\hat t'^\beta t^\gamma
            \\
            &~~~~~+\big(-(-1)^{|\beta|} m^\alpha_\beta s^\beta_{\gamma\delta}+(-1)^{|\gamma|} s^\alpha_{\beta\delta}m^\beta_\gamma
            +(-1)^{|\gamma|+|\delta|} s^\alpha_{\gamma\beta}m^\beta_\delta\big)\hat t'^\gamma t^\delta+\cdots~,
        \end{aligned}
    \end{equation}
    where the ellipsis denotes cubic and higher terms. Let us now further restrict to $\ophLie_2$-algebras obtained from a differential graded algebra via \ref{thm:antisym_hLie} with differential $\Theta^\alpha_\beta$ and structure constants $f^\alpha_{\beta\gamma}$. In this case, we have
    \begin{subequations}
        \begin{equation}
            m^\alpha_\beta=\Theta^\alpha_\beta
            ~,~~~
            m^\alpha_{\beta\gamma}=\begin{cases}
                \tfrac12 f^\alpha_{\delta\gamma}\Theta^\delta_\beta & \mbox{if $|\beta|=1$}~,
                \\
                0 & \mbox{else}~;
            \end{cases}
        \end{equation}
        we also put
        \begin{equation}
            s^\alpha_{\beta\gamma}=\tfrac12(-1)^{|\beta|(|\gamma|+1)} f^\alpha_{\beta\gamma}~.
        \end{equation}
    \end{subequations}
    In the above formulas, $|\alpha|,|\beta|,|\gamma|\geq 1$, and $|\delta|=0$. Together with the Jacobi identity for the $f^\alpha_{\beta\gamma}$, one can then easily verify that $Q'$ becomes a firmly adjusted Weil differential,
    \begin{equation}\label{eq:fadjust_Q}
        Q_{\rm fadj} \hat t'^\alpha=(-1)^{|\beta|}m^\alpha_\beta \hat t'^\beta
        +(-1)^{1+|\beta|\,|\gamma|}\tfrac12f^\alpha_{\beta\gamma}\hat t'^\beta \hat t'^\gamma~.
    \end{equation}
    We thus conclude the following theorem.
    \begin{theorem}\label{thm:firmly_adjust}
        Given an $L_\infty$-algebra with maximally ternary brackets that is obtained from the antisymmetrization of a differential graded Lie algebra by~\ref{prop:dgLA_to_L_infty}, then there is a corresponding firmly adjusted Weil algebra. The data necessary for an adjustment arises from the alternators in the corresponding $\ophLie_2$-algebra.
    \end{theorem}
    
    Below, we shall give examples motivated from higher gauge theory. We stress, however, that the definition of an adjustment is also interesting for purely algebraic considerations, as it allows for the definition of a differential graded algebra of invariant polynomials for an $L_\infty$-algebra which is compatible with quasi-isomorphisms of this $L_\infty$-algebra, cf.~the discussion in~\cite{Saemann:2019dsl}. 
    
    We also note that our construction highlights the features needed for obtaining a firmly adjusted Weil algebra. In particular, it is not necessary that the $\ophLie_2$-algebra was obtained from a differential graded Lie algebra; it was sufficient that there be a relation between the parameters $s^\alpha_{\beta\gamma}$ of the coordinate change and the structure constants $f^\alpha_{\beta\gamma}$ of the Lie algebra to ensure that~\eqref{eq:fadjust_Qp} reduces to~\eqref{eq:fadjust_Q}. This is the case, for example, in the tensor hierarchies in non-maximally supersymmetric gauged supergravity.
    
    \subsection{Example: (1,0)-gauge structures}\label{ssec:10-structures}
    
    As a first more involved example of $E_2L_\infty$-algebras arising in the context of higher gauge theory, let us consider the higher gauge algebra defined in~\cite{Saemann:2017rjm}, see also~\cite{Saemann:2017zpd,Saemann:2019dsl,Rist:2020uaa}. This algebra is a specialization of the general non-abelian algebraic structure identified in~\cite{Samtleben:2011fj} and can be derived from tensor hierarchies, to which we shall return shortly. The latter had received an interpretation as an $L_\infty$-algebra with some ``extra structure'' before, cf.~\cite{Palmer:2013pka} as well as~\cite{Lavau:2014iva}. Here, we show that it is, in fact an $E_2L_\infty$-algebra.
    
    The higher gauge algebra $\hat\frg_{\rm sk}^\omega$ for $\frg$ a quadratic Lie algebra has underlying graded complex
    \begin{equation}\label{eq:ghsk-complex}
        \hat\frg_{\rm sk}^\omega=\left(
        \begin{tikzcd}[row sep=0cm,column sep=2cm]
            \frg^*_v\arrow[r]{}{\mu_1=\sfid} & \frg^*_u & \IR^*_s \arrow[r]{}{\mu_1=\sfid} & \IR_p^*
            \\
            & \oplus & \oplus & \oplus
            \\
            \underbrace{\phantom{\frg^*_v}}_{\hat \frg^\omega_{\rm sk,-3}} & \underbrace{\IR_q}_{\hat \frg^\omega_{\rm sk,-2}} \arrow[r]{}{\mu_1=\sfid} & \underbrace{\IR_r}_{\hat \frg^\omega_{\rm sk,-1}} & \underbrace{\frg_t}_{\hat \frg^\omega_{\rm sk,0}}
        \end{tikzcd}\right)~,
    \end{equation}
    where the subscripts merely help to distinguish between isomorphic subspaces. In~\cite{Saemann:2017zpd}, this differential complex was extended to an $L_\infty$-algebra $\hat\frg_{\rm sk}^\omega$ with higher products
    \begin{equation}\label{eq:ghsk-brackets}
        \begin{aligned}
            \mu_2(t_1,t_2)&=[t_1,t_2]\in\frg_t~,\\
            \mu_2(t,u)&=u\big([-,t]\big)\in\frg^*_u~,~~~&
            \mu_2(t,v)&=v\big([-,t]\big)\in\frg_v^*~,\\
            \mu_3(t_1,t_2,t_3)&=(t_1,[t_2,t_3])\in\IR_r~,~~~&
            \mu_3(t_1,t_2,s)&= s\big(\,(-,[t_1,t_2])\,\big)\in\frg^*_u~,
        \end{aligned}
    \end{equation}
    where $t\in \frg_t$, etc. Moreover, $[-,-]$ and $(-,-)$ denote the Lie bracket and the quadratic form in $\frg$, respectively. When constructing gauge field strengths based on this $L_\infty$-algebra, the following, additional maps feature:
    \begin{equation}
        \begin{aligned}
            \nu_2(t_1,t_2)&=-2(t_1,t_2)\in \IR_r~,~~~&\nu_2(t,s)&=2s(-,t)\in\frg_u^*~,
            \\
            \nu_2(t_1,u_1)&=u_1\big([-,t_1]\big)\in \frg_v^*~.
        \end{aligned}
    \end{equation}
    
    As motivated in more detail later, it is useful to first perform a quasi-isomorphism on $\hat\frg_{\rm sk}^\omega$ leading to the higher brackets
    \begin{equation}\label{eq:ghsk-brackets2}
        \begin{aligned}
            \mu_2(t_1,t_2)&=[t_1,t_2]\in\frg_t~,\\
            \mu_2(t,u)&=\tfrac12u\big([-,t]\big)\in\frg^*_u~,~~~&
            \mu_2(t,v)&=\tfrac12v\big([-,t]\big)\in\frg_v^*~,\\
            \mu_3(t_1,t_2,t_3)&=(t_1,[t_2,t_3])\in\IR_r~,~~~&
            \mu_3(t_1,t_2,s)&= s\big(\,(-,[t_1,t_2])\,\big)\in\frg^*_u~,
            \\
            \mu_3(t_1,t_2,u)&= \tfrac14v\big(\,(-,[t_1,t_2])\,\big)\in\frg^*_v~.
        \end{aligned}
    \end{equation}
    This is the $L_\infty$-algebra obtained by \ref{thm:antisym_hLie} from the $\ophLie_2$-algebra $\frE$ with differential complex~\eqref{eq:ghsk-complex} with $\eps_1=\mu_1$ and the additional binary products
    \begin{equation}
        \begin{aligned}
            \eps_1&=\mu_1~,
            \\
            \eps^0_2(t_1,t_2)&=-\eps^0_2(t_2,t_1)=[t_1,t_2]\in \frg_t~,
            \\
            \eps^0_2(t,u)&=u\big([-,t]\big)\in \frg^*_u~,~~~&
            \eps^0_2(t,v)&=v\big([-,t]\big)\in \frg_v^*~,
            \\
            \eps^1_2(t_1,t_2)&=\eps^1_2(t_2,t_1)=2(t_1,t_2)\in \IR_r~,
            \\
            \eps^1_2(t,s)&=3!s\big(-,t\big)\in\frg_u^*~,~~~
            &\eps^1_2(s,t)&=\eps^1_2(t,s)=3!s\big(-,t\big)\in \frg_u^*~,
            \\
            \eps^1_2(t,u)&=u\big([-,t]\big)\in \frg_v^*~,~~~&\eps^1_2(u,t)&=\eps^1_2(t,u)=u\big([-,t]\big)\in \frg_v^*~,
        \end{aligned}
    \end{equation}
    as one verifies by direct computation. This $\ophLie_2$-algebra is obtained from a differential graded Lie algebra $\frG$ by \ref{thm:ophLie_from_dgLA}, and we have
    \begin{equation}\label{eq:ghsk-complex2}
        \frG=\left(
        \begin{tikzcd}[row sep=0cm,column sep=1.6cm]
            \frg^*_v\arrow[r]{}{\mu_1=\sfid} & \frg^*_u & \IR^*_s \arrow[r]{}{\mu_1=\sfid} & \IR_p^*
            \\
            & \oplus & \oplus & \oplus
            \\
            \underbrace{\phantom{\frg^*_v}}_{\frG_{-4}} & \underbrace{\IR_q}_{\hat \frG_{-3}} \arrow[r]{}{\mu_1=\sfid} & \underbrace{\IR_r}_{\frG_{-2}} & \underbrace{\frg_t}_{\frG_{-1}} \arrow[r]{}{\mu_1=\sfid} & \underbrace{\frg_{\hat t}}_{\frG_0}
        \end{tikzcd}\right)
    \end{equation}
    with the non-trivial Lie brackets $[-,-]_\frG$ fixed by
    \begin{equation}
        \begin{aligned}
            [\hat t_1,\hat t_2]_\frG&\coloneqq [\hat t_1,\hat t_2]\in\frg_{\hat t}~,
            ~~~&[\hat t_1,t_2]_\frG&\coloneqq [\hat t_1,t_2]\in\frg_{t}~,
            \\
            [\hat t_1,u]_\frG&\coloneqq u([-,t])\in\frg_u^*~,
            ~~~&[\hat t_1,v]_\frG&\coloneqq v([-,t])\in\frg_u^*~,
            \\
            [t_1,t_2]_\frG&\coloneqq (t_1,t_2)\in\IR_t~,
            ~~~&[t_1,s]_\frG&\coloneqq \alpha_2 s(-,t_1)\in\frg_u^*~.
        \end{aligned}
    \end{equation}
    This is an extension of the example presented at the end of~\ref{ssec:antisym_hLie}. 
    
    We thus see that we have the following sequence that leads to a construction of $\hat\frg_{\rm sk}^\omega$:
    \begin{equation}
        \mbox{dg Lie algebra}~\frG~~\xrightarrow{~\mbox{\Cref{thm:ophLie_from_dgLA}}~}~~\mbox{$\ophLie_2$-algebra}~\frE~~\xrightarrow{~\mbox{\Cref{thm:antisym_hLie}}~}~~\mbox{$L_\infty$-algebra}~\hat\frg_{\rm sk}^\omega~,
    \end{equation}
    specializing the picture~\eqref{eq:diagram_algebras}. The additional information (i.e.~structure constants) contained in the $E_2L_\infty$-algebra are vital for constructing the adjusted form of the curvatures. 
    
    A corresponding adjusted Weil algebra was found in~\cite{Saemann:2019dsl}, and it agrees with the one obtained from our construction of a firmly adjusted Weil algebra from \ref{ssec:firmly_adjusted}:
    \begin{equation}
        \def\arraystretch{1.5}
        \begin{aligned}
            Q_{\rm fadj}~&:~&t^\alpha &\mapsto  -\tfrac12 f^\alpha_{\beta\gamma} t^\beta t^\gamma+\hat t^\alpha~,~~~& p&\mapsto -s+\hat p~,\\
            &&\hat t^\alpha &\mapsto  -f^\alpha_{\beta\gamma} t^\beta \hat t^\gamma~,~~~& \hat p&\mapsto \hat s~,\\
            &&r&\mapsto \tfrac{1}{3!} f_{\alpha\beta\gamma} t^\alpha t^\beta t^\gamma-\kappa_{\alpha\beta}t^\alpha \hat t^\beta+q+\hat r~,~~~& s&\mapsto \hat s~,\\
            &&\hat r&\mapsto \kappa_{\alpha\beta}\hat t^\alpha \hat t^\beta-\hat q~,~~~& \hat s&\mapsto 0~,\\
            &&u_\alpha &\mapsto -f^\gamma_{\alpha\beta}t^\beta u_\gamma-\tfrac12 f_{\alpha\beta\gamma}t^\beta t^\gamma s-v_\alpha+\hat u_\alpha~,~~~&
            q &\mapsto \hat q~,\\
            &&\hat u_\alpha &\mapsto -f^\gamma_{\alpha\beta}t^\beta \hat u_\gamma+\hat v_\alpha~,~~~&
            \hat q &\mapsto 0~,\\
            &&v_\alpha&\mapsto -f^\gamma_{\alpha\beta}t^\beta v_\gamma - f^\gamma_{\alpha\beta}\hat t^\beta u_\gamma +f_{\alpha\beta\gamma}t^\beta\hat t^\gamma s-\tfrac12 f_{\alpha\beta\gamma} t^\beta t^\gamma \hat s+\hat v_\alpha~,\\
            &&\hat v_\alpha&\mapsto -f^\gamma_{\alpha\beta}t^\beta \hat v_\gamma+f^\gamma_{\alpha\beta}\hat t^\beta \hat u_\gamma~.
        \end{aligned}
    \end{equation}

    \section{Tensor hierarchies}\label{sec:tensor_hierachies}
    
    Tensor hierarchies are particular forms of higher gauge theories that were introduced in the context of gauging maximal supergravity theories~\cite{deWit:2005hv, deWit:2005ub,Samtleben:2005bp,deWit:2008ta,Bergshoeff:2009ph}. They are constructed using the embedding tensor formalism, introduced in~\cite{Cordaro:1998tx, Nicolai:2000sc, deWit:2002vt, deWit:2003hq}. For comprehensive reviews  see~\cite{Samtleben:2008pe, Trigiante:2016mnt}. Tensor hierarchies are also crucial to, for example, conformal field theories such as the $\mathcal{N}=(1,0)$ superconformal models of~\cite{Samtleben:2011fj, Samtleben:2012mi,Samtleben:2012fb}. 
    
    Although initially applied to gauged supergravity theories,  tensor hierarchies do not  require supersymmetry and appear through the embedding tensor formalism applied to  the gauging of  a  broader class of Einstein--Maxwell-matter  theories, as discussed in~\cite{Bergshoeff:2009ph,Hartong:2009vc}.  
    
    \subsection{Physical context}
    
    Before analyzing the algebraic structure underlying tensor hierarchies in more detail, let us briefly review the physical context. Consider the Lagrangian of ungauged Einstein--Maxwell-scalar theory in $d$ dimensions,
    \begin{equation}\label{Lth}
        \mathcal{L}_{\text{ungauged}} = \star R +\tfrac{1}{2} g_{xy} \rmd \varphi^x \wedge \star \rmd \varphi^y  -\tfrac{1}{2} a_{ij} F^{i} \wedge \star F^{j} +\cdots
    \end{equation}
    with scalars $\varphi^x$ mapping spacetime to a scalar manifold $\caM$ and 1-form abelian gauge potentials $A^i$ with field strengths $F^{i}=\rmd A^{i}$. Here, $g_{xy}(\varphi)$ and $a_{ij}(\varphi)$ are symmetric and positive-definite on the entire scalar manifold $\caM$. The ellipsis denotes possible deformations, such as a scalar potential $\mathcal{V}(\varphi)$ or topological terms such as, e.g., ${d}_{ijk}F^{i} \wedge  F^{j}\wedge A^{k}$ for $d=5$, as familiar from supergravity. The set of constant tensors controlling these deformation terms, which includes those appearing in the tensor hierarchies that do not enter~\eqref{Lth}, will be referred to as {\em deformation tensors}. 
    
    We assume that there is a global symmetry group $\sfG$ acting on scalars and 1-form potentials such that the undeformed action~\eqref{Lth} is invariant under this action. In particular, the total gauge potential one-form $A$ takes values in a representation $V_{-1}$ of $\sfG$, which is isomorphic to the fibers of the tangent space of the scalar manifold. In the presence of deformations, we assume that there is a non-abelian subgroup $\sfK\subseteq \sfG$ leaving the full action invariant.    
    
    Infinitesimally, the corresponding Lie algebra actions of $\frg=\sfLie(\sfG)$ on the scalars and the gauge potential are given by 
    \begin{equation}    
        \delta_\lambda \varphi^x = \lambda^{\alpha} k_{\alpha}{}^{x}(\varphi)
        \eand
        \delta_\lambda A^{i} =  \lambda^{\alpha} t_{\alpha}{}^{i}{}_{j} A^{j}~,
    \end{equation}
    where  $t_{\alpha}{}^{i}{}_{j}, \alpha=1,2,\ldots,\dim \frg$, are the generators of $\frg$ in the representation $V_{-1}$ with respect to some bases $e_\alpha$ of $\frg$ and $e_i$ of $V_{-1}$. Invariance under $\sfG$ requires that $k(\phi)$ be  Killing vectors of the scalar manifold and $\mathcal{L}_{k_\alpha}a_{ij}=-2\ t_{\alpha}{}^{k}{}_{(i} a_{j)k}$. 
    
    In order to gauge\footnote{That is, we promote a global symmetry $\sfH$ to a local one by adding a principal $\sfH$-bundle on our spacetime and consider (a part of) the one-form potential as a connection on this bundle.} a subgroup $\sfH\subseteq \sfK\subseteq \sfG$ with Lie algebra $\frh=\sfLie(\sfH)$, we first note that we can trivially regard the pair $(V_{-1},\frg$) as a differential graded Lie algebra
    \begin{equation}\label{eq:trivialdgLA}
        V=(~V_{-1}~\xrightarrow{~0~}~V_{0}=\frg~)
    \end{equation}
    with evident Lie bracket on $V_0$ and the Lie bracket $[-,-]:V_{0}\times V_{-1}\rightarrow V_{-1}$ given by the action of $\frg$ on $V_{-1}$. Because the gauge potential takes values in $V_{-1}$, it does not make sense to gauge a Lie subalgebra of $\frg$ which is larger than $V_{-1}$. Therefore, we can identify the subalgebra $\frh$ with the image of a linear map 
    \begin{equation}
        \begin{aligned}
            \Theta: V_{-1}&\twoheadrightarrow \frh~~~~~~\subset \frg~,
            \\
            e_i&\mapsto \Theta_i{}^\alpha e_\alpha~.
        \end{aligned}
    \end{equation}
    The $(\Theta_{i}{}^{\alpha} e_\alpha)$ then form a spanning set\footnote{but not necessarily a basis} of the Lie algebra $\frh$. Moreover, we have an induced action of $\frh$ on $V_{-1}$, given by 
    \begin{equation}
        (\Theta_{i}{}^{\alpha} e_\alpha) \acton e_j=\Theta_{i}{}^{\alpha} t_{\alpha j}{}^k e_k\eqqcolon X_{ij}{}^k e_k
    \end{equation}
    with 
    \begin{equation}
        t_{\alpha j}{}^k=-t_{\alpha}{}^k{}_j\eand t_{\alpha i}{}^jt_{\beta j}{}^k-t_{\beta i}{}^jt_{\alpha j}{}^k=-f_{\alpha\beta}{}^\gamma t_{\gamma i}{}^k~.
    \end{equation}
    
    In order to guarantee closure of the Lie bracket on $\frh$ and consistency of the action, we can assume that we can incorporate $\Theta$ into~\eqref{eq:trivialdgLA} such that 
    \begin{equation}\label{eq:ThetadgLA}
        V_\Theta=(~V_{-1}~\xrightarrow{~\Theta~}~V_{0}=\frg~)
    \end{equation}
    is again a differential graded Lie algebra. To jump ahead of the story, note that this guarantees the existence of a higher gauge algebra via~\ref{prop:dgLA_to_L_infty}, which we anticipate as part of the construction of a higher gauge theory. The fact that $\Theta$ is a derivation for the graded Lie bracket then implies the quadratic {\em closure constraint}
    \begin{equation}\label{eq:quad_closure}
        f_{\alpha\beta}{}^{\gamma}\Theta^{\alpha}_{i}\Theta^{\beta}_{j}=\Theta^{\gamma}_{k}X_{ij}{}^{k}~~~\Leftrightarrow~~X_{im}{}^\ell X_{j\ell}{}^n-X_{jm}{}^\ell X_{i\ell}{}^n=-X_{ij}{}^\ell X_{\ell m}{}^n~. 
    \end{equation}
    It is well known that this quadratic closure constraint implies that the $X_{ij}{}^k$ form the structure constants of a Leibniz algebra on $V_{-1}$,
    \begin{equation}\label{eq:action_V_1_V_1}
        e_i\circ e_j\coloneqq X_{ij}{}^k e_k~,
    \end{equation}
    cf.~e.g.~\cite{Lavau:2017tvi,Hohm:2018ybo,Kotov:2018vcz}. This is unsatisfactory given that the 1-form gauge potentials $A$ will take values in $V_{-1}$ and $V_{-1}$ should therefore have some Lie structure. As noted in~\cite{Saemann:2019dsl}, and as evident from \ref{prop:Leib_is_hemistrict_ELinfty}, this Leibniz algebra can be promoted to an $\ophLie_2$-algebra. Moreover, the fact that we have the differential graded algebra~\eqref{eq:ThetadgLA} guarantees that we will have an $\ophLie_2$-algebra via \ref{thm:antisym_hLie} (or, if preferred, the corresponding $L_\infty$-algebra obtained from \ref{thm:antisym_hLie}). This will turn out to be indeed the higher gauge algebra underlying the tensor hierarchies.
    
    But let us continue with the tensor hierarchy from the physicists' perspective. The quadratic closure constraint~\eqref{eq:quad_closure} allows us to introduce a consistent combination of a covariant derivative on the scalar fields and local transformations parameterized by $\Lambda_{(0)}\in C^\infty(M)\otimes V_{-1}$:
    \begin{equation}\label{eq:first_gauge}
        \begin{gathered}
            \nabla\varphi^i\coloneqq  \rmd \varphi^i + \Theta^{\alpha}_{j} A^{j} k_{\alpha}{}^{i}(\varphi)~,
            \\
            \delta_{\Lambda_{(0)}} \varphi^i\coloneqq  {\Lambda_{(0)}}^{i}\Theta^{\alpha}{}_{j}  k_{\alpha}{}^{j}(\varphi)~,~~~
            \delta_{\Lambda_{(0)}} A^i\coloneqq \rmd {\Lambda_{(0)}}^i+X_{jk}{}^iA^j{\Lambda_{(0)}}^k~.            
        \end{gathered}
    \end{equation}
    Note that the action~\eqref{eq:action_V_1_V_1} of $V_{-1}$ on $V_{-1}$ is usually not faithful, and the parameterization by $\Lambda^i$ is thus usually highly degenerate. 
    
    In light of our above discussion of the higher Lie algebra arising from the Leibniz algebra~\eqref{eq:action_V_1_V_1}, it is not surprising that the naive gauge transformation~\eqref{eq:first_gauge} of the gauge potential $A$ does not render the naive definition of curvature $\rmd A^{i}+\tfrac12 X_{jk}{}^{i}A^{j}\wedge A^{k}$ covariant. This is remedied by introducing a second $\sfG$-module $V_{-2}\subset \text{Sym}^2(V_{-1})$, where $r=1,2\ldots\dim V_{-2}$ for some basis $(e_{r})$ together with a map
    \begin{equation}
        \begin{aligned}
            Z: V_{-2}&\rightarrow V_{-1}~,
            \\
            e_r&\mapsto Z^i{}_r e_i~.
        \end{aligned}
    \end{equation}
    This allows us to introduce a $V_{-2}$-valued 2-form potential $B$ and a $V_{-2}$-valued 1-form gauge parameter $\Lambda_{(1)}$ to generalized gauge transformations and curvatures as usual in higher gauge theory:
    \begin{equation}
        \begin{gathered}
            \delta A^{i} = \rmd \Lambda^i_{(0)}+X^i_{jk}A^j\Lambda^k_{(0)} +Z^{i}{}_{r}\Lambda^{r}_{(1)}~,~~~
            \delta B^r= \nabla\Lambda^{r}_{(1)}+\ldots~,
            \\
            F^{i}=\rmd A^{i}+\tfrac12 X^{i}{}_{jk}A^{j}\wedge A^{k}+Z^{i}{}_{r}B^{r}~,~~~H^r=\nabla B^r+\ldots~,
        \end{gathered}
    \end{equation}
    where here $\nabla$ is the covariant derivative given by the natural action of $\frh$ on $V_{-2}$ and the ellipses refer to covariantizing terms that are needed to complete the kinematical data to that of an adjusted higher gauge theory. In particular, the latter will include terms involving the various deformation tensors. This process is then iterated in a reasonably obvious fashion until the full kinematical data of an adjusted higher gauge theory is obtained\footnote{The fact that this iteration terminates is guaranteed because spacetime is finite-dimensional.}.
    
    In the gauged supergravity literature there is also often a  linear \emph{representation constraint}
    \begin{equation}    
        P_\Theta \Theta =\Theta~,
    \end{equation}
    where $P_\Theta$ is the projector onto the representation contained in $V_1^*\otimes \mathfrak{g}$  carried by $\Theta$, which  will be denoted $\rho_\Theta$. This can be understood as a requirement of supersymmetry~\cite{deWit:2002vt,deWit:2005hv}, the mutual locality of the action~\cite{deWit:2005ub} or anomaly cancellation~\cite{DeRydt:2008hw}. 
    
    A final important ingredient is now that the electromagnetic duality contained in U-duality needs to link potential $p$-forms to potential $d-p-2$-forms, and correspondingly the $\sfG$-modules $V_{-p}$ and $V_{p+2-d}$ have to be dual to each other in the lowest degrees that involve physical gauge potentials. 
    
    The above constraints restrict severely the choices of representations $V_{-2}$, $V_{-3}$. In \ref{Mcharges} we listed some important concrete examples of maximal supergravities, in which $\sfK=\sfG$. In this case, there is a tensor hierarchy dgLa determined by the U-duality group~\cite{Palmkvist:2011vz, Palmkvist:2013vya}, with graded vector space described in \autoref{Mcharges} and derivation given by the action of $\Theta$. Also, the electromagnetic duality is visibly reflected in the duality of representations in the cases $d=5,6,7$. 
    \begin{table}[!ht]\small
        \begin{tabular*}{\textwidth}{cccccccccccccccc}
            \hline
            \hline
            $d$ &$ \sfG                           $&$ V_{-1}            $&$ V_{-2}           $&$ V_{-3}           $&$ V_{-4}            $&$ V_{-5}       $&$ V_{-6}       $   \\
            \hline
            7   &$ \sfSL(5,\IR )                  $&$ \rep{10}_c    $&$ \rep{5}     $&$ \rep{5}_c    $&$ \rep{10} $&$ \rep{24} $&$ \rep{15}_c\oplus\rep{40} $& \\
            6   &$ \sfSO(5,5)                $&$ \rep{16}_c     $&$ \rep{10}    $&$ \rep{16} $&$ \rep{45} $&$ \rep{144} $&$ \rep{10\oplus126\oplus320} $&\\
            5   &$ \sfE_{6(6)}              $&$ \rep{27}_c     $&$ \rep{27} $&$ \rep{78} $&$ \rep{351}_c $&$ \rep{27\oplus1728} $\\
            4   &$ \sfE_{7(7)}              $&$ \rep{56} $&$ \rep{133} $&$ \rep{912} $&$ \rep{133\oplus8645} $\\
            3   &$ \sfE_{8(8)}              $&$ \rep{248} $&$ \rep{1\oplus 3875} $&$ \rep{3875\oplus147250} $\\
            \hline
            \hline
        \end{tabular*}
        \caption{Global symmetry groups $\sfG$ of maximal supergravity in $3\leq d\leq 7$ spacetime dimensions and their maximal compact subgroups (ignoring discrete factors). The $\sfG$ representations $V_{-p}$ are carried by  $p$-forms in the tensor hierarchy. The scalars ($0$-forms) are valued in $\caM\coloneqq \sfG/\sfG_0$,  where $\sfG_0\subset \sfG$ is the maximal compact subgroup. } \label{Mcharges}
    \end{table}
    
    We note that in the presence of generic deformations, the differential graded Lie algebra constructed in the maximally supersymmetric case is actually insufficient and needs to be extended further by at least one step in both directions. We shall explain this below, when discussing the example $d=5$.
    
    \subsection{Generic tensor hierarchies}\label{ssec:gen_tensor_hierarchies}
    
    Let us ignore the link between tensor hierarchies and gauged supergravity for a moment; clearly, the resulting kinematical data is potentially of interest in higher gauge theory in a much wider context. 
    
    The construction prescription is rather straightforward. We consider a Lie algebra $\frg$, which we enlarge to a differential graded Lie algebra 
    \begin{equation}
        \begin{aligned}
            V=\Big(~\dotsb \xrightarrow{~\rmd~} V_{-2} \xrightarrow{~\rmd~} V_{-1} \xrightarrow{~\rmd~} V_0=\frg \xrightarrow{~\rmd~} V_1 \xrightarrow{~\rmd~} \dotsb~\Big)~,
        \end{aligned}
    \end{equation}
    where we allowed for additional vector spaces $V_i$ with $i>0$. All vector spaces $V_i$ are $\frg$-modules, and the Lie bracket on $V_0$ as well as the Lie brackets on $V_0\otimes V_i$ are given. Further Lie brackets $[-,-]:V_i\otimes V_j\rightarrow V_{i+j}$ can be introduced, but due to the Jacobi identity, the underlying structure constants have to be invariant tensors of $\frg$ (as we shall also see below in an example). The differentials do not have to satisfy this restriction. As an additional constraint, we can also impose the condition that $V_{-p}^*=V_{p+2-d}$ as required by the U-duality condition from supergravity. This can be useful in the construction of action principles.
    
    To illustrate the above, let us construct a generic example in $d=5$. Let $\frg$ be a Lie algebra and $V_{-1}$ any representation. Imposing the duality constraint and allowing for an extension in one degree on either side leads to the differential complex 
    \begin{equation}
        \begin{aligned}
            V=\Big(~V_{-4}\cong{\rm coker}(\Theta)^* \xrightarrow{~\rmd~} V_{-3}\cong\frg^* &\xrightarrow{~\rmd~} V_{-2}\cong V_{-1}^* 
            \\
            &\xrightarrow{~\rmd~} V_{-1} \xrightarrow{~\Theta~} V_0=\frg \xrightarrow{~\rmd~} V_1\cong{\rm coker}(\Theta)~\Big)~.
        \end{aligned}
    \end{equation}
    Let us now switch to the Chevalley--Eilenberg description $\sfCE(V)$ of the differential graded Lie algebra $V$ we want to construct, which is generated by coordinates $r^\mu$, $r^\alpha$, $r^a$, $r_a$, $r_\alpha$, $r_\mu$ of degrees $0,1,2,3,4,5$, respectively. We note that we have a natural symplectic form on $V[1]^*$ of degree 5,
    \begin{equation}
        \omega=\rmd r^\alpha \wedge \rmd r_\alpha+\rmd r^a\wedge \rmd r_a+\rmd r^\mu \wedge \rmd r_\mu~.
    \end{equation}
    Compatibility of the Lie algebra action with the duality pairing amounts to the fact that the Chevalley--Eilenberg differential $Q$ is Hamiltonian for the Poisson bracket of degree~$-5$,
    \begin{equation}\label{eq:Poisson_5d}
        \begin{aligned}
            \{f,g\}\coloneqq &-\parder[f]{r_\alpha}~\parder[g]{r^\alpha}+(-1)^{|f|+1}\parder[f]{r^\alpha}~\parder[g]{r_\alpha}-\parder[f]{r_a}~\parder[g]{r^a}+(-1)^{|g|+1}\parder[f]{r^a}~\parder[g]{r_a}
            \\
            &-\parder[f]{r_\mu}~\parder[g]{r^\mu}+(-1)^{|g|+1}\parder[f]{r^\mu}~\parder[g]{r_\mu}
        \end{aligned}
    \end{equation}
    induced by $\omega$. That is,
    \begin{equation}
        Q=\{\caQ,-\}~,~~~|\caQ|=6~.
    \end{equation}
    The most generic Hamiltonian $\caQ$ of degree~$6$ that is at most cubic in the generators\footnote{This restriction is required to obtain a differential graded Lie algebra, as opposed to an $L_\infty$-algebra} is
    \begin{equation}
        \begin{aligned}
            \caQ&=\tfrac12 f_{\beta\gamma}{}^\alpha r_\alpha r^\beta r^\gamma+t_{\alpha a}{}^b r^\alpha r^a r_b+\tfrac1{3!}d_{abc} r^ar^br^c+\tfrac12 Z^{ab}r_ar_b+\Theta_{a}{}^{\alpha}r^ar_\alpha
            \\
            &~~~~+g^{\mu}_{1\alpha}r_\mu r^\alpha+g^{\mu}_{2\alpha\nu}r_\mu r^\alpha r^\nu+g_{3\mu a}^\alpha r^\mu r^ar_\alpha+g_{4\mu}^{ab} r^\mu r_ar_b~,
        \end{aligned}
    \end{equation}
    where besides the structure constants $f_{\beta\gamma}{}^\alpha$ and the embedding tensor $\Theta_{a}{}^{\alpha}$ we have the deformation tensors $d_{abc}$ and $Z^{ab}$, which are totally symmetric and antisymmetric, respectively, due to the grading of the generators. The remaining structure constants will be called {\em auxiliary}. For $\caQ$ to give rise to a Chevalley--Eilenberg differential, we have to impose
    \begin{equation}
        Q^2=0~~~\Leftrightarrow~~~\{\caQ,\caQ\}=0~.
    \end{equation}
    This equation implies conditions on the structure constants. For example, we have
    \begin{equation}\label{eq:cond_1}
        \Theta_{a}{}^{\gamma} f_{\beta\gamma}{}^{\alpha}+t_{\beta a}{}^b\Theta_{b}{}^{\alpha}-g^\mu_{1\beta}g_{3\mu a}^\alpha=0~.
    \end{equation}
    For $g_{1}=g_{3}=0$, this means that the embedding tensor is an invariant tensor, which is clearly too strong a condition. We can make a non-canonical choice of an embedding
    \begin{equation}
        i:{\rm coker}(\Theta)\hookrightarrow\frg~,
    \end{equation}
    which is given by structure constants $i^\alpha_\mu$ such that
    \begin{equation}
        i^\alpha_\mu g_{1\alpha}^\nu=\delta_\mu^\nu~.
    \end{equation}
    With this choice, we can split the condition~\eqref{eq:cond_1} into
    \begin{equation}
        \begin{aligned}
            \Theta_c{}^\beta\Theta_{a}{}^{\gamma} f_{\beta\gamma}{}^{\alpha}+X_{ca}^b\Theta_{b}{}^{\alpha}&=0~,
            \\
            i^\beta_\mu(\Theta_{a}{}^{\gamma} f_{\beta\gamma}{}^{\alpha}+t_{\beta a}{}^b\Theta_{b}{}^{\alpha})&=g_{3\mu a}^\alpha~,
        \end{aligned}
    \end{equation}
    and the first condition is the usual one encountered in the $d=5$ tensor hierarchy, while the second condition fixes one of the auxiliary structure constants. Besides the above condition and the fact that $f_{\beta\gamma}{}^{\alpha}$ and $t_{\alpha a}{}^b$ are the structure constants of the Lie algebra $\frg$ and a representation of $\frg$, we also have
    \begin{equation}
        \begin{aligned}
            Z^{ab}\Theta_{b}{}^{\alpha} &=0
            ~,~~~&
            \Theta_{a}{}^{\alpha} g_{1\alpha}^\mu&=0~,
            \\
            Z^{ab}d_{acd}-2X^{a}_{(cd)}&=0
            ~,~~~&
            Z^{a[b} t_{\alpha a}{}^{c]}+2g_{1\alpha}^\mu g_{4\mu}^{bc}&=0~,
            \\
            t_{\alpha (a}{}^d d_{bc)d}&=0~,
        \end{aligned}
    \end{equation}
    as well as a number of conditions for the auxiliary structure constants. As expected, the tensor $d_{abc}$ capturing the Lie bracket $V_{-1}\otimes V_{-1}\rightarrow V_{-2}$ has to be an invariant tensor.
    
    The kinematical data of a generic tensor hierarchy can then be constructed from the firmly adjusted Weil algebra of the corresponding $L_\infty$-algebra as described in detail in \ref{ssec:firmly_adjusted}.
    
    We note that the condition that $d_{abc}$ be an invariant tensor is to strong a constraint, e.g.~for the non-maximally supersymmetric case. From the formulas of the curvatures, it is clear that there is no differential graded Lie algebra underlying this case, if the higher gauge algebra is constructed using the formulas of \ref{thm:ophLie_from_dgLA}. This observation strongly suggests that there are  generalizations of these derived bracket constructions, but this is beyond the scope of this paper.
    
    \subsection{Example: \texorpdfstring{$d=5$}{d=5} maximal supergravity}
    
    Let us give a concrete and complete picture of the interpretation of a tensor hierarchy using $\ophLie_2$-algebras, including the construction of curvatures. We choose the case $d=5$, which allows us to recycle observations made in \ref{ssec:gen_tensor_hierarchies}. For a detailed discussion of this theory, see~\cite{deWit:2004nw}.
    
    Maximal supergravity in $d=5$ dimensions has the non-compact global symmetry group $\sfE_{6(6)}(\IR)$~\cite{Cremmer:1980gs}. When dimensionally reducing from $d=11$, in order to make manifest the $\fre_{6(6)}$ structure of the scalar sector in $d=5$, one must first dualize the 3-form potential, as described in detail in~\cite{Cremmer:1997ct}. This gives a total of 42~scalars parameterizing
    $\sfE_{6(6)}(\IR)/\sfUSp(8)$.
    
    The fully dualized bosonic Lagrangian with manifest $\sfE_{6(6)}(\IR)$-invariance can be written as
    \begin{equation}\label{D5L}
        \mathcal{L}_5 = R\star 1 +\tfrac{1}{2} g_{xy} \rmd \varphi^x \wedge \star \rmd \varphi^y  -\tfrac{1}{2} a_{ab} F^{a} \wedge \star F^{b} -\tfrac{1}{6} d_{abc}F^{a}_{\2} \wedge F^{b}_{\2} \wedge A^{c}_{\1}~.
    \end{equation}
    The 1-form potentials transform linearly in the ${\rep{27 } }_c$ of  $\mathfrak{e}_{6(6)}$, and $a,b,c \in \{ 1,\ldots, 27\}$. In addition to the singlet $\rep{1}\in{\rep{27 } }_c\otimes \rep{27}$, used to construct the 1-form kinetic term, there is a singlet in the totally symmetric 3-fold tensor product $\rep{1}\in\bigodot^3(\mathbf{{27} }_c)$, which is used to construct the topological cubic term.
    
    For the construction of the tensor hierarchy, we shall need the following $\sfE_6$-invariant tensors: 
    \begin{equation}        
        \begin{aligned}
            f_{\alpha\beta\gamma} & \in \bigwedge\nolimits^3~\mathbf{78}
            ~,~~~&
            t_{\alpha}{}_{a}{}^{b}& \in \mathbf{78}\otimes \mathbf{{27}}\otimes\mathbf{{27} }_c~,
            \\
            d^{abc}& \in \bigodot\nolimits^3~\mathbf{27}
            ~,~~~&
            d_{abc}& \in \bigodot\nolimits^3~\mathbf{{27} }_c~.
        \end{aligned}
    \end{equation}
    To optimize our notation, we also introduce the following tensors:
    \begin{equation}        
        \begin{aligned}
            X_{ab}{}^{c}& = \Theta_{a}{}^{\alpha} t_{\alpha}{}_{b}{}^{c}~,~~~&Y_{a}{}_{\alpha}{}^{\beta}& =  \Theta_{a}{}^{\gamma} f_{\gamma\alpha}{}^{\beta}+ t_{\alpha}{}_{a}{}^{b}\Theta_{b}{}^{\beta}\equiv\delta_\alpha \Theta_{b}{}^{\beta}~,
            \\
            X_{a\alpha}{}^{\beta}& =\Theta_{a}{}^{\gamma} f_{\gamma\alpha}{}^{\beta}~,~~~
            &Z^{ab}& = \Theta_{c}{}^{\alpha} t_{\alpha}{}_{d}{}^{a}d^{bcd}= X_{cd}{}^{a}d^{bcd}=Z^{[ab]}~.
        \end{aligned}
    \end{equation}
    The above tensors satisfy the following identities~\cite{deWit:2004nw}:
    \begin{subequations}\label{eq:def_tens_identities}
        \begin{equation}
            d_{a cd}d^{b cd} = \delta_{a}{}^{b}~,~~~
            X_{(ab)}{}^{c} = d_{ab d}Z^{cd}~,~~~X_{[ab]}{}^{c} = 10 d_{adf}d_{beg}d^{c de}Z^{fg}~,
        \end{equation}
        and in addition, we have the following three equivalent forms of the closure constraints:
        \begin{equation}        
            2X_{[a|c|}{}^{d}X_{b]d}{}^{e}+X_{[ab]}{}^{d}X_{dc}{}^{e}=0~,~~~
            Z^{ab}X_{bc}{}^{d}=0~,~~~
            X_{dc}{}^{[a}Z^{b]c}=0~.
        \end{equation}
    \end{subequations}
    Using these, we can now apply the formalism of \ref{ssec:gen_tensor_hierarchies} and construct the differential graded Lie algebra. It helps to broaden the perspective a bit and derive the latter from a graded Lie algebra $V$, with underlying vector space consisting of $\mathfrak{e}_{6(6)}$-modules:
    \begin{equation}    
        \begin{array}{ccccccccccccccccccccc}
            V_{\mathfrak{e}_{6(6)}}& =& V_{-5}&\oplus& V_{-4} &\oplus& V_{-3}&\oplus& V_{-2}&\oplus& V_{-1}&\oplus& V_{0} &\oplus& V_{1}\\[5pt]
            \rho_{\sst{(k)}}&& \mathbf{27\oplus 1728}&& \mathbf{351}_c && \mathbf{78} && \mathbf{27} && \mathbf{27}_c && \mathbf{78} && \mathbf{351} \\[5pt]
            e_{\sst{(k)}}&& (e^{a}, e_{ab}{}^{\alpha}) && e_{a}{}^{\alpha} && e^{\alpha} &&e^{a} && e_a &&  e_\alpha  &&  e_{\alpha}{}^{a}
        \end{array}
    \end{equation}
    We have indicated the $\mathfrak{e}_{6(6)}$-representations $\rho_{\sst{(k)}}$ carried by each $V_{\mathfrak{e}_{6(6)}}$-degree~$k$ summand, $V_k$, and their corresponding basis elements $e_{\sst{(k)}}$, e.g.~$(e_\0)_\alpha=e_\alpha$, where $(e_\alpha)$ is some basis for the exceptional Lie algebra $\mathfrak{e}_{6(6)}$. Note that the embedding tensor $\Theta=\Theta_{a}{}^{\alpha}e_{\alpha}{}^{a}$ is an element of $V_1$ and $e_{\alpha}{}^{a}= P_{\mathbf{351}}e_{\alpha}\otimes e^{a}$.  
    
    The graded Lie bracket on $V$ is now given mostly by the obvious projections of the graded tensor products,
    \begin{equation}\label{D5_gla_com}
        [e_\alpha,e_\beta]=f_{\alpha\beta}{}^\gamma e_\gamma~,~~~
        [e_\alpha, e_{\sst{(k)}}] = \rho_{\sst{(k)}}(e_\alpha) e_{\sst{(k)}}~,~~~
        [e_{\sst{(k)}}, e_{\sst{(l)}}] = T_{k,l}e_{\sst{(k+l)}}~.
    \end{equation}
    Here, $T_{k,l}$ are the intertwiners dual to the projectors  $P_{k,l}:V_k\wedge V_l\rightarrow V_{k+l}$. For example, 
    \begin{subequations}\label{eq:e6_gla_comm_def}
        \begin{equation}
            [e_a, e_b] = 2d_{abc}e^c~,~~~[e_a, e^b] = (t_{\alpha})_{a}{}^{b} e^\alpha~,
        \end{equation}
        where
        \begin{equation}
            e^a \coloneqq \tfrac12 d^{abc}[e_b, e_c], \qquad  e^\alpha \coloneqq (t_{\alpha})_{b}{}^{a} [e_a, e^b]~. 
        \end{equation}
    \end{subequations}
    The adjoint indices are raised/lowered with $\eta_{\alpha\beta}= \tr (t_\alpha t_\beta)$, which is proportional to the Cartan--Killing form. 
    
    Selecting an element $\Theta=\Theta_{a}{}^{\alpha}e_{\alpha}{}^{a}\in V_1$ now defines a differential 
    \begin{equation}    
        \rmd v \coloneqq [\Theta, v] 
    \end{equation}
    for $v\in V$, and we note that $[\Theta,\Theta]=0$ for degree reasons. The explicit action of $\rmd e_{\sst{(k)}} \coloneqq [\Theta, e_{\sst{(k)}} ] $ can be determined using the graded Jacobi identity  from the initial condition
    \begin{equation}    
        [ e_{\alpha}{}^{a} , e_b] = P_{\rep{351}}\delta_{b}{}^{a}e_\alpha~, 
    \end{equation}
    where $P_{\rep{351}}$ is the projector $P_{\mathbf{351}}: \rep{78}\otimes \rep{27}_c \rightarrow \rep{351}$,
    \begin{equation}    
        (P_{\rep{351}})_{\alpha}{}^{a}{}_{b}{}^{\beta} = -\tfrac65 t_{\alpha}{}_{b}{}^{c}t^{\beta}{}_{c}{}^{a}+\tfrac{3}{10} t_{\alpha}{}_{c}{}^{a}t^{\beta}{}_{b}{}^{c}+\tfrac15 \delta_{\alpha}{}^{\beta}\delta^{a}{}_{b}~.
    \end{equation}
    We thus obtain a differential graded Lie algebra, and this is a special case of the algebra called dgLie (THA$'$) in \ref{ssec:tensor_hierarchy_algebras}.
    
    Let us now construct the $\ophLie_2$-algebra $\frE$ of this differential graded Lie algebra using \ref{thm:ophLie_from_dgLA}. We arrive at the graded vector space 
    \begin{equation}    
        \begin{array}{ccccccccccccccccccccc}
            \frE_{\mathfrak{e}_{6(6)}}=&  \frE_{-4} &\oplus& \frE_{-3}&\oplus& \frE_{-2}&\oplus& \frE_{-1}&\oplus& \frE_{0} \\
            & \mathbf{27\oplus 1728} && \mathbf{351}_c && \mathbf{78} && \mathbf{27} && \mathbf{27}_c \\
            & (e^{a}, e_{ab}{}^{\alpha}) && e_{a}{}^{\alpha} &&e^{\alpha} &&  e^{a}&&  e_a
        \end{array}
    \end{equation}
    with non-trivial products
    \begin{equation}
        \eps_1(x)\coloneqq [\Theta, x]
        ~,~~~
        \eps^0_2(x,y)\coloneqq     [[\Theta, x] , y]
        ~,~~~ 
        \eps^{1}_{2}(x,y)\coloneqq     (-1)^{|x|}[x , y]~.
    \end{equation}
    Explicitly, we have the differentials
    \begin{subequations}
        \begin{equation}
            \begin{split}
                \eps_1(e^a)&=  \Theta_{b}{}^{\alpha} t_{\alpha}{}_{c}{}^{d}d^{bca}e_d = X_{bc}{}^{d}d^{bca}e_d=-Z^{ad}e_d~,
                \\
                \eps_1(e^\alpha)&=   \Theta_{b}{}^{\alpha}e^{b}~,\\
                \eps_1(e_{a}{}^{\alpha})&=  - \delta_{\beta}\Theta_{a}{}^{\alpha}e^{\beta}=-Y_{a}{}_{\beta}{}^{\alpha}e^{\beta}~,         
            \end{split}
        \end{equation}
        the Leibniz-like products
        \begin{equation}
            \begin{aligned}
                \eps^0_2(e_a, e_b)&=    [\Theta_{a}{}^{\alpha} e_\alpha  , e_b] =\Theta_{a}{}^{\alpha} t_{\alpha}{}_{b}{}^{c}e_{c}  =X_{ab}{}^{c}e_{c}~,
                \\
                \eps^0_2(e_a, e^b)& =    [\Theta_{a}{}^{\alpha} e_\alpha  , e^b] =-\Theta_{a}{}^{\alpha} t_{\alpha}{}_{c}{}^{b}e^{c}  =-X_{ac}{}^{b}e^{c}~,
                \\
                \eps^0_2(e_a, e^\beta)& =    [\Theta_{a}{}^{\alpha} e_\alpha  , e^\beta] =-\Theta_{a}{}^{\alpha}f_{\alpha\gamma}{}^\beta e^\gamma =-X_{a\gamma}{}^\beta e^\gamma~,
                \\
                \eps^0_2(e_a, e_{b}{}^{\beta})& =    [\Theta_{a}{}^{\alpha} e_\alpha  , e_{b}{}^{\beta}] =-\Theta_{a}{}^{\alpha}f_{\alpha\gamma}{}^\beta e_{b}{}^{\gamma}+\Theta_{a}{}^{\alpha}t_{\alpha b}{}^{c} e_{c}{}^{\beta} =-X_{a\gamma}{}^\alpha   e_{b}{}^{\gamma}+X_{ab}{}^{c}e_{c}{}^{\beta}~,
            \end{aligned}
        \end{equation}
        as well as the alternator-type products
        \begin{equation}
            \begin{aligned}
                \eps^{1}_{2}(e_a, e_b)&=  2 d_{abc}e^{c}~,
                \\
                \eps^{1}_{2}(e_a, e^b)&=  t_{\alpha a}{}^{b}e^{\alpha}   =\eps^{1}_{2}(e^b, e_a)~,
                \\
                \eps^{1}_{2}(e_a, e^\alpha)&= e_{a}{}^{\alpha}  = P_{\rep{351}_c}e_{b}{}^{\beta} =-\eps^{1}_{2}(e^\alpha, e_a)~,
                \\
                \eps^{1}_{2}(e^a, e^b)&= -e^{[ab]}  = t_{\alpha c}{}^{[a}d^{b]cd}e_{d}{}^{\alpha}~,
            \end{aligned}
        \end{equation}
        where we used that  $t_{\alpha c}{}^{[a}d^{b]cd}$ is the intertwiner between the $\mathbf{351 }_c\in \mathbf{27 }_c \otimes \mathbf{78}$ and $\mathbf{351 }_c\cong  \bigwedge^2\,\mathbf{27}$. 
    \end{subequations}
    
    We can now construct the corresponding curvatures. We start from the Chevalley--Eilenberg algebra of $\frE_{\mathfrak{e}_{6(6)}}$ with the following generators $(r^A)$ spanning $\frE_{\mathfrak{e}_{6(6)}}[1]^*$:
    \begin{equation}    
        \begin{array}{ccccccccccccccccccccc}
            \text{degree} & 1&& 2&&3&&4&&5\\[5pt]
            \frE_{\mathfrak{e}_{6(6)}}[1]^*= & \mathbf{27} &\oplus& \mathbf{27 }_c &\oplus& \mathbf{78} &\oplus& \mathbf{351} &\oplus& \mathbf{27 }_c\oplus\rep{1728 }_c \\[5pt]
            & r^a &&r_{a} &&r_{\alpha} &&  r_{\alpha}{}^{a} &&  (r_{a}, r_{\alpha}{}^{ab})
        \end{array}
    \end{equation}
    Consider now the Weil algebra $\sfW(\frE_{\mathfrak{e}_{6(6)}})$, cf.~\ref{def:Weil-EL-infty}. Here, we introduce a second copy of shifted generators $(\hat r^A)$ spanning $\frE_{\mathfrak{e}_{6(6)}}[2]^*$ with $|\hat r^A|=|r^A|+1$. The usual Weil differential up to degree~$3$ elements, dual to scalars in $d=5$, then reads as 
    \begin{equation}\label{N=8D5QW}
        \begin{aligned}
            Q_{\sf W} r^{a} &=-Z^{ab}r_{b} - X_{bc}{}^{a}    r^{b} \oslash_0 r^{c}+\hat r^a~,
            \\
            Q_{\sf W} r_{a} &= 
            \Theta_{a}{}^{\alpha}r_{\alpha}
            +X_{ba}{}^{c}r^b \oslash_0 r_c 
            - d_{abc} r^{b} \mathbin{\hat \oslash_1}   r^{c}+ \hat r_a~,
            \\
            Q_{\sf W}r_{\alpha} &=  
            Y_{a\alpha}{}^{\beta}r_{\beta}{}^{a}
            +X_{a\alpha}{}^{\beta}r^{a} \oslash_0 r_{\beta}
            + t_{\alpha}{}_{a}{}^{b}  r^{a} \mathbin{\hat \oslash_1} r_{b} 
            +  \hat r_\alpha~,
            \\ 
            Q_{\sf W} \hat r^{a} &=
            Z^{ab} \hat r_{b} + X_{bc}{}^{a}  \hat  r^{b} \oslash_0 r^{c}- X_{bc}{}^{a}    r^{b} \oslash_0 \hat r^{c}~,
            \\
            Q_{\sf W} \hat r_{a} &= 
            - \Theta_{a}{}^{\alpha}\hat r_{\alpha}
            - X_{ba}{}^{c}\hat r^b \oslash_0 r_c 
            + X_{ba}{}^{c} r^b \oslash_0 \hat r_c 
            +2 d_{abc} \hat r    ^{b} \mathbin{\hat \oslash_1} r    ^{c}~,
            \\
            Q_{\sf W} \hat r_{\alpha} &=-Y_{a\alpha}{}^{\beta}\hat r_{\beta}{}^{a}
            -X_{a\alpha}{}^{\beta} \hat r^{a} \oslash_0 r_{\beta}
            + X_{a\alpha}{}^{\beta}  r^{a} \oslash_0 \hat r_{\beta}
            - t_{\alpha}{}_{a}{}^{b} \hat r^{a} \mathbin{\hat \oslash_1} r_{b} 
            + t_{\alpha}{}_{a}{}^{b}  \hat r_{b} \mathbin{\hat \oslash_1} r^{a}~,
        \end{aligned}
    \end{equation}
    where we have introduced the notation
    \begin{equation}
        \begin{split}
            a \mathbin{\hat \oslash_i} b& =  a  \oslash_i b +(-1)^{i+|a|\,|b|} b  \oslash_i a~, \\
            a \mathbin{\check \oslash_i} b& = a  \oslash_i b -(-1)^{i+|a|\,|b|} b  \oslash_i a~.
        \end{split}
    \end{equation} 
    The deformed Leibniz rule~\eqref{eq:def_Leibniz}, together with the remaining $\opEilh_2$-relations~\eqref{eq:Eilh-relations} and the identities~\eqref{eq:def_tens_identities}, then imply $Q_\sfW^2=0$ as one can check by direct computation. 
    
    In order to define the curvatures of the tensor hierarchy, we symmetrize to an $L_\infty$-algebra using \ref{thm:antisym_hLie}. We can then use the formalism of \ref{ssec:firmly_adjusted} to construct an adjusted Weil algebra in the sense of~\cite{Saemann:2019dsl}, ensuring closure of the gauge algebra without any further constraints on the field strengths.
    
    To illustrate in more detail the procedure and what it achieves, we can perform the coordinate change already at the level of the Weil algebra of the $\ophLie_2$-algebra. This coordinate change yields a symmetrized and firmly adjusted Weil algebra through an  evident  coordinate change, $r^A\mapsto \tilde r^A$, which removes all appearances of $\oslash_1$ in $Q_{\sf W}\tilde r^A$ via the deformed Leibniz rule~\eqref{eq:def_Leibniz}. Hence,  by \ref{thm:el_infty_contain_l_infty} we are left with an $L_\infty$-algebra. Explicitly, the following  coordinate change manifestly removes all appearances of $\oslash_1$:
    \begin{equation}    
        \begin{aligned}\label{eq:can_rotation}
            r^a&\mapsto a^a \coloneqq  r^a~,
            \\
            r_a&\mapsto b_a \coloneqq  r_a  +\tfrac12  d_{abc} r^b \mathbin{\check \oslash_0} r^c~,
            \\
            r_{\alpha}&\mapsto c_{\alpha}\coloneqq  r_{\alpha}- \tfrac12   t_{\alpha}{}_{a}{}^{b}  r^{a} \mathbin{\check  \oslash_0} r_{b}~,
            \\
            r_{\alpha}{}^{a}&\mapsto d_{\alpha}{}^{a} \coloneqq  r_{\alpha}{}^{a}+\tfrac12 P_{\rep{351}_c} r^{b} \mathbin{\check  \oslash_0} r_{\beta}+\tfrac12 t_{\alpha c}{}^{[b}d^{c]ad}r_{b} \mathbin{\check  \oslash_0} r_{c}~,
        \end{aligned}
    \end{equation}
    where $ d_{\alpha}{}^{a}$ is included as it is needed for $Q_{\sf W}\tilde r_{\alpha}$. The corresponding coordinate change on $\hat r^A$ is firmly adjusted by simply first ordering the occurrences of  $\hat r^B$ in  $\hat{\tilde{r}}^A$ to the left (which is permitted by the appearance of only $\check \oslash_0$ in $\hat{\tilde{r}}^A$) and then  sending $\check \oslash_i$ to $\check \oslash_i+\hat \oslash_i=2\oslash_i$. The choice of left ordering follows from the choice of left Leibniz rule, which is a matter of convention. Applied to~\eqref{eq:can_rotation} this yields
    \begin{equation}    
        \begin{aligned}
            \hat   r^a&\mapsto f^a \coloneqq  \hat r^a~,
            \\
            \hat    r_a&\mapsto h_a \coloneqq \hat  r_a  + 2d_{abc}\hat r^b  \oslash_0  r^c~,
            \\
            \hat    r_{\alpha}&\mapsto g_{\alpha} \coloneqq  \hat r_{\alpha}-  t_{\alpha}{}_{a}{}^{b} (\hat  r^{a}   \oslash_0  r_{b} -  \hat r_{b}  \oslash_0  r^{a})~, 
            \\
            \hat  r_{\alpha}{}^{a}&\mapsto k_{\alpha}{}^{a} \coloneqq  \hat r_{\alpha}{}^{a}+ P_{\rep{351}_c} (\hat r^{b}   \oslash_0  r_{\beta} +\hat r_{\beta}  \oslash_0  r^{b})+2 t_{\alpha c}{}^{[b}d^{c]ad} \hat r_{b}   \oslash_0  r_{c}~.
        \end{aligned}
    \end{equation}
    Note that this is a special case of the transformation~\eqref{eq:firm_rotation} for a firm adjustment.

    The result of this coordinate change is the differential graded commutative algebra $\sfW_{\rm adj}(\frE_{\mathfrak{e}_{6(6)}})$ generated by $\frE_{\mathfrak{e}_{6(6)}}[1]^*\oplus\frE_{\mathfrak{e}_{6(6)}}[2]^*$ and differential
    \begin{equation}\label{symN=8D5QW}
        \begin{aligned}
            Q_{\sfW_{\rm adj}} a^{a} &=-Z^{ab}b_{b} - \tfrac12 X_{bc}{}^{a}    a^{b}   a^{c}+f^a~,
            \\
            Q_{\sfW_{\rm adj}} b_{a} &= 
            \Theta_{a}{}^{\alpha}c_{\alpha}
            + \tfrac12 X_{ba}{}^{c}a^b   b_c + \tfrac16 d_{abc} X_{de}{}^{b}a^c   a^d   a^e 
            - d_{abc} f^{b}   a^{c}+ h_a~,
            \\
            Q_{\sfW_{\rm adj}}c_{\alpha} &=  
            Y_{a\alpha}{}^{\beta}d_{\beta}{}^{a}
            +\tfrac12 X_{a\alpha}{}^{\beta}a^{a}   c_{\beta}
            +(\tfrac14 X_{a\alpha}{}^{\beta}t_{\beta b}{}^{c}+\tfrac13t_{\alpha a}{}^{d} X_{(db)}{}^{c})a^a  a^b b_c
            \\ 
            &~~~~~+\tfrac12 t_{\alpha}{}_{a}{}^{b} f^a   b_b -\tfrac12 t_{\alpha}{}_{a}{}^{b} h_b a^a    -\tfrac{1}{6} t_{\alpha}{}_{a}{}^{b} d_{bcd}a^{a}  a^{c}  f^{d}+ g_\alpha~,\\
            Q_{\sfW_{\rm adj}} f^{a} &=
            Z^{ab} h_{b} + X_{bc}{}^{a} a^{b}    f^{c} ~,
            \\
            Q_{\sfW_{\rm adj}} h_{a} &= 
            -\Theta_{a}{}^{\alpha}g_{\alpha}
            + X_{ab}{}^{c} a^b   h_c 
            + d_{abc}f^b   f^c~,
            \\
            Q_{\sfW_{\rm adj}} g_{\alpha} &=-Y_{a\alpha}{}^{\beta}k_{\beta}{}^{a}
            +X_{a\alpha}{}^{\beta}  a^{a}   g_{\beta}
            - t_{\alpha}{}_{a}{}^{b}   h_{b}    f^{a}~.
        \end{aligned}
    \end{equation}
    
    We can now define the corresponding curvatures in the adjusted higher gauge theory as usual as a morphism of differential graded algebras
    \begin{equation}\label{eq:dga_morphism_5d}
        (\caA,\caF): \sfW_{\rm adj}(\frE_{\mathfrak{e}_{6(6)}})~\longrightarrow~\Omega^\bullet(M)~,
    \end{equation}
    where\footnote{The additional signs here follow from the choice of sign convention in~\eqref{eq:e6_gla_comm_def}.}
    \begin{equation}
        \begin{split}
            (a^a,b_a,c_\alpha,  d_{\alpha}{}^{a}) &\mapsto  (A^a, B_a, -C_\alpha,  -D_{\alpha}{}^{a})~,\\
            (f^a,h_a,g_\alpha,  k_{\alpha}{}^{a}) &\mapsto  (F^a, H_a, -G_\alpha, -K_{\alpha}{}^{a})~.
        \end{split}
    \end{equation}
    This indeed yields  the gauge potentials and curvatures of the $d=5$ tensor hierarchy:
    \begin{subequations}\label{eq:field_strengths}
        \begin{align} 
            F^a &= \rmd A^a +\tfrac{1}{2} X_{bc}{}^{a} A^b\wedge A^c  + Z^{ab}B_b~,
            \\
            H_a 
            &= \rmd B_a -\tfrac 12 X_{ba}{}^{c} A^b\wedge B_c   - \tfrac{1}{6} d_{abc} X_{de}{}^{b} A^c\wedge A^d \wedge A^e  +d_{abc} A^b\wedge F^c + \Theta_{a}{}^{\alpha}  C_\alpha~,
            \\
            G_\alpha  &= \rmd  C_\alpha-\tfrac12 X_{a\alpha}{}^{\beta}A^a\wedge  C_\gamma+(\tfrac14 X_{a\alpha}{}^{\beta}t_{\beta b}{}^{c}+\tfrac13t_{\alpha a}{}^{d} X_{(db)}{}^{c})A^a\wedge  A^b \wedge B_c \\
            &\phantom{=} +\tfrac12 t_{\alpha}{}_{a}{}^{b} F^a \wedge B_b -\tfrac12 t_{\alpha}{}_{a}{}^{b}H_b  \wedge A^a  -\tfrac{1}{6} t_{\alpha}{}_{a}{}^{b} d_{bcd}A^{a}\wedge A^{c}\wedge F^{d}- Y_{a\alpha}{}^{\beta} D_{\beta}{}^{a}~,
        \end{align}
    \end{subequations}
    along with the corresponding Bianchi identities,
    \begin{subequations}\label{D5Bianchi}
        \begin{align} 
            0&=\rmd F^{a} 
            - X_{bc}{}^{a} A^{b} \wedge   F^{c} -Z^{ab} H_{b}  ~,
            \\
            0 &= \rmd H_{a}
            - X_{ab}{}^{c} A^b  \wedge H_c 
            - d_{abc}F^b \wedge  F^c-\Theta_{a}{}^{\alpha}G_{\alpha}~,
            \\
            0&=  \rmd G_{\alpha}- X_{a\alpha}{}^{\beta}  A^{a} \wedge  G_{\beta}- t_{\alpha}{}_{a}{}^{b}   H_{b} \wedge   F^{a} +Y_{a\alpha}{}^{\beta}K_{\beta}{}^{a}~.
        \end{align}
    \end{subequations}
    We note that the full kinematical data is determined in this way: the Bianchi identities are implied by compatibility of the morphism~\eqref{eq:dga_morphism_5d} with the differential, and the gauge transformations are constructed as infinitesimal partially flat homotopies, cf.~e.g.~\cite{Saemann:2019dsl} for details.
    
    To make contact with the expressions in the supergravity literature, cf.~\cite{deWit:2004nw, Hartong:2009vc}, one must make the field redefinitions 
    \begin{equation}
        \begin{aligned}
            C_{\alpha}&\mapsto C_\alpha+\tfrac12t_{\alpha a}{}^b A^a\wedge B_b~,
            \\
            D_{\alpha}{}^{a} &\mapsto D_{\alpha}{}^{a} - \tfrac12 P_{\rep{351 }_c}  A^a\wedge  C_\alpha~.
        \end{aligned}
    \end{equation}
    Similar field redefinitions were also used in~\cite{Greitz:2013pua} to link another elegant derivation of the curvature forms (in which, however, the link to higher gauge algebras also is somewhat obscured) to the supergravity literature. We stress that from the higher gauge algebra point of view, the form~\eqref{eq:field_strengths} is special in the sense that all exterior derivatives of gauge potentials in non-linear terms have been absorbed in field strengths. This makes~\eqref{eq:field_strengths} particularly useful, as it exposes cleanly the separation of unadjusted curvature and adjustment. From the former, one can straightforwardly identify the higher Lie algebra of the structure group of the underlying higher principal bundle. Moreover, gauge transformations are readily derived from partially flat homotopies, as mentioned above. As a side effect, it is interesting to note that the arising higher products are at most ternary.
    
    An interesting aspect of~\eqref{eq:field_strengths} is the fact that the covariantizations of the differentials $\rmd B$ and $\rmd C$ contain a perhaps unexpected factor of $\tfrac12$. This factor is a clear indication that
    the origin of the gauge $L_\infty$-algebra is indeed an $\ophLie_2$-algebra: the action $\acton$ of $A$ on $B$ and $C$ is encoded in an $\ophLie_2$-algebra with
    \begin{equation}
        \eps_2^0(A,B)\coloneqq A\acton B\eand \eps_2^0(A,C)\coloneqq A\acton C~,
    \end{equation}
    which is then antisymmetrized by \ref{thm:antisym_hLie} to 
    \begin{equation}
        \mu_2(A,B)\coloneqq \tfrac12 \eps^0_2(A,B)\eand \mu_2(A,C)\coloneqq \tfrac12 \eps^0_2(A,C)~,
    \end{equation}
    at the cost of introducing non-trivial higher products $\mu_3$, cf.~\eqref{eq:antisymmetrization_hLie}. This is fully analogous to the situation in generalized geometry, cf.~e.g.~the Dorfman and Courant brackets~\eqref{eq:Dorfman} and~\eqref{eq:ass_Courant_algebra}.
    
    \section{Comparison to the literature}\label{sec:comparison}
    
    We conclude by comparing our results with algebraic structures previously introduced in the literature to capture the gauge structure underlying the higher gauge theories obtained in the tensor hierarchies of gauged supergravity. We shall focus on the particularities of the gauge algebraic structures of the tensor hierarchies; for other work linking the tensor hierarchy to ordinary $L_\infty$-algebras, see also~\cite{Cagnacci:2018buk}.
    
    \subsection{Enhanced Leibniz algebras}
    
    A notion of enhanced Leibniz algebras was introduced in~\cite{Strobl:2016aph,Strobl:2019hha} to capture the parts of the higher gauge algebraic structures appearing in the tensor hierarchy. See also~\cite{Kotov:2018vcz} for a discussion of the higher gauge theory employing these enhanced Leibniz algebras and the link to the tensor hierarchy.
    
    \begin{definition}[\cite{Strobl:2019hha}]
        An \uline{enhanced Leibniz algebra} is a Leibniz algebra $(\sfV,[-,-])$ together with a vector space $\sfW$ and a linear map $t:\sfW\rightarrow \sfV$ as well as a binary operation $\circ:\sfV\otimes\sfV\to\sfW$ such that
        \begin{equation}
            \begin{aligned}
                [t(w),v]&=0~~~&u\stackrel{s}{\circ}[v,v]&=v\stackrel{s}{\circ}[u,v]
                \\
                t(w)\circ t(w)&=0~,~~~&[v,v]&=t(v\circ v)
            \end{aligned}
        \end{equation}
        for all $u,v\in \sfV$ and $w\in \sfW$, where $u\stackrel{s}{\circ}v$ denotes the symmetric part of $u\circ v$.
        
        A \uline{symmetric enhanced Leibniz algebra} additionally satisfies the condition that
        \begin{equation}
            u\circ v=v\circ u
        \end{equation}
        for all $u,v\in \sfV$.
    \end{definition}
    
    A symmetric enhanced Leibniz algebra is an $\ophLie_2$-algebra concentrated in degrees $-1$ and $0$ with a few axioms missing. We can identify the structure maps as follows.
    \begin{equation}
        \begin{gathered}
            \frE=(\frE_{-1}\xrightarrow{~\eps_1~}\frE_0)~=~(\sfW\xrightarrow{~\sft~}\sfV)~,
            \\
            \eps_2(v_1,v_2)=[v_1,v_2]~,~~~\eps_2(v,w)=0~,~~~\sfalt(v_1,v_2)=v_1\circ v_2~,
        \end{gathered}
    \end{equation}
    for $v,v_1,v_2\in \sfV$ and $w\in\sfW$. The $\ophLie_2$-algebra relations~\eqref{eq:hLie-relations} are trivially satisfied since $\eps_2$ is a Leibniz bracket. Moreover, $\eps_1$ is trivially a differential and a derivation of $\eps_2$. The relation $\eps_2(v_1,v_2)+\eps_2(v_2,v_1)=\eps_1(\sfalt(v_1,v_2))$ is the polarization of $[v,v]=t(v\circ v)$. The relation $u\stackrel{s}{\circ}[v,v]=v\stackrel{s}{\circ}[u,v]$ fails to accurately reproduce the relation between $\eps_2$ and the alternator, $\sfalt(v_1,\eps_2(v_2,v_3))=\sfalt(\eps_2(v_2,v_3),v_1)$. Moreover, the relation $t(w)\circ t(w)=0$ fails to reproduce the appropriate relation for the alternator, $\sfalt(v_1,\sft(w_1))=\sfalt(\sft(w_1),v_1)=0$.
    
    The original definition in~\cite{Strobl:2019hha} of a (not necessarily symmetric) ``enhanced Leibniz algebra'' is slightly more general, allowing for the operation $\circ$ to be not symmetric. However, this is not very natural, as discussed in \cref{ssec:EL_infty_and_L_infty}. Moreover, the algebraic structure underlying the tensor hierarchy is an $\ophLie_2$-algebra, so enhanced Leibniz algebras require axiomatic completion.
    
    \subsection{\texorpdfstring{$\infty$}{Infinity}-Enhanced Leibniz algebras}
    
    A similar notion of extended Leibniz algebras was formulated in~\cite{Bonezzi:2019ygf}, see also~\cite{Bonezzi:2019bek} as well as the previous work on Leibniz algebra gauge theories~\cite{Hohm:2018ybo}.
    \begin{definition}[\cite{Bonezzi:2019ygf}]
        An \uline{\(\infty\)-enhanced Leibniz algebra} is an $\IN$-graded differential complex $(X=\oplus_{i\in\IN} X_i,\rmd)$ with differential of degree~$-1$, endowed with two binary operations 
        \begin{subequations}
            \begin{equation}
                \begin{aligned}
                    \circ&:X_0\otimes X_0\rightarrow X_0~,
                    \\
                    \bullet&:X_i\otimes X_j\rightarrow X_{i+j+1}~,
                \end{aligned}
            \end{equation}
            satisfying the following relations:
            \begin{align}
                (x\circ y)\circ z&=x \circ(y\circ z)-y\circ (x\circ z)~,\label{eq:eL_Leibniz}
                \\
                a\bullet b&=(-1)^{|a|\,|b|}(b\bullet a)~,\label{eq:eL_sym_bullet}
                \\
                (\rmd w)\circ x &= 0~, \label{eq:ielax1}
                \\ 
                \rmd(x\bullet y)&= x\circ y+y\circ x~,\label{eq:ielax2}
                \\
                \rmd(u\bullet v) &= - (\rmd u)\bullet v + (-1)^{|u|+1}u\bullet\rmd v~, \label{eq:ielax5} 
                \\
                (a\bullet b)\bullet c&= 
                (-1)^{|a|+1}a\bullet(b\bullet c)-(-1)^{(|a|+1)|b|} b\bullet(a\bullet c)~, \label{eq:ielax6} 
                \\ 
                \rmd(x\bullet(y\bullet z)) &= (x\circ y)\bullet z + (x\circ z)\bullet y - (y\circ z+z\circ y)\bullet x~, \label{eq:ielax3} 
                \\
                \left[\rmd(x\bullet(y\bullet z))\right]_{x\leftrightarrow y}&=\left[(x\circ y)\bullet u-2x\bullet\rmd(y\bullet u)-x\bullet(y\bullet \rmd u)
                \right]_{x\leftrightarrow y}~, \label{eq:ielax4} 
            \end{align}
        \end{subequations}
        where \(x,y,z\) range over degree~$0$ elements, \(w\) ranges over degree \(1\) elements, \(u,v\) range over positive degree elements, and \(a,b,c\) over arbitrary elements of homogeneous degrees, and where \(
        [\dotsb]_{x\leftrightarrow y}
        \)
        signifies that the enclosed expression is antisymmetrized with respect to the permutation between \(x\) and \(y\).
    \end{definition}
    
    An $\infty$-enhanced Leibniz algebra is a particular type of $\ophLie_2$-algebra with some axioms missing. Clearly, to compare the axioms, we have to invert the sign of the degree. We thus consider an $\ophLie_2$-algebra $\frE$ concentrated in non-positive degrees with $\eps_2^0=\circ$ non-trivial only on elements of degree~$0$. Moreover, we are led to identify $\eps_2^1$ with $\bullet$; all other $\eps_2^i$ are trivial. Then we have the following relations between the axioms of an $\infty$-enhanced Leibniz algebra and an $\ophLie_2$-algebra:
    \begin{itemize}[leftmargin=1.5cm]
        \item[\eqref{eq:eL_Leibniz}] is simply the Leibniz identity and follows from the quadratic relation for $\eps_2^0$.
        \item[\eqref{eq:eL_sym_bullet}] amounts to $\eps_2^1$ being graded symmetric and follows from the modified Leibniz rule, as do~\eqref{eq:ielax1}--\eqref{eq:ielax5}.
        \item[\eqref{eq:ielax6}] follows from the $\ophLie_2$-axiom for $\eps_2^1\circ \eps_2^1$.
        \item[\eqref{eq:ielax3}] follows from the modified Leibniz rule together with the $\ophLie_2$-axioms for $\eps_2^1\circ \eps_2^0$ and $\eps_2^0\circ \eps_2^1$:
        \begin{equation}
            \begin{aligned}
                \eps_1(\eps_2^1(x,\eps_2^1(y,z))) &=
                \eps_2^0(x,\eps_2^1(y,z))
                +  \eps_2^0(\eps_2^1(y,z),x)
                - \eps_2^1(x,\eps_2^0(y,z)+\eps_2^0(z,y)) \\
                &= \eps_2^1(\eps_2^0(x,y),z) + \eps_2^1(y,\eps_2^0(x,z))
                - \eps_2^1(x,\eps_2^0(y,z)+\eps_2^0(z,y))~,
            \end{aligned}
        \end{equation}
        as does~\eqref{eq:ielax4}: we have:
        \begin{equation}
            \begin{aligned}
                \eps_2^0(\eps_2^1(-,-),-)&=0~,
                \\
                \eps_1(\eps_2^1(x,\eps_2^1(y,u))) &= \eps_2^1(x,\eps_2^1(y,\eps_1(u)))
                + \eps_2^0(x,\eps_2^1(y,u)) - \eps_2^1(x,\eps_2^0(y,u)+\eps_2^0(u,y)) 
                \\
                &= \eps_2^1(x,\eps_2^1(y,\eps_1(u)))
                +\eps_2^1(\eps_2^0(x,y),u) + \eps_2^1(y,\eps_2^0(x,u))
                \\
                &\qquad - \eps_2^1(x,\eps_2^0(y,u)) -\eps_2^1(x,\eps_2^0(u,y))~,
                \\
                \eps_2^1(x,\eps_1(\eps_2^1(y,u))) &= -\eps_2^1(x,\eps_2^1(y,\eps_1(u))) + \eps_2^1(x,\eps_2^0(y,u)+\eps_2^0(u,y))~,
            \end{aligned}
        \end{equation}
        and putting this together, we obtain 
        \begin{equation}
            \begin{aligned}
                &\left[\rmd(\eps_2^1(x,\eps_2^1(y,z)))
                +2\eps_2^1(x,\rmd(\eps_2^1(y,u))
                \right]_{x\leftrightarrow y}
                \\
                &\hspace{3cm}=[
                -\eps_2^1(x,\eps_2^1(y,\rmd u))
                +\eps_2^1(\eps_2^0(x,y),u)
                +\eps_2^1(x,\eps_2^0(u,y))
                ]_{x\leftrightarrow y}~.
            \end{aligned}
        \end{equation}
    \end{itemize}
    Note, however, that while the $\ophLie_2$-algebra axioms imply the axioms of an $\infty$-enhanced Leibniz algebra, the reverse statement is not true, even for $\infty$-enhanced Leibniz algebras concentrated in degrees $0$ and $1$. The latter essentially implies that $\infty$-enhanced Leibniz algebras are an incomplete abstraction of the operad $\opLie$ and thus do not give the full picture. Altogether, we arrive at the same conclusion as for enhanced Leibniz algebras. 
    
    As a side remark, we note that in the outlook of~\cite{Bonezzi:2019ygf}, the authors mentioned the desire for the interpretation of $\infty$-enhanced Leibniz algebras as the homotopy algebras of some simpler algebraic structure. Our discussion suggests that this is not possible; instead, the axiomatic completion of $\infty$-enhanced Leibniz algebras yields $\ophLie_2$-algebras whose homotopy algebras form $E_2L_\infty$-algebras, a much weaker version of $L_\infty$-algebras.
    
    \subsection{Algebras producing the tensor hierarchies}\label{ssec:tensor_hierarchy_algebras}
    
    We now come to larger picture of algebras that lead to the gauge structures visible in the tensor hierarchies, see \ref{fig:alg_diagram}. Note that this picture has only been applied in the context of the tensor hierarchy for maximal supersymmetry. We shall be less detailed in the following.
    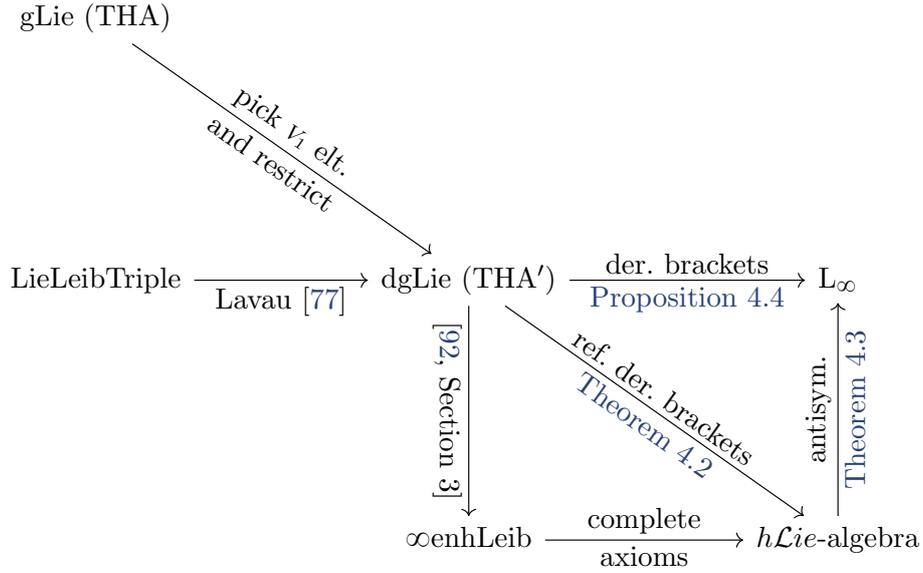
\begin{figure}[h]
        \begin{center}
            \begin{tikzcd}[row sep=2.8cm,column sep=2.3cm]
                \mbox{gLie (THA)} \arrow{rd}[sloped]{\mbox{pick $V_1$ elt.}}[sloped,swap]{\mbox{and restrict}}& 
                \\
                \mbox{LieLeibTriple} \rar[swap,"\mbox{Lavau~\cite{Lavau:2017tvi}}"] & \mbox{dgLie (THA\('\))} \arrow{r}{ \mbox{der.~brackets}}[swap]{\mbox{\Cref{prop:dgLA_to_L_infty}}}
                \arrow{d}[sloped,swap]{\mbox{\cite[Section 3]{Lavau:2019oja}}}
                \arrow{rd}[sloped]{\mbox{ref.~der.~brackets}}[swap,sloped]{\mbox{\Cref{thm:ophLie_from_dgLA}}}
                & \mbox{L\(_\infty\)}\\
                & \mbox{$\infty$enhLeib} \arrow{r}{\mbox{complete}}[swap]{\mbox{axioms}}& \mbox{$\ophLie_2$-algebra} \arrow{u}[sloped]{ \mbox{antisym.}}[swap,sloped]{\mbox{\Cref{thm:antisym_hLie}}}
            \end{tikzcd}
        \end{center}
        \caption{The relation between the various algebraic structures in the literature and how $\ophLie_2$-algebras fit into the picture.}\label{fig:alg_diagram}
    \end{figure}
    
    In~\cite{Palmkvist:2013vya}, Palmkvist constructs an infinite-dimensional $\IZ/2$-graded Lie algebra, which he calls the {\em tensor hierarchy algebra}, ``gLie (THA)'' in \ref{fig:alg_diagram}. For further work on the tensor hierarchy algebra, see also~\cite{Cederwall:2019qnw,Cederwall:2019bai,Cederwall:2021ymp}. As observed in~\cite{Greitz:2013pua}, see also~\cite{Palmkvist:2013vya}, this $\IZ/2$-grading can be naturally refined into a $\IZ$-grading, and picking an element of degree~1 and subsequent restriction induces the structure of a differential graded Lie algebra, ``dgLie (THA\('\))'' in~\ref{fig:alg_diagram}. In~\cite{Lavau:2017tvi}, Lavau called this differential graded Lie algebra the ``tensor hierarchy algebra'' (not to be confused with Palmkvist's larger graded Lie algebra), and derived it from a further algebraic structure called {\em Lie--Leibniz triples}, ``LieLeibTriple'' in~\ref{fig:alg_diagram}. This differential graded Lie algebra then naturally gives rise to $\infty$-enhanced Leibniz algebras, as described in~\cite[Section 3]{Lavau:2019oja}. As explained above, the $\infty$-enhanced Leibniz algebra were an incomplete ``guess'' of the axioms of an $\ophLie_2$-algebra with $\eps_2^i=0$ for $i\geq 2$. Thus, from our perspective, $\infty$-enhanced Leibniz algebras are appropriately replaced by these, and we then have the construction of the gauge $L_\infty$-algebra via the picture~\eqref{eq:diagram_algebras}, which is refined in \ref{fig:alg_diagram}. We note that the composition of the arrows ``complete axioms'' and ``antisym.'', which produces an $L_\infty$-algebra from an $\infty$-enhanced Leibniz algebra, is found in~\cite[Appendix B]{Bonezzi:2019ygf}. As indicated in~\ref{fig:alg_diagram}, the direct construction of an $L_\infty$-algebra from a differential graded Lie algebra is the Fiorenza--Manetti--Getzler construction~\Cref{prop:dgLA_to_L_infty}, as pointed out in~\cite{Lavau:2019oja}, where Getzler's formulas were specialized to the tensor hierarchy differential graded Lie algebra.
    
    For prior relations among tensor hierarchies, the embedding tensor formalism and (homotopy) algebras see also~\cite{Hohm:2018ybo, Kotov:2018vcz}. We again stress that from our point of view, it is not natural to consider gauge theories with infinitesimal symmetries that are not (weaker forms of) Lie algebras. Axiomatically completing the various forms of Leibniz algebras to $\ophLie_2$-algebras solves this issue.
    
    As a side remark, let us note that the fact that Leibniz algebras naturally produce $L_\infty$-algebras has been pointed out in~\cite{Lavau:2020pwa}. This is \ref{prop:Leib_is_hemistrict_ELinfty} stating that any Leibniz algebra naturally extends to an $\ophLie_2$-algebra combined with \ref{thm:antisym_hLie} antisymmetrizing this $\ophLie_2$-algebra to an $L_\infty$-algebra.

    \section*{Acknowledgments}
    
    This work was supported by the Leverhulme Research Project Grant RPG--2018--329 ``The Mathematics of M5-Branes.'' We would like to thank Alexander Schenkel, Bruno Vallette, and Martin Wolf for interesting discussions and in particular Jim Stasheff for helpful comments on an earlier version of the manuscript.
    
    \appendices
    
    \subsection{Example of an \texorpdfstring{$E_2L_\infty$}{E2L-infinity}-algebra}\label{app:involved_example}
    
    Let us give the explicit form of an $E_2L_\infty$-algebra $\frE$, in which the products $\eps_1$, $\eps_2^0$, $\eps_2^1$, and $\eps_3^{00}$ are generic while all other products are trivial. We do not impose any conditions on the underlying differential complex. The compatibility relations are readily computed in the Chevalley--Eilenberg picture, and they read as follows:
    \begin{subequations}
        \begin{equation}
            \begin{aligned}
                \eps_1(\eps_1(e_1))&=0~,
                \\
                \eps_1(\eps_2^0(e_1,e_2))&=\eps_2^0(\eps_1(e_1),e_2))+(-1)^{|e_1|}\eps_2^0(e_1,\eps_1(e_2))~,
                \\
                \eps_1(\eps_2^1(e_1,e_2))&=\eps_2^0(e_1,e_2)+(-1)^{|e_1|\,|e_2|}\eps_2^0(e_2,e_1)-\eps_2^1(\eps_1(e_1),e_2)\\
                &~~~-(-1)^{|e_1|}\eps_2^1(e_1,\eps_1(e_2))~,
                \\
                \eps_2^1(e_1,e_2)&=(-1)^{|e_1|\,|e_2|}\eps_2^1(e_2,e_1)
            \end{aligned}
        \end{equation}
        \begin{equation}
            \begin{aligned}
                \eps_1(\eps_3^{00}(e_1,e_2,e_3))&=\eps_2^0(\eps_2^0(e_1,e_2),e_3)+(-1)^{|e_1|\,|e_2|}\eps_2^0(e_2,\eps_2^0(e_1,e_3))-\eps_2^0(e_1,\eps_2^0(e_2,e_3))\\
                &~~~-\eps_3^{00}(\eps_1(e_1),e_2,e_3)-(-1)^{|e_1|}\eps_3^{00}(e_1,\eps_1(e_2),e_3)\\
                &~~~-(-1)^{|e_1|+|e_2|}\eps_3^{00}(e_1,e_2,\eps_1(e_3))~,
                \\
                \eps_2^0(\eps_2^1(e_1,e_2),e_3)&=\eps_3^{00}(e_1,e_2,e_3)+(-1)^{|e_1|\,|e_2|}\eps_3^{00}(e_2,e_1,e_3)~,
                \\
                \eps_2^0(e_1,\eps_2^1(e_2,e_3))&=-(-1)^{|e_1|}\eps_3^{00}(e_1,e_2,e_3)-(-1)^{|e_2|\,|e_3|+|e_1|}\eps_3^{00}(e_1,e_3,e_2)
                \\
                &~~~+(-1)^{|e_1|}\eps_2^1(\eps_2^0(e_1,e_2),e_3)+(-1)^{|e_1|(1+|e_2|)}\eps_2^1(e_2,\eps_2^0(e_1,e_3))~,\\
                \eps_2^1(e_1,\eps_2^1(e_2,e_3))&=-(-1)^{|e_1|}\eps_2^1(\eps_2^1(e_1,e_2),e_3)+(-1)^{(|e_1|+1)(|e_2|+1)}\eps_2^1(e_2,\eps_2^1(e_1,e_3))~,
            \end{aligned}
        \end{equation}
        \begin{equation}\label{eq:E2L:axioms_b}
            \begin{aligned}
                \eps_2^0(&e_1,\eps_3^{00}(e_2,e_3,e_4))+(-1)^{|e_1|}\eps_3^{00}(e_1,\eps_2^0(e_2,e_3),e_4)+(-1)^{|e_1|+|e_2|\,|e_3|}\eps_3^{00}(e_1,e_3,\eps_2^0(e_2,e_4))\\
                &+(-1)^{|e_1|}\eps_2^0(\eps_3^{00}(e_1,e_2,e_3),e_4)+(-1)^{(|e_1|+|e_2|+1)|e_3|+|e_1|}\eps_2^0(e_3,\eps_3^{00}(e_1,e_2,e_4))\\
                &=(-1)^{|e_1|}\eps_3^{00}(e_1,e_2,\eps_2^0(e_3,e_4))+(-1)^{|e_1|}\eps_3^{00}(\eps_2^0(e_1,e_2),e_3,e_4)\\
                &\hspace{1cm}-(-1)^{(|e_1|+1)(|e_2|+1)}\eps_2^0(e_2,\eps_3^{00}(e_1,e_3,e_4))+(-1)^{|e_1|(|e_2|+1)}\eps_3^{00}(e_2,\eps_2^0(e_1,e_3),e_4)\\
                &\hspace{1cm}+(-1)^{|e_1|(|e_2|+|e_3|+1)}\eps_3^{00}(e_2,e_3,\eps_2^0(e_1,e_4))~,
            \end{aligned}
        \end{equation}
        \begin{equation}
            \begin{aligned}
                \eps_3^{00}(\eps_2^1(e_1,e_2),e_3,e_4)&=0~,\\
                \eps_3^{00}(e_1,\eps_2^1(e_2,e_3),e_4)&=0~,\\
                \eps_2^1(\eps_3^{00}(e_1,e_2,e_3),e_4)&=-(-1)^{|e_1|+|e_2|}\eps_3^{00}(e_1,e_2,\eps_2^1(e_3,e_4))\\
                &\hspace{1cm}-(-1)^{|e_3|(|e_1|+|e_2|+1)}\eps_2^1(e_3,\eps_3^{00}(e_1,e_2,e_4))~,
            \end{aligned}
        \end{equation}
        and
        \begin{equation}
            \begin{aligned}
                0=&\eps_3^{00}(\eps_3^{00}(e_1, e_2, e_3), e_4, e_5)+
                (-1)^{|e_1| + |e_2|} \eps_3^{00}(e_1, e_2, \eps_3^{00}(e_3, e_4, e_5))\\
                &\hspace{0.5cm}- (-1)^{|e_1| + (1 + |e_2|)|e_3|} \eps_3^{00}(e_1,e_3, \eps_3^{00}(e_2, e_4, e_5))\\
                &\hspace{0.5cm} 
                + (-1)^{|e_1| + (1 + |e_2| + |e_3|) |e_4|} \eps_3^{00}(e_1,e_4, \eps_3^{00}(e_2, e_3, e_5))\\
                &\hspace{0.5cm}
                + (-1)^{|e_1|} \eps_3^{00}(e_1, \eps_3^{00}(e_2, e_3, e_4), e_5) + (-1)^{(|e_1| +1)(|e_2| + |e_3|)} \eps_3^{00}(e_2, e_3, \eps_3^{00}(e_1, e_4, e_5))\\
                &\hspace{0.5cm}- (-1)^{(1 + |e_1|) |e_2| + (1 + |e_1| + |e_3|)|e_4|}\eps_3^{00}(e_2, e_4, \eps_3^{00}(e_1, e_3, e_5))\\
                &\hspace{0.5cm}- (-1)^{|e_2| + |e_1| |e_2|} \eps_3^{00}(e_2, \eps_3^{00}(e_1, e_3, e_4), e_5)\\
                &\hspace{0.5cm}+ (-1)^{(1 + |e_1| + |e_2|) (|e_3| + |e_4|)} \eps_3^{00}(e_3, e_4, \eps_3^{00}(e_1, e_2, e_5))\\
                &\hspace{0.5cm}+ (-1)^{|e_3| + (|e_1| + |e_2|) |e_3|} \eps_3^{00}(e_3, \eps_3^{00}(e_1, e_2, e_4), e_5)
            \end{aligned}
        \end{equation}
    \end{subequations}
    for all $e_i\in \frE$.
    
    \bibliography{bigone}
    
    \bibliographystyle{latexeu}
    
\end{document}